\PassOptionsToPackage{expansion=false}{microtype}
\documentclass[acmsmall, screen]{acmart}

\acmJournal{TOSEM}
\acmYear{2026}
\acmVolume{0}
\acmNumber{0}
\acmArticle{0}
\acmMonth{0}
\setcopyright{none}    
\acmConference[arXiv preprint]{Manuscript under review for ACM TOSEM}{}{}

\usepackage{amsmath, amsthm}
\usepackage{mathtools}
\usepackage{booktabs}
\usepackage{multirow}
\usepackage{listings}
\usepackage{xcolor}
\usepackage[breakable]{tcolorbox}
\usepackage{enumitem}
\usepackage{tikz}
\usetikzlibrary{positioning, arrows.meta, shapes.geometric, fit, calc, backgrounds}
\newcommand{\cmark}{\textbf{Y}}
\newcommand{\xmark}{\ensuremath{-}}
\newcommand{\tildemark}{\textbf{P}}
\microtypesetup{expansion=false}

\usepackage{adjustbox}

\newtheorem{theorem}{Theorem}
\newtheorem{lemma}{Lemma}
\newtheorem{proposition}{Proposition}
\newtheorem{definition}{Definition}
\newtheorem{hypothesis}{Hypothesis}
\newtheorem{remark}{Remark}
\newtheorem*{conjecture*}{Conjecture}
\newtheorem*{corollary*}{Corollary}

\lstdefinestyle{pythonstyle}{
  language=Python,
  basicstyle=\ttfamily\footnotesize,
  keywordstyle=\color{blue!70!black},
  commentstyle=\color{green!50!black}\itshape,
  stringstyle=\color{red!70!black},
  showstringspaces=false,
  breaklines=true,
  frame=single,
  rulecolor=\color{gray!30},
  numbers=left,
  numberstyle=\tiny\color{gray},
  numbersep=5pt
}

\title{NOETHER: A Constructive Framework for Metamorphic Pattern Discovery from Operator Algebras}

\newcommand{\threeaffil}{%
  \affiliation{%
    \institution{School of Computing, University of South China}
    \city{Hengyang}
    \postcode{421001}
    \country{China}
  }%
  \affiliation{%
    \institution{Hunan Engineering Research Center of Software Evaluation and Testing for Intellectual Equipment}
    \city{Hengyang}
    \postcode{421001}
    \country{China}
  }%
  \affiliation{%
    \institution{CNNC Key Laboratory on High Trusted Computing}
    \city{Hengyang}
    \postcode{421001}
    \country{China}
  }%
}

\author{Meng Li}
\authornote{Corresponding author.}
\email{mlemon@usc.edu.cn}
\threeaffil

\author{Xiaohua Yang}
\threeaffil

\author{Jie Liu}
\threeaffil

\author{Shiyu Yan}
\threeaffil

\begin{abstract}
\textbf{Context.} Metamorphic Testing is recognised in IEEE/ISO software-testing standards and increasingly recommended for AI systems, but its progress is bottlenecked by metamorphic relation (MR) identification: existing approaches (structured frameworks, mining and evolutionary pipelines, LLM-assisted methods, MetaPattern catalogues) share an inductive grounding that leaves three foundational questions open: \emph{origin}, \emph{closure}, and \emph{transferability}.

\textbf{Objective.} We propose a framework whose downstream step from program-induced operator algebra to MetaPattern set is mechanical and provable, while making the upstream curation of the algebra a stated empirical hypothesis with an explicit scope precondition.

\textbf{Method.} We introduce \textbf{NOETHER}, a two-layer framework. The upstream layer is an eight-block decomposition over recurrent mathematical structures (symmetry, order, self-adjoint, time-reversal, limit, qualitative-dynamics, method-comparison, relational equivalence), curated as an empirical hypothesis. The downstream layer is the \texttt{CONSTRUCT-MP} algorithm: given the decomposition, it produces a MetaPattern set with an algebraic-closure guarantee under the \texttt{Translate} operator (Theorem~1) and polynomial-time decidability under a finite generating set (Theorem~2). Theorem~1 converts an empirical-adequacy claim into a structural-adequacy obligation within an explicitly bounded scope. We test the framework on three operator-algebraic domains and an empirical head-to-head against an automated SOTA baseline.

\textbf{Results.} On Boltzmann reactor physics the framework systematises a prior inductive catalogue and re-classifies further equivalence classes; on equivariant ML it derives executable MRs for rotation invariance, adjoint duality, and training-trajectory reversibility; on relational query optimisers it exercises the relational-equivalence block. The central falsifiable prediction --- $\mathcal{L}^{*}$-blindness on homogeneity-preserving mutators, derivable ex-ante from the algebra and the mutator specification --- holds on the in-scope substrate. The empirical head-to-head against a GP-evolved baseline reports Set~N dominated on the scope-matched D1 stratum; the framework's contribution is read as algebraic derivability, per-block complementarity, and an out-of-scope D2-stratum boundary no inductive baseline can derive. The absolute-completeness conjecture (Theorem~$1'$) is falsified on $\mathcal{A}_{\mathrm{PWR}}$ via two independent counterexamples that identify five \texttt{Translate}-extension dimensions, with five further candidate dimensions on the equivariant-ML and relational-query algebras as the principal locus of follow-up work.

\textbf{Conclusion.} NOETHER lifts induction from per-program MR sampling to a stable per-domain algebraic layer; the downstream step is deductive and mechanical, while the upstream empirical layer is stated as an explicit hypothesis with documented out-of-scope cases (web applications, RLHF reward models, distributed-consensus protocols, compiler-internal optimisations).
\end{abstract}



\begin{document}
\maketitle

\section{Introduction}
\label{sec:intro}

Noether's first theorem replaced an empirically curated catalogue of conservation laws with a derivation from the structure of an action functional.\footnote{The theorem appeared in Noether's 1918 paper \emph{Invariante Variationsprobleme}, addressing a question Hilbert and Klein had raised about energy conservation in general relativity. We do not invoke it as a theorem about programs: program semantics and metamorphic relations do not provide an action functional. The analogy is methodological only.} A catalogue of observed invariants can sometimes be replaced by a derivation procedure grounded in the underlying structure. The replacement does not eliminate empirical work; it lifts the empirical step from per-instance enumeration to per-domain structural identification, one cycle of induction per stable domain rather than per program.

Software testing faces a related problem. Metamorphic Testing (MT), introduced by Chen et~al.\ in 1998~\cite{Chen1998}, was admitted into the IEEE/ISO/IEC software-testing standards in 2022~\cite{ISO29119} and has been increasingly endorsed as a technique for testing AI and machine-learning systems~\cite{Segura2016, LiTOSEM2025}. Its central artefact is the \emph{metamorphic relation} (MR): a property that constrains how a program's outputs must covary across multiple executions, thereby substituting for an explicit oracle in domains where one cannot exist. Despite MT's maturity, MR identification, the task of deciding \emph{which} properties hold for a program under test, remains its binding constraint. Practitioners report three recurring difficulties: high dependence on tester domain knowledge, mutually incompatible MR formulations for the same program, and low reuse of MR sets across teams or projects~\cite{Segura2016, LiTOSEM2025}. Recent work has responded at the application layer, by mining MRs for particular domains, and at the integration layer, by automating search through evolutionary, mining-based, or LLM-assisted pipelines~\cite{MRScout2024, GenMorph2024, Shin2024}. The foundational layer has not advanced at the same pace.

The principal artefact at that foundational layer is the \emph{MetaPattern} (MP): an equivalence class over MRs that captures a recurrent structural strategy a tester invokes when reasoning about program properties. MetaPatterns organise the otherwise unbounded MR design space into a small, interpretable scaffold; methods that exploit them, including the structured MR identification approaches METRIC~\cite{ChenMETRIC2016} and METRIC+~\cite{SunMETRICplus2021}, and the recent wave of LLM-prompted MR generators, depend on the scaffold's quality. Yet across the literature, MetaPattern catalogues continue to be assembled in the same way conservation laws were assembled before 1918: by induction over observed examples. A typical proposal lists $k$ patterns drawn from cluster analysis or expert codification, demonstrates that the patterns ``cover'' some corpus of MRs to a target threshold, and stops. None of the existing MP proposals, including the authors' own prior work on five reactor-physics patterns, answers the three questions that any foundational theory of MetaPatterns should answer:

\begin{enumerate}
  \item \textbf{Origin.} \emph{Why} exactly these MetaPatterns and not others? What is the structural source of an MP, as distinct from an empirical regularity in the corpus on which it was induced?
  \item \textbf{Closure.} Under what mathematical conditions is a discovered MP set \emph{closed under a stated derivation operator}, in the sense that it is guaranteed not to miss patterns reachable through the operator from a structurally fixed input? Absolute completeness over all properties one might write over the underlying structure is a strictly stronger demand and remains open in general.
  \item \textbf{Transferability.} When the program family changes (from reactor physics to natural-language inference, from numerical libraries to recommender systems), how does the MP set change, and can the new set be obtained without re-running the entire empirical induction in the new domain?
\end{enumerate}

We call this the \emph{origin--closure--transferability gap}. The gap matters because it explains why MR sets continue to grow as one-off artefacts. Without a structural source, there is no clear boundary on what the pattern space contains. Without a transfer rule, a relation written for one program rarely moves cleanly to another, even when both programs share a deep structural template. These are not only tooling limitations; they follow from grounding MetaPatterns in observed examples rather than in the structure that makes the relations hold.

This paper proposes such a structural source. We introduce \textbf{NOETHER}, a constructive framework that derives MetaPatterns from the operator-algebraic structure of the program family under test. The framework extracts MetaPatterns from invariants of an operator algebra $\mathcal{A}_P$ (formally defined in Section~\ref{sec:prelim}) that captures the program family's mathematical scaffolding. Given $\mathcal{A}_P$, NOETHER produces a MetaPattern set $\mathbb{M}(\mathcal{A}_P)$ (constructed by the algorithm of Section~\ref{sec:framework}) together with a closure guarantee over the algebra-induced MR space. Given a different program family with a different algebra $\mathcal{A}_{P'}$, the same construction produces $\mathbb{M}(\mathcal{A}_{P'})$ without re-running empirical induction. The framework does not abolish induction: domain experts must still distil $\mathcal{A}_P$ from program semantics, and Section~\ref{sec:threats-limitations} states this limitation explicitly. What becomes algebraic is the downstream step from $\mathcal{A}_P$ to the MetaPattern set. Figure~\ref{fig:noether-arch} summarises the two-layer architecture.

\begin{figure}[h]
\centering
\adjustbox{max width=\textwidth}{%
\begin{tikzpicture}[
  font=\small,
  node distance=4mm and 6mm,
  box/.style={draw, rounded corners=2pt, align=center, inner sep=4pt, minimum height=8mm, fill=white},
  thmbox/.style={draw, rounded corners=2pt, align=center, inner sep=3pt, minimum height=7mm, fill=gray!8, font=\footnotesize\itshape},
  inst/.style={draw, rounded corners=2pt, align=center, inner sep=3pt, minimum height=6mm, fill=gray!5, font=\footnotesize},
  arr/.style={-{Latex[length=2mm]}, thick},
  layerlbl/.style={font=\footnotesize\bfseries, anchor=west}
]
\node[box] (prog) {Program family $P$};
\node[box, right=12mm of prog] (alg) {Operator algebra $\mathcal{A}_P$};
\node[box, right=12mm of alg] (decomp) {Block decomposition $\mathcal{D}(\mathcal{A}_P)$\\\footnotesize $\{G, O_{\le}, T^{*}, \mathcal{T}^{*}_{\mathrm{rev}}, \mathcal{L}^{*}, \mathcal{D}^{*}, \mathcal{E}^{*}, \mathcal{B}^{*}_{\mathrm{rel}}\}$};

\draw[arr] (prog) -- node[above, font=\scriptsize, align=center] {expert\\$+$ LLM grid} (alg);
\draw[arr] (alg) -- node[above, font=\scriptsize] {Hypothesis~\ref{hyp:seven-blocks}} (decomp);

\node[box, below=20mm of alg] (mset) {MetaPattern set $\mathbb{M}(\mathcal{A}_P)$};
\node[box, right=12mm of mset] (mrset) {Algebra-induced MR space $\mathrm{MR}(\mathcal{A}_P)$};

\draw[arr] (decomp) -- ++(0,-6mm) -| node[above, near start, font=\scriptsize] {\texttt{CONSTRUCT-MP}} (mset);
\draw[arr] (mset) -- node[above, font=\scriptsize] {\texttt{Translate}} (mrset);

\node[thmbox, below=4mm of mset] (thm2) {Theorem~\ref{thm:decidable}\\poly-time};
\node[thmbox, below=4mm of mrset] (thm1) {Theorem~\ref{thm:closure}\\closure};

\draw[densely dotted] (mset) -- (thm2);
\draw[densely dotted] (mrset) -- (thm1);

\node[inst, right=10mm of mrset, anchor=west, yshift=4mm] (inst1) {Boltzmann reactor physics (\S\ref{sec:reactor})};
\node[inst, below=2mm of inst1.south west, anchor=north west] (inst2) {Equivariant ML (\S\ref{sec:cross-domain})};
\node[inst, below=2mm of inst2.south west, anchor=north west] (inst3) {Relational query optimisers (\S\ref{subsec:third-domain})};

\draw[arr, densely dashed] (mrset.east) -- ++(8mm,0) |- (inst1.west);
\draw[arr, densely dashed] (mrset.east) -- ++(8mm,0) |- (inst2.west);
\draw[arr, densely dashed] (mrset.east) -- ++(8mm,0) |- (inst3.west);

\begin{pgfonlayer}{background}
  \node[fit=(prog)(decomp), draw=gray!50, dashed, rounded corners, inner sep=6pt, label={[layerlbl, gray!70!black]above left:UPSTREAM (empirical hypothesis)}] {};
  \node[fit=(mset)(mrset)(thm2)(thm1), draw=gray!50, dashed, rounded corners, inner sep=6pt, label={[layerlbl, gray!70!black]below left:DOWNSTREAM (mechanical, provable)}] {};
\end{pgfonlayer}
\end{tikzpicture}%
}
\caption{NOETHER framework architecture (two layers). The
\emph{upstream layer} curates the program family's operator algebra
$\mathcal{A}_P$ and its eight-block decomposition as an empirical
hypothesis (Hypothesis~\ref{hyp:seven-blocks}). The \emph{downstream
layer} mechanically derives the MetaPattern set
$\mathbb{M}(\mathcal{A}_P)$ via \texttt{CONSTRUCT-MP}, with closure
under \texttt{Translate} over $\mathrm{MR}(\mathcal{A}_P)$
(Theorem~\ref{thm:closure}) and polynomial-time decidability under a
finite generating set (Theorem~\ref{thm:decidable}). The framework
is instantiated on three structurally distinct domains
(\S\ref{sec:reactor}, \S\ref{sec:cross-domain}, \S\ref{subsec:third-domain})
to test transferability at the algebra-skeleton level.}
\label{fig:noether-arch}
\end{figure}
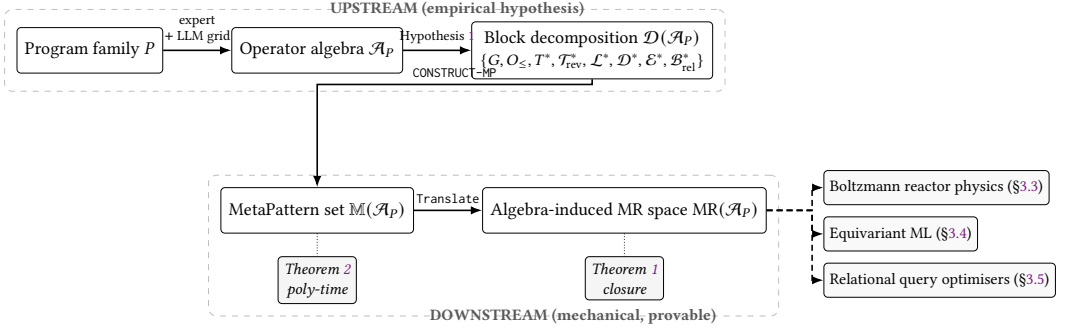

We make four contributions.

\begin{itemize}
  \item \textbf{C1.} We introduce NOETHER, a \emph{two-layer} framework for MetaPattern discovery that combines an empirically curated eight-block decomposition (Hypothesis~\ref{hyp:seven-blocks}; upstream layer, itself partly distilled from the present authors' prior 84-MR PWR catalogue, see \S\ref{subsec:reactor-mapping} for provenance) with a constructive algorithm for deriving MetaPatterns from a program-induced operator algebra (downstream layer).
  \item \textbf{C2a (positive theory).} Given the eight-block decomposition, we prove an Algebraic Closure Theorem (Theorem~1): the constructed MetaPattern set is closed under the framework's \texttt{Translate} operator over the algebra-induced MR space $\mathrm{MR}(\mathcal{A}_{P})$. Theorem~1 converts an empirical-adequacy claim (``our pattern grid covers $X\%$ of observed MRs'') into a structural-adequacy obligation (``every $\mathrm{Translate}$-reachable MR is assigned to a unique block under the canonical-block ordering of Definition~\ref{def:canonical-order}''); we acknowledge in Section~\ref{subsec:completeness} that the closure is by-construction within the explicit scope of Definition~\ref{def:alg-induced}, and document the three classes of MRs outside that scope (Remark~\ref{rem:scope}). Polynomial-time decidability of CONSTRUCT-MP holds when the algebra admits a finite generating set (Theorem~2).
  \item \textbf{C2b (negative theory).} The strictly stronger absolute-completeness conjecture (Theorem~$1'$ / Conjecture~\ref{conj:absolute}), over arbitrary properties expressible in $\mathcal{A}_P$, is \emph{false} on $\mathcal{A}_{\mathrm{PWR}}$: two MRs from the standard PWR safety-analysis literature (non-additivity of rod-bank reactivity worth; second-order mixed dependence of $k_{\mathrm{eff}}$ on moderator temperature and boron concentration) are formulable over $\mathcal{A}_{\mathrm{PWR}}$ but not in $\mathrm{MR}(\mathcal{A}_{\mathrm{PWR}})$ (\S\ref{subsec:negative-pwr}, Appendix~\ref{app:negative-proofs}). The two MRs identify five structural obstructions in \texttt{Translate}'s signature, with pairwise independence established by per-block exhaustion. Companion surveys on $\mathcal{A}_{\mathrm{equi}}$ and $\mathcal{A}_{\mathrm{rel}}$ (\S\ref{subsec:third-domain}) identify five further candidate dimensions (two on $\mathcal{A}_{\mathrm{equi}}$ specialising PWR-side dimensions; three net on $\mathcal{A}_{\mathrm{rel}}$), yielding ten \texttt{Translate}-extension dimensions across the three algebras; pairwise independence on the candidate five is asserted by inspection, with formal per-dimension exhaustion committed as follow-up. Scope statements are in Section~\ref{subsec:completeness}.
  \item \textbf{C3.} We instantiate NOETHER on a real-world program family and show how it \emph{systematises and re-classifies} previously catalogued patterns. The framework reproduces three prior MetaPatterns, refines two on a sounder algebraic basis, and includes a deflationary correction direction (revealing that the inductive grid was \emph{over-counted} on some sub-families and \emph{under-counted} by the absence of $m_{\mathrm{adj}}$ and $m_{\mathrm{rev}}$). The two structurally distinct re-classified equivalence classes are not de novo discoveries; domain experts could have written them down. NOETHER's contribution is the algebraic warrant for treating them as separate MetaPatterns within a uniform structure. The deflationary direction is non-circular relative to the prediction caveat of \S\ref{subsec:reactor-mapping}: \S\ref{subsec:pmcm-worked} demonstrates it on three independent test cases (a sorting library; Murphy et al.\ 2008's six-class ML categorisation~\cite{Murphy2008} re-decoded against a feedforward classifier algebra; the prior reactor-physics catalogue itself), without invoking $T^{*}$ or $\mathcal{T}^{*}$ structure that was curated from reactor physics.
  \item \textbf{C4.} We demonstrate \emph{structural} transferability at the algebra-skeleton level (not cross-domain empirical superiority), within the framework's scope precondition, by instantiating NOETHER on three structurally distinct operator-algebraic domains: Boltzmann reactor-physics transport, equivariant ML (Section~\ref{sec:cross-domain}), and relational query optimisers (Section~\ref{subsec:third-domain}); the last exercises the relational-equivalence block, whose algebraic skeleton differs from the Lie-group / self-adjoint / time-reversal core. A small-scale comparative case study, a DeepCrime-style real-fault pilot ($n=5$, underpowered for $\alpha=0.05$), and a cross-codebase pilot replication on Apache Commons Math ($n=3$ SUTs, $77$ mutants) report data consistent with the structural-coverage prediction within case-study scope; the case study's category-(iv) detection contrast exhibits construct validity rather than averaged superiority. A larger comparative evaluation against three SOTA representatives (GP-evolved, LLM-assisted, mining-based) is reported as a pre-registered protocol.
\end{itemize}

\paragraph{Scope of contribution.} This is a theoretical paper, and its contribution is systematisation rather than deduction from first principles. The eight blocks are curated by inspecting mathematical structures that recur across the program families we have studied; they are not derived from an algebraic axiom. The framework therefore has two layers: an \emph{upstream layer} (curating $\mathcal{A}_P$ and its block decomposition) that remains empirical and human, and a \emph{downstream layer} (mechanically deriving $\mathbb{M}(\mathcal{A}_P)$ from $\mathcal{A}_P$) that is algorithmic and provable. We do not claim to have eliminated induction from MetaPattern discovery. We move induction one level up, from ``what MetaPatterns recur in observed MR samples?'' to ``what algebraic structures recur in the program families practitioners care about?'', and make the downstream step mechanical. The engineering payoff of this re-grounding awaits empirical follow-up work; comparative evaluation against existing automated MR-identification pipelines is reported as a protocol in Section~\ref{subsec:case-study} and constitutes part of the framework's resubmission obligations.

\begin{tcolorbox}[breakable,colback=gray!5,colframe=black!50,arc=2pt,boxrule=0.5pt,fontupper=\small,title=Boundary of contribution,fonttitle=\small\bfseries]
This paper establishes:
\begin{enumerate}[leftmargin=*,nosep]
  \item Algebraic closure of $\mathbb{M}(\mathcal{A}_P)$ under \texttt{Translate} over the algebra-induced MR space $\mathrm{MR}(\mathcal{A}_P)$ for the operator algebras stated, given a block decomposition (Theorem~\ref{thm:closure}; closure is by-construction within the explicit scope of Definition~\ref{def:alg-induced}, see Remark~\ref{rem:scope}). Closure over the strictly larger space of arbitrary properties expressible in $\mathcal{A}_P$ is Theorem~$1'$, falsified on $\mathcal{A}_{\mathrm{PWR}}$ (see item (a) below);
  \item Polynomial-time decidability of CONSTRUCT-MP under explicit complexity assumptions (Theorem~\ref{thm:decidable});
  \item Three non-vacuous instantiations: Boltzmann reactor physics (\S\ref{sec:reactor}), equivariant ML (\S\ref{sec:cross-domain}), and relational query optimisers (\S\ref{subsec:third-domain}); the relational-equivalence block $\mathcal{B}^{*}_{\mathrm{rel}}$ extends the framework beyond the Lie-group / self-adjoint / time-reversal mathematical core.
\end{enumerate}

It does \emph{not} establish:
\begin{enumerate}[leftmargin=*,nosep]
  \item[(a)] Absolute completeness over arbitrary properties expressible in $\mathcal{A}_P$. This is Theorem~1$'$ (Conjecture~\ref{conj:absolute}, Appendix~\ref{app:proofs}); \S\ref{subsec:negative-pwr} establishes that it is false on the PWR core diffusion algebra $\mathcal{A}_{\mathrm{PWR}}$ via two independent counterexamples (non-additivity of rod-bank reactivity worth, and second-order mixed dependence of $k_{\mathrm{eff}}$ on moderator temperature and boron concentration), identifying five structural obstructions in \texttt{Translate}'s signature whose pairwise independence is proved by per-block exhaustion (Appendix~\ref{app:negative-proofs}). Companion surveys on $\mathcal{A}_{\mathrm{equi}}$ and $\mathcal{A}_{\mathrm{rel}}$ (\S\ref{subsec:third-domain}, companion artefacts in \texttt{theory/}) identify five further candidate dimensions on those algebras (so the falsification is not PWR-specific). Pairwise independence on the candidate five is asserted by inspection rather than by formal exhaustion proof; full per-dimension exhaustion proofs on $\mathcal{A}_{\mathrm{equi}}$ and $\mathcal{A}_{\mathrm{rel}}$ are committed as follow-up. The question of whether a Composite-\texttt{Translate} extension absorbs the ten \texttt{Translate}-extension dimensions identified across the three algebras while preserving Theorem~\ref{thm:closure}'s closure and Theorem~\ref{thm:decidable}'s polynomial-time decidability remains open.
  \item[(b)] Sufficiency of the eight-block list. Hypothesis~\ref{hyp:seven-blocks} is an empirical hypothesis with six enumerated out-of-scope program-family classes (Remark~\ref{rem:counterex}); each such class is a candidate ninth block.
  \item[(c)] Superiority over existing automated MR-identification pipelines on average. The comparative evaluation in \S\ref{subsec:case-study} establishes effects for specific defect categories on small benchmarks; it does not characterise the framework's behaviour on arbitrary defect distributions.
  \item[(d)] Elimination of induction from MetaPattern discovery. Induction is \emph{relocated} from MR-instance level to algebra-block level, not eliminated.
\end{enumerate}
\end{tcolorbox}

The remainder of the paper is organised as follows. Section~\ref{sec:related} surveys the four lines of prior work; Section~\ref{sec:noether-framework} presents the NOETHER framework, including operator-algebraic preliminaries (\S\ref{sec:prelim}), the construction algorithm and two theorems (\S\ref{sec:framework}), three structurally distinct instantiations on Boltzmann reactor physics (\S\ref{sec:reactor}), equivariant ML (\S\ref{sec:cross-domain}), and relational query optimisers (\S\ref{subsec:third-domain}), and the falsification of the absolute-completeness conjecture (\S\ref{subsec:negative-pwr}); Section~\ref{sec:empirical-evaluation} reports the empirical evaluation against five research questions including the cross-domain case study (\S\ref{subsec:case-study}), the $\mathcal{L}^{*}$-blindness prediction (\S\ref{sec:empirical-vs-sota}), and head-to-head comparisons against GenMorph (\S\ref{subsec:pooled-headtohead}) and METRIC+ (\S\ref{subsec:metricplus-relationship}); Section~\ref{sec:threats-limitations} discusses threats and limitations; Section~\ref{sec:conclusion} concludes. Appendix~\ref{app:proofs} contains the proofs of Theorems~\ref{thm:closure},~\ref{thm:decidable},~$1'$ (including the per-block exhaustion on $\mathcal{A}_{\mathrm{PWR}}$). Supplementary~S1--S9 carry illustrative material migrated for length, including NOETHER on the remaining reactor equations (S9, A), per-MR provenance (S9, B), the worked Boltzmann CONSTRUCT-MP enumeration (S9, C7), a Python reference implementation (S9, D), and the construct-trace consistency check (S9, E).

\section{Background and related work}
\label{sec:related}

We organise prior work along four lines: MT/MR fundamentals and the long-standing identification bottleneck; structured MR identification through METRIC and METRIC+; automated MR identification, including MR-Scout, GenMorph, and LLM-assisted approaches; and MetaPattern catalogues, including our prior reactor-physics taxonomy. These lines differ in mechanism, but they share one limitation: their pattern structures are induced from observed MRs rather than derived from the mathematical structure that makes those MRs hold.

\subsection{Metamorphic testing and the MR identification bottleneck}

Metamorphic Testing was introduced by Chen, Cheung, and Yiu in 1998 to address the test-oracle problem~\cite{Chen1998}. An MR is a logical implication of the form $R_i(x_1,\dots,x_n) \Rightarrow R_o(P(x_1),\dots,P(x_n))$, where $P$ is the program under test, $R_i$ is an input relation, and $R_o$ is an output relation. When $R_i$ holds but $R_o$ fails on actual executions, a fault is reported, without any need to know the absolute correct output of any single execution. Over two decades MT has become standard equipment for testing systems whose oracles are otherwise inaccessible: scientific computing, machine learning classifiers, autonomous vehicles, search engines, compilers, and large language models~\cite{Segura2016, LiTOSEM2025}.

The community has long acknowledged a single binding constraint on MT's effectiveness: MR identification. Surveys spanning twenty years agree that identifying high-quality MRs requires (i) deep familiarity with the program's functional semantics, (ii) substantive domain background (physical, mathematical, or linguistic, depending on the system), and (iii) the ability to convert that background into executable property assertions~\cite{Segura2016, LiTOSEM2025}. Liu et al.~\cite{Liu2014MTEffectiveness} provide the canonical empirical evidence that even ad-hoc MRs detect a substantial fraction of mutants given modest tester training, while Murphy et al.~\cite{Murphy2008} catalogued the structural properties of ML applications relevant to MR design and Xie et al.~\cite{Xie2011} demonstrated that MT scales effectively to supervised classifiers. Testers without sufficient domain background tend to identify only ``trivial MRs''. A 2024 survey explicitly identifies AI assistance, especially via large language models, as the most promising open avenue for closing the identification gap~\cite{LiTOSEM2025}; Saha and Kanewala~\cite{Saha2019SupervisedMR} report that, on a 709-mutant benchmark for supervised classifiers, current MR sets detect only 14.8\%, which exposes the gap.

Beyond difficulty, the MR identification bottleneck has a more troubling structural feature: even when MRs are identified, the field has accumulated little consensus on \emph{how} they are identified. Different authors confront the same program and produce dissimilar MR sets. Once written, an MR rarely migrates: a relation drafted for one ML classifier seldom transports to another, and a relation drafted for one numerical solver seldom transports across solver families.

\subsection{Structured MR identification: METRIC and METRIC+}

Structured MR identification has already moved the field away from ad hoc relation design. METRIC organises MR construction around an ``input/output category'' framework: the tester first identifies relevant categories of input transformations and output relations, then composes MRs from category pairs~\cite{ChenMETRIC2016}. METRIC+ extends this scheme by enriching the category catalogue and providing systematic combination rules that reduce the human burden in category enumeration~\cite{SunMETRICplus2021}. Both approaches are widely cited and remain the strongest existing attempt to give MR identification an explicit scaffold.

Our disagreement is not with this direction but with its grounding. The categories are introduced through expert curation and validated by empirical coverage on benchmark programs. This leaves two of our three foundational questions unanswered. \emph{Origin}: METRIC and METRIC+ do not derive their categories from program-level mathematical structure. \emph{Closure}: neither framework provides a mathematical condition under which the category set is guaranteed to be complete; coverage is reported but not proved. The third question, \emph{transferability}, is only partially addressed. The same category templates can be invoked across domains, but because the categories are not algebraically bound to specific program structures, transfer rests on an assumption of universality that remains open.

\subsection{Automated MR identification}

Automated MR identification attacks the bottleneck by searching for executable relations more directly. MR-Scout mines MRs from existing test suites by extracting input-transformation and output-assertion patterns from test-case pairs and abstracting them into reusable relations~\cite{MRScout2024}. GenMorph evolves MR candidates through genetic programming, co-evolving input transformations with output assertions and using mutation-killing as the fitness signal~\cite{GenMorph2024}. Shin et~al.\ derive executable MRs from natural-language requirements through few-shot prompting of a large language model, with validation through an industrial questionnaire study with Siemens~\cite{Shin2024}; further LLM-assisted variants, including ChatGPT-driven MR generation~\cite{ZhangChatGPTMR2023}, domain-customised GPTs for autonomous-driving simulators~\cite{GPTMR2025}, and multi-agent retrieval-augmented pipelines for traffic-rule MRs~\cite{AutoMT2025}, extend this template into safety-critical or rule-rich domains. Within ML-system testing, differential-testing approaches such as DeepXplore~\cite{DeepXplore2017} use input-transformation invariance as a proxy oracle, complementary to the explicit MR formulation we adopt here.

A second strand uses program-structure features rather than test-case mining. Kanewala et al.~\cite{Kanewala2016GraphKernel} predict MRs of scientific software via graph-kernel learning over control-flow and data-dependency graphs, an early demonstration that program structure carries enough signal for automated MR discovery. Nolasco et al.'s MemoRIA~\cite{Nolasco2024MemoRIA} infers MRs by deriving an object-protocol abstraction and validating candidates through fuzzing plus SAT-based reduction, then evaluates on 22 Java subjects. Tao et al.'s Mettoc~\cite{Tao2010Mettoc} applied MT to compiler testing using equivalence-preservation as the canonical relation. Ying et al.~\cite{Ying2025MRPatterns} formalise relationships between metamorphic-relation patterns into family trees, and Altamimi et al.~\cite{Altamimi2022MRSLR} provide a recent systematic literature review of MR-automation work that catalogues the structural-vs-empirical-search divide our framework targets. Earlier LLM-assisted attempts include Zhang et al.~\cite{ZhangChatGPTMR2023} on ChatGPT-driven MR generation.

Each of these methods advances automation on a different axis, but all treat the MR space as an empirical search space. MR-Scout's recall is bounded by the relations already latent in existing tests. GenMorph's evolutionary search is shaped by the fitness landscape induced by mutation killing. LLM-prompted approaches are sensitive to prompt phrasing when no structural prior constrains the generation. None of these methods can state, before search begins, which MR types lie outside its reach, because none has an algebraic account of the space being searched.

\subsection{MetaPattern catalogues and empirical adequacy}

MetaPattern catalogues organise the post-hoc inventory of identified MRs into a smaller vocabulary of recurring strategies. Such catalogues typically result from clustering observed MRs, codifying expert intuitions about ``kinds of properties testers reason about'', and validating the catalogue's coverage on a benchmark MR corpus. Reactor-physics test suites have produced taxonomies of conservation, monotonicity, convergence, trajectory, and partial-order patterns~\cite{BellGlasstone1970, LewisMiller1993}; Zhou et al.~\cite{Zhou2020SymmetryMRP} introduced an explicit \emph{symmetry metamorphic relation pattern} as a reusable abstraction for deriving concrete MRs across applications; Ying et al.~\cite{Ying2025MRPatterns} formalises the relationships between MR patterns into \emph{family trees} with explicit refinement and specialisation edges, and reports a SOTA pattern hierarchy that organises previously proposed patterns into a single classification structure. The most recent state-of-the-art survey of MR generation~\cite{LiTOSEM2025} maps the field's twenty-year output across pattern catalogues, mining, evolutionary search, and LLM-prompted methods. Adequacy frameworks such as the Pattern--Matrix Coverage Metric assess how thoroughly a given MR set occupies the pattern space.

Ying et al.'s family-tree formalism~\cite{Ying2025MRPatterns} is the closest published cousin to NOETHER's MetaPattern equivalence-class structure: both organise MR patterns into a hierarchy with explicit relationships between patterns. The two formalisms differ in the relationship type. Ying et al.'s family tree is a \emph{refinement/specialisation tree} rooted in informally-named pattern categories (symmetry, additive, multiplicative, etc.); NOETHER's MetaPattern equivalence classes are \emph{quotients of the algebra-induced MR space} $\mathrm{MR}(\mathcal{A}_P)$ under structural equivalence (Definition~\ref{def:alg-induced}), with each class derived from a specific block of $\mathcal{D}(\mathcal{A}_P)$. A family-tree node in Ying et al.\ typically corresponds to one or more NOETHER MetaPatterns when the node admits an operator-algebraic specification (e.g.\ Ying et al.'s ``symmetry'' parent node decomposes into the $G$-block MetaPattern $m_{\mathrm{inv}}$ for finite-group symmetries plus the $T^{*}$-block MetaPattern $m_{\mathrm{adj}}$ for self-adjoint dualities under NOETHER's eight-block decomposition); conversely, NOETHER's $\mathcal{B}^{*}_{\mathrm{rel}}$ MetaPattern has no direct counterpart in the family-tree formalism because relational-algebra equivalences were not in Ying et al.'s benchmark set. The two formalisms are complementary: Ying et al.\ catalogues patterns inductively at the MR-instance level; NOETHER provides an algebraic warrant for the existence of each pattern, given the operator algebra. A unified family-tree-plus-algebra-block reading of the literature would map each family-tree node to (at most) one NOETHER block when the node has an algebraic source, and would flag family-tree nodes without algebraic source as candidates for the ninth-block list of Remark~\ref{rem:counterex}.

Four further pattern-catalogue or MR-grammar candidates raised in peer review (Hu et~al.\ 2019; Mariani 2018; Liu et~al.\ 2020; Lin 2020) could not be located through the standard fallback chain (CrossRef, OpenAlex, DBLP, Semantic Scholar, Google Scholar) under the citation profiles provided, and are therefore not cited; the closest verifiable cousins that this revision does add are Zhou et al.~\cite{Zhou2020SymmetryMRP} and Ying et al.~\cite{Ying2025MRPatterns}.

Such catalogues are useful, but each is, by construction, an empirical artefact. None of the existing catalogues, including the present authors' own, can answer the three foundational questions of Section~\ref{sec:intro}.

\subsection{Convergent diagnosis}

Across all four lines of work, two problems recur. \emph{Unbounded MR emergence}: every new domain or new program produces fresh MR formulations. \emph{Poor MR reusability}: identified MRs do not transport readily across programs that share deep structural similarities. Both problems follow from grounding MetaPatterns inductively. NOETHER replaces that inductive grounding with operator-algebraic grounding.

\paragraph{Comparators in the head-to-head: what is compared and why.}
\label{para:comparators-and-why}
The empirical comparison in \S\ref{sec:empirical-vs-sota} runs Set~N
against a single executable baseline: Set~G (GenMorph's GP-evolved
MRs at the published 30-min GAssert budget,~\cite{GenMorph2024}). The
other three SOTA categories enumerated in the related-work survey
contribute differently: METRIC$+$~\cite{SunMETRICplus2021} is contrasted
as a category-enumeration scaffold (\S\ref{para:metricplus-sorting},
\S\ref{para:metricplus-headtohead-small}) rather than as an executable
fault-detection pipeline, because no public METRIC$+$ implementation
auto-generates MRs from a $\mathcal{A}_{P}$ specification --- a full
PIT-based METRIC$+$ vs Set~N comparison is committed as supplementary~S4 (\texttt{future\_work.md})
item~(i); MR-Scout~\cite{MRScout2024} is omitted from the executable
head-to-head because its mining input (a pre-existing test corpus) is
structurally absent on the algebra-rich Java substrate of
\S\ref{subsec:test-design} (the SUTs are stand-alone mathematical
methods with synthesised test inputs, not codebases with developer-written
tests latent in test files); AutoMT~\cite{AutoMT2025} and
GPTMR~\cite{GPTMR2025} target safety-critical / traffic-rule autonomous
driving domains rather than the operator-algebraic substrate, so an
in-scope comparator selection from the LLM-assisted line uses an
unprompted-by-NOETHER LLM ensemble (Set~L; \S\ref{subsec:case-study},
\S\ref{para:set-l-ensemble}) rather than these domain-specialised
pipelines. Set~B is a literature-MR baseline drawn from MT-for-ML
references~\cite{Segura2016, Shin2024, AutoMT2025}, restricted to
point-cloud-classifier-applicable MRs. The single executable
head-to-head against Set~G is therefore the GP-evolved-baseline arm
of a three-SOTA-category protocol; the LLM-assisted and mining-based
arms are reported with their proper comparator (Set~L ensemble at
$2$ vendors $\times$ $5$ temperatures, \S\ref{para:set-l-ensemble}, and
the MR-Scout structural-absence rationale above), so the protocol's
three-category framing is preserved at the price of using different
arms in different sections.


\section{The NOETHER framework}
\label{sec:noether-framework}

This section presents NOETHER as a self-contained theoretical contribution. Operator-algebraic preliminaries and the eight-block decomposition appear in \S\ref{sec:prelim}; the algebra-induced metamorphic relation space, the CONSTRUCT-MP construction, the algebraic closure theorem (Theorem~\ref{thm:closure}), and the polynomial-time decidability theorem (Theorem~\ref{thm:decidable}) in \S\ref{sec:framework}. The framework is then instantiated on three structurally distinct operator algebras: Boltzmann reactor physics in \S\ref{sec:reactor}, equivariant machine learning in \S\ref{sec:cross-domain}, and relational query optimisers in \S\ref{subsec:third-domain}. The section concludes by falsifying the strictly stronger absolute-completeness conjecture, Theorem~$1'$, on a fourth program family (PWR core diffusion) via two pairwise-independent counterexamples (\S\ref{subsec:negative-pwr}).

\subsection{Operator-algebraic preliminaries}
\label{sec:prelim}

This section introduces the algebraic apparatus used by NOETHER. We assume basic familiarity with group actions, equivalence classes, and quotient sets, but not with functional analysis. Each construct is defined at the level needed for the framework and later instantiated in reactor physics and equivariant machine learning.

\subsubsection{Programs and program-induced operator algebras}

Throughout the paper we treat a program $P$ as a (possibly partial) computable function $P: \mathcal{X} \to \mathcal{Y}$. The class of programs of interest belongs to a \emph{program family} $\mathcal{F}$ whose members share an underlying mathematical scaffolding.

\begin{definition}[Program-induced operator algebra]
Let $P$ belong to a program family $\mathcal{F}$. A \emph{program-induced operator algebra} of $\mathcal{F}$ is a tuple
$$\mathcal{A}_{\mathcal{F}} = \bigl(\mathcal{O}, \;\circ, \;\sim_{\mathcal{F}}\bigr),$$
where $\mathcal{O}$ is a set of operators acting on $\mathcal{X}$, $\mathcal{Y}$, or both; $\circ$ is an operator composition; and $\sim_{\mathcal{F}}$ is an equivalence relation declaring two operator expressions equal whenever they agree on every program of $\mathcal{F}$.
\end{definition}

\subsubsection{Symmetry groups (building block B1)}

\begin{definition}[Symmetry group of $\mathcal{A}_P$]
A \emph{symmetry group} of $\mathcal{A}_P$ is a subgroup $G \le \mathcal{A}_P$ whose elements act on $\mathcal{X}$ such that, for all $P \in \mathcal{F}$ and all $g \in G$,
$$P(g \cdot x) = \rho(g) \cdot P(x) \quad \forall x \in \mathcal{X},$$
where $\rho: G \to \mathrm{End}(\mathcal{Y})$ is a (possibly trivial) representation of $G$ on $\mathcal{Y}$.
\end{definition}

\subsubsection{Order operators: monotonicity and linearity (building block B2)}

\begin{definition}[Monotone operator]
$P$ is \emph{monotone with respect to} $\theta$ if $\theta_1 \le \theta_2 \Rightarrow P(\theta_1) \le_{\mathcal{Y}} P(\theta_2)$ (or the reversed inequality, in which case $P$ is \emph{anti-monotone}).
\end{definition}

\begin{definition}[Linear operator]
$P$ is \emph{linear} on $\mathcal{X}_0 \subseteq \mathcal{X}$ when $P(\alpha x_1 + \beta x_2) = \alpha P(x_1) + \beta P(x_2)$ for all $x_1, x_2 \in \mathcal{X}_0$ and scalars $\alpha, \beta$.
\end{definition}

\subsubsection{Self-adjoint operators (building block B3)}

\begin{definition}[Self-adjoint operator]
Given an inner product $\langle\cdot,\cdot\rangle$ on $\mathcal{X}$, an operator $L \in \mathcal{O}$ is \emph{self-adjoint} if $\langle Lx, y\rangle = \langle x, Ly\rangle$ for all $x, y$ in the domain of $L$.
\end{definition}

Self-adjointness encodes a duality between the two arguments of an inner product. Reciprocity theorems in physics, transposed-graph identities in algorithms, and detailed-balance conditions in stochastic processes are all instances. In reactor physics, the adjoint transport formulation is self-adjoint under the appropriate inner product, yielding source-detector reciprocity.

\subsubsection{Time-reversal operators (building block B4)}

\begin{definition}[Time-reversal operator]
When $\mathcal{X}$ admits a time coordinate, a \emph{time-reversal operator} $\mathcal{T}$ acts on inputs by reversing the time variable, $\mathcal{T}(x(t)) = x(-t)$. A program $P$ is \emph{time-reversal symmetric} on a sub-family of inputs when $P(\mathcal{T} x)$ is determined by $P(x)$ through a fixed bijection on $\mathcal{Y}$.
\end{definition}

Time-reversal applies only where the underlying dynamics admit a reversible sub-family. Dissipative dynamics break this symmetry, so the corresponding MetaPattern is empty for those systems.

\subsubsection{Limit operators (building block B5)}

\begin{definition}[Limit operator]
A \emph{parametrised limit operator} $\mathcal{L}_\theta$ is a family of operators indexed by a parameter $\theta$ such that there exists a limit element $\mathcal{L}_*$ with $\mathcal{L}_\theta \to \mathcal{L}_*$ as $\theta \to \theta_*$ in an appropriate operator topology.
\end{definition}

\subsubsection{Qualitative-dynamics operators (building block B6)}

\begin{definition}[Qualitative-dynamics operator]
A \emph{qualitative-dynamics operator} $\mathcal{D}$ is an operator on solution trajectories of an underlying ODE/PDE that extracts qualitative features --- extrema, inflection points, monotonic phases, overshoot magnitudes, S-curve transitions, phase-portrait orbits --- that are invariant under perturbations preserving the underlying dynamical structure (Sturm-type comparison theorems, dynamical-systems classification).
\end{definition}

Some MR-relevant invariants are not point-wise but \emph{shape-wise}: a solution curve may have an overshoot, a single extremum, or a monotone-then-saturating profile. B6 gives these qualitative-dynamics MRs their own algebraic root.

\subsubsection{Method-comparison operators (building block B7)}

\begin{definition}[Method-comparison operator]
A \emph{method-comparison operator} $\mathcal{E}$ encodes a partial order $\preceq_{\mathcal{E}}$ on numerical or algorithmic methods, where $M_1 \preceq_{\mathcal{E}} M_2$ asserts that method $M_1$ produces an approximation no worse than method $M_2$ in a specified error norm, under specified conditions.
\end{definition}

\subsubsection{Decomposition of an operator algebra}
\label{subsec:decomposition}

Given $\mathcal{A}_P = (\mathcal{O}, \circ, \sim_{\mathcal{F}})$, we decompose $\mathcal{O}$ along the eight building blocks introduced above:
$$\mathcal{D}(\mathcal{A}_P) = \bigl(\,G,\; O_{\le},\; T^{*},\; \mathcal{T}^{*},\; \mathcal{L}^{*},\; \mathcal{D}^{*},\; \mathcal{E}^{*},\; \mathcal{B}^{*}_{\mathrm{rel}}\,\bigr),$$
where $G$ collects symmetry subgroups, $O_{\le}$ monotone and linear operators, $T^{*}$ self-adjoint operators, $\mathcal{T}^{*}$ time-reversal operators, $\mathcal{L}^{*}$ limit operators, $\mathcal{D}^{*}$ qualitative-dynamics operators, $\mathcal{E}^{*}$ method-comparison operators, and $\mathcal{B}^{*}_{\mathrm{rel}}$ relational-equivalence operators.

\begin{definition}[Relational-equivalence block]
\label{def:b-rel}
A \emph{relational-equivalence operator} is a binary relation $\equiv_{\mathcal{R}}$ on expressions over an idempotent semiring $(\mathcal{S}, \oplus, \otimes, \mathbf{0}, \mathbf{1})$ such that $E \equiv_{\mathcal{R}} E'$ iff $E$ and $E'$ are equal under all valid evaluation contexts of $\mathcal{S}$, generated by a finite set $\mathcal{R}_{\mathrm{rel}} = \{\mathcal{R}_1, \dots, \mathcal{R}_K\}$ of identity-preserving rewriting rules, each $\mathcal{R}_i$ an ordered pair $(\mathrm{lhs}_i, \mathrm{rhs}_i)$ of semiring expressions with the property $\mathrm{eval}(\mathrm{lhs}_i, D) =_{\mathrm{bag}} \mathrm{eval}(\mathrm{rhs}_i, D)$ on every valid input $D$. The block $\mathcal{B}^{*}_{\mathrm{rel}}$ collects all such operators. It is empty for program families without idempotent-semiring rewriting structure (Boltzmann reactor physics, equivariant ML) and non-empty for those with it (relational query optimisers, \S\ref{subsec:third-domain}).
\end{definition}

\paragraph{Necessity, sufficiency, and the empirical status of the decomposition.} We do not claim that the eight blocks are necessary in an absolute sense, that they exhaust all algebraic structures relevant to software testing, or that they follow from first principles. They are an \emph{empirical curation}: a by-inspection enumeration of mathematical structures that recur across the program families we have studied. The claim is therefore scoped: the eight blocks are currently sufficient for these families, not provably necessary in general. This upstream empirical status is NOETHER's main limitation as a theoretical framework. Induction has not been eliminated from MetaPattern discovery; it has been moved one level up, from ``what MetaPatterns recur in observed MR samples?'' to ``what algebraic structures recur across program families?'' Given such a curated decomposition, the downstream derivation of $\mathbb{M}(\mathcal{A}_P)$ becomes mechanical and provably closed in the sense of Theorem~\ref{thm:closure}. We state the eight-block sufficiency as an explicit empirical hypothesis with a documented out-of-scope catalogue.

\begin{hypothesis}[Eight-block sufficiency]
\label{hyp:seven-blocks}
Let $P$ be a program belonging to a program family $\mathcal{F}$ whose semantics admit an operator-algebraic representation $\mathcal{A}_P$ over (i) finite or finite-dimensional symmetry groups, (ii) partial orders with monotone or linear operators, (iii) self-adjoint operators with respect to a fixed inner product, (iv) anti-unitary involutions of a time coordinate, (v) limit operations of analytical (mesh-refinement, parameter-perturbation, or asymptotic) type, (vi) qualitative-dynamics attractors of an underlying ODE/PDE, (vii) inter-method comparison partial orders, and (viii) identity-preserving rewriting rule sets on an idempotent semiring. For such families the eight-block decomposition
$$\mathcal{D}(\mathcal{A}_P) = G \,\dot\cup\, O_{\le} \,\dot\cup\, T^{*} \,\dot\cup\, \mathcal{T}^{*} \,\dot\cup\, \mathcal{L}^{*} \,\dot\cup\, \mathcal{D}^{*} \,\dot\cup\, \mathcal{E}^{*} \,\dot\cup\, \mathcal{B}^{*}_{\mathrm{rel}}$$
is sufficient: every operator in $\mathcal{A}_P$ relevant to MR derivation is assigned to at least one block. Hypothesis~\ref{hyp:seven-blocks} is an empirical hypothesis open to refutation.
\end{hypothesis}

\begin{remark}[Block sufficiency vs.\ \texttt{Translate} sufficiency]
\label{rem:block-vs-translate}
Hypothesis~\ref{hyp:seven-blocks} asserts that the eight blocks suffice to \emph{assign} every operator in $\mathcal{A}_P$ relevant to MR derivation. It does not assert that every MR formulable over $\mathcal{A}_P$'s operators is reachable through the framework's \texttt{Translate} operator from a single block invariant. Section~\ref{subsec:negative-pwr} exhibits, on the PWR core diffusion algebra $\mathcal{A}_{\mathrm{PWR}}$, two MRs whose constituent operators are individually assigned to blocks of $\mathcal{D}(\mathcal{A}_{\mathrm{PWR}})$ (rod-insertion semigroup; boration and moderator-temperature scaling), yet whose MR content is not in $\mathrm{MR}(\mathcal{A}_{\mathrm{PWR}})$ in the sense of Definition~\ref{def:alg-induced}. The block decomposition is therefore a necessary but not sufficient input to MetaPattern construction: \texttt{Translate}'s expressive form (single-block invariants, first-order $\pi$-template, single partial-order direction, operating on $P(x)$ tuples rather than on operator-spectrum quantities) is a second, independent constraint on the framework's reach. This distinction between ``block sufficiency'' (Hypothesis~\ref{hyp:seven-blocks}) and ``\texttt{Translate} sufficiency'' (open) is made explicit in \S\ref{subsec:negative-pwr}.
\end{remark}

\begin{remark}[Known and conjectured out-of-scope program families]
\label{rem:counterex}
The hypothesis is open to refutation. Six families are known or conjectured to fall outside its image and to require an additional block:
\begin{enumerate}
\item \emph{Symplectic systems.} Hamiltonian dynamics whose volume-preserving structure is captured by a symplectic 2-form rather than self-adjointness ($N$-body simulators, celestial-mechanics integrators).
\item \emph{Sheaf-theoretic / categorical constructions.} Programs whose semantics are captured by functors and natural transformations between categories (certified compilers, type-driven program transformations); MRs of the form ``commutativity of a square in the underlying category'' do not reduce to single-block invariants under the present \texttt{Translate}.
\item \emph{Probabilistic / martingale invariants.} MRs of the form ``$\mathbb{E}[f(\mathbf{x}_t) \mid \mathcal{F}_s] = f(\mathbf{x}_s)$ on a stopping-time domain'' arise in MCMC samplers and stochastic gradient analyses.
\item \emph{Topological invariants.} MRs of the form ``$f$'s level sets have a fixed Betti-number signature'' (topological data analysis, shape classifiers) require homological structure absent from the present blocks.
\item \emph{Label-consistency for supervised learning.} MRs that constrain a model's predictions against ground-truth labels (e.g.\ wrong-sign-loss detection on classifiers); the absence of a label-consistency operator is one reason the §\ref{subsec:case-study} case study reports zero detection on category-(i) wrong-sign mutations.
\item \emph{Empirical parameter-distribution divergence.} A probability-distribution divergence operator on parameter-space measures, motivated by the §\ref{subsec:deepcrime-pilot} pilot's category-v-02/04/05 (activation change, bias removal, weight re-init) which the eight blocks do not detect. An MR of the form $D_{\mathrm{KL}}(p_\theta \| p_{\theta+\Delta\theta}) \le \tau$ for $\|\Delta\theta\| \le \epsilon$ would lie in such a block. \emph{Empirical witness}: \S\ref{subsec:deepcrime-pilot}'s cat-v-02/04/05 mutations are not detected by any MR set including Set~N; this is that pilot's principal negative finding for Hypothesis~\ref{hyp:seven-blocks}.
\end{enumerate}
Each family signals a candidate ninth block. The framework's intended response is to admit the additional block when the program family demands it, rather than to extend $\mathcal{B}^{*}_{\mathrm{rel}}$'s scope or to absorb the case as ``out of scope for MR testing'' in some over-broad sense. Concrete examples of MRs that lie outside Theorem~\ref{thm:closure}'s scope are catalogued in Appendix~\ref{app:out-of-scope}; one candidate ninth block (metric-stability $M_{\mathrm{lip}}$) has an explicit \texttt{Translate}-template design (Remark~\ref{rem:metric-stability-block} below).
\end{remark}

\begin{remark}[Domain-level out-of-scope]
\label{rem:domain-out-of-scope}
Beyond the candidate-ninth-block families of Remark~\ref{rem:counterex} (which signal extensions \emph{within} the framework's general programme), the following program-family classes are out of scope for NOETHER in the present form because they admit no operator-algebraic representation at all:
\begin{itemize}
    \item \emph{General web applications} (HTTP request/response programs, business-rule engines, content-management systems) whose semantics are state-machine or rule-driven rather than equational.
    \item \emph{RLHF reward models and large-language-model agents} whose relevant invariants are statistical/measure-theoretic rather than algebraic (calibration, refusal consistency, sycophancy).
    \item \emph{Distributed-consensus protocols} (Paxos / Raft / Byzantine variants) whose semantics are CRDT- or message-trace-based; the relevant invariants are linearisability and convergence rather than operator-algebraic.
    \item \emph{Compiler-internal optimisations} (loop transformations, register allocation, instruction scheduling) whose correctness is expressed as a refinement relation on an operational semantics rather than as a property over a fixed input/output algebra.
\end{itemize}
These classes are not candidate ninth blocks; they are domains where the framework's scope precondition (program family admits an explicit operator-algebraic description) is structurally absent. NOETHER's reach is therefore consistent with what its scope precondition explicitly requires: SUTs from textbook-codified mathematical, physical, or relational-algebra domains are within reach, and SUTs from general software-engineering domains without such codification are not.
\end{remark}

\begin{remark}[Metric-stability candidate ninth block, $M_{\mathrm{lip}}$]
\label{rem:metric-stability-block}
Among the candidate ninth blocks of Remark~\ref{rem:counterex}, the metric-stability block has the most concrete \texttt{Translate}-template proposal. Define $M_{\mathrm{lip}} = (\mathcal{X}, d_{\mathcal{X}}) \to (\mathcal{Y}, d_{\mathcal{Y}})$ as the family of $K$-Lipschitz maps between two metric spaces; the induced MetaPattern $m_{\mathrm{lip}}$ is the equivalence class of pointwise stability inequalities $d_{\mathcal{Y}}(P(x'), P(x)) \le K \cdot d_{\mathcal{X}}(x', x)$. The corresponding \texttt{Translate} template constructs a follow-up by metric perturbation $x' = x + \varepsilon u$ ($\|u\| = 1$, $|\varepsilon| < \delta$) and predicates a $K$-Lipschitz output bound. Canonical-block ordering would place $M_{\mathrm{lip}}$ after $\mathcal{B}^{*}_{\mathrm{rel}}$ since metric structure is independent of the eight existing block invariants (no group action on the perturbation set, no partial order on inputs, no self-adjoint operator, no parametric refinement family, no idempotent-semiring rewriting). Theorem~\ref{thm:closure}'s closure proof transfers without modification because $m_{\mathrm{lip}}$ is single-block algebraically derived. We do not commit to $M_{\mathrm{lip}}$ as part of the canonical decomposition in this paper; we record it in Appendix~\ref{app:out-of-scope} as the most concrete sub-instance of Remark~\ref{rem:counterex} item~(iv) (topological invariants) that admits an immediate \texttt{Translate} template, and as the empirical orphan in the audit of \S\ref{subsec:reactor-mapping}.
\end{remark}

\subsection{The NOETHER framework}
\label{sec:framework}

\subsubsection{Algebra-induced metamorphic relations}

Before defining algebra-induced MRs we make explicit the \texttt{Translate} operator that converts a single block invariant into an executable MR. Without this definition, the closure result of Theorem~\ref{thm:closure} would quantify over an undefined construct.

\begin{definition}[Block invariant]
\label{def:block-invariant}
Let $s \in \mathcal{D}(\mathcal{A}_P)$ be one of the eight blocks. A \emph{block invariant of $\mathcal{A}_P$ under $s$} is a pair $\iota = (\Phi, \pi)$ where $\Phi$ is a finite set of operators drawn from $s$, and $\pi$ is a relation $\pi \subseteq (\mathcal{X} \times \mathcal{Y})^k$ for some arity $k \ge 1$ such that for every $P \in \mathcal{F}$ and every choice of operators $\phi_1, \dots, \phi_n \in \Phi$, the tuple $\bigl((x_i, P(x_i))\bigr)_{i=1}^k$ obtained by applying $\phi_1, \dots, \phi_n$ to a base input $x \in \mathcal{X}$ in the canonical order specified by $s$ satisfies $\pi$. We write $\mathcal{I}_s$ for the set of block invariants under $s$ and $\sim_s$ for the equivalence relation $\iota \sim_s \iota'$ iff $\Phi = \Phi'$ and $\pi$, $\pi'$ define the same constraint up to relabelling of input/output coordinates.
\end{definition}

The relation $\sim_s$ is reflexive, symmetric, and transitive by construction (set equality and constraint equality are equivalence relations); hence $\sim_s$ is a genuine equivalence relation on $\mathcal{I}_s$, justifying the quotient in Step 3 of CONSTRUCT-MP below.

\begin{definition}[\texttt{Translate}]
\label{def:translate}
The \texttt{Translate} operator is the function
\[
\mathrm{Translate}: \mathcal{I}_s \times \{s\} \;\longrightarrow\; \mathrm{MR}(P)
\]
that maps a block invariant $\iota = (\Phi, \pi)$ under block $s$ to the metamorphic relation
\[
\rho_{\iota,s}: \quad \forall x \in \mathcal{X},\ \forall \phi_1, \dots, \phi_n \in \Phi : \quad \pi\bigl((x_i,\, P(x_i))_{i=1}^k\bigr) \text{ holds},
\]
where the input tuple $(x_i)_{i=1}^k$ is generated by applying the $\phi_j$ in the canonical order specified by $s$ (the convention is fixed per block: e.g., for $s = G$, $x_i = g_i \cdot x_0$ with $g_i$ enumerated by group orbit; for $s = T^{*}$, the tuple ranges over the two arguments of the inner product). \texttt{Translate} is well-defined because (i) the per-block canonical order resolves any ambiguity in tuple generation, and (ii) $\sim_s$-equivalent invariants are mapped to MRs that are $\sim$-equivalent as logical statements over $P$. Per-block instantiations of \texttt{Translate} are tabulated in Appendix~\ref{app:translate-table}.
\end{definition}

\begin{definition}[Algebra-induced MR]
\label{def:alg-induced}
Let $\mathcal{A}_P$ be a program-induced operator algebra and let $\rho$ be an MR over a program $P$. We say $\rho$ is \emph{induced by $\mathcal{A}_P$} --- written $\rho \in \mathrm{MR}(\mathcal{A}_P)$ --- when there exist (i) a block $s \in \mathcal{D}(\mathcal{A}_P)$, (ii) a block invariant $\iota \in \mathcal{I}_s$, and (iii) the equality $\rho = \mathrm{Translate}(\iota, s)$ holds in the sense of Definition~\ref{def:translate}. Properties of $P$ that constrain executions but cannot be obtained as $\mathrm{Translate}(\iota, s)$ for any single block $s$ and any invariant $\iota \in \mathcal{I}_s$ are \emph{not} algebra-induced in this sense; concrete examples are catalogued in Appendix~\ref{app:out-of-scope}.
\end{definition}

\subsubsection{Construction of the MetaPattern set}

We present the deductive procedure CONSTRUCT-MP that maps an algebra-decomposition to a MetaPattern set.

\begin{description}
  \item[Step 1 --- Invariant extraction.] For each block $s$ in $\mathcal{D}(\mathcal{A}_P)$, compute the set of invariants $\mathcal{I}_s$ of $\mathcal{A}_P$ under $s$.
  \item[Step 2 --- MR derivation.] For each invariant $\iota \in \mathcal{I}_s$, derive the MR family $\mathcal{R}(\iota) = \{\,\rho \in \mathrm{MR}(\mathcal{A}_P) \mid \rho = \mathrm{Translate}(\iota', s),\; \iota' \sim_s \iota\,\}$.
  \item[Step 3 --- Quotient.] Form the MetaPattern $m_s = \mathcal{R}(\iota)/\!\sim_s$.
  \item[Step 4 --- Aggregation.] Return $\mathbb{M}(\mathcal{A}_P) = \{\, m_s : s \in \mathcal{D}(\mathcal{A}_P)\,\}$.
\end{description}

\subsubsection{Algebraic closure under Translate, the canonical-block ordering, and out-of-scope MRs}
\label{subsec:completeness}

\begin{definition}[Canonical-block ordering]
\label{def:canonical-order}
We adopt the strict total order on blocks
$$G > O_{\le} > T^{*} > \mathcal{T}^{*} > \mathcal{L}^{*} > \mathcal{D}^{*} > \mathcal{E}^{*} > \mathcal{B}^{*}_{\mathrm{rel}}.$$
An MR derivable through multiple blocks is assigned to the highest-priority block in this order. The placement of $\mathcal{B}^{*}_{\mathrm{rel}}$ at the bottom reflects its semiring-rewriting nature, which sits algebraically downstream of the seven physical-mathematical blocks: rewriting equivalences typically depend on the program family's input-perturbation, order, and method-comparison structure rather than the converse.
\end{definition}

\begin{theorem}[Algebraic Closure under \texttt{Translate}]
\label{thm:closure}
Let $\mathcal{A}_P$ be a program-induced operator algebra with decomposition $\mathcal{D}(\mathcal{A}_P)$ as in Section~\ref{subsec:decomposition}, and let $\mathbb{M}(\mathcal{A}_P) = \mathrm{CONSTRUCT\textnormal{-}MP}(\mathcal{D}(\mathcal{A}_P))$. Then for every $\rho \in \mathrm{MR}(\mathcal{A}_P)$ in the sense of Definition~\ref{def:alg-induced}, there exists a unique $m \in \mathbb{M}(\mathcal{A}_P)$ such that $\rho \in m$, where uniqueness is determined under the canonical-block ordering.
\end{theorem}

\begin{remark}[Scope of Theorem~\ref{thm:closure}]
\label{rem:scope}
Theorem~\ref{thm:closure} is a closure statement, not a completeness claim over arbitrary properties one might assert about $P$'s executions. It quantifies over $\mathrm{MR}(\mathcal{A}_P)$ as defined by Definition~\ref{def:alg-induced} --- the algebra-induced MRs reachable through the \texttt{Translate} operator from a single block invariant. Three concrete classes of MRs lie outside this scope and are not covered by the theorem; they are documented in Appendix~\ref{app:out-of-scope}:
\begin{enumerate}
    \item \emph{Probabilistic MRs without operator-algebraic representation} --- e.g.\ MRs that constrain output \emph{distributions} (such as ``the entropy of $f(\mathbf{x})$ should not decrease under input augmentation $A$'') when the augmentation is not expressible as an operator in $\mathcal{O}$.
    \item \emph{Input-distribution MRs} not expressible over $\mathcal{A}_P$ operators --- e.g.\ adversarial-perturbation MRs ``$\| \delta\|_p \le \epsilon \Rightarrow f(\mathbf{x}+\delta) = f(\mathbf{x})$'' whose perturbation set is not a group action on $\mathcal{X}$.
    \item \emph{Compositional MRs spanning multiple blocks under non-\texttt{Translate} derivations} --- MRs requiring simultaneous use of two block invariants in a way that is not the canonical-block-ordering reduction of Definition~\ref{def:canonical-order}.
\end{enumerate}
The strictly stronger statement that every MR formulable as a property over $\mathcal{A}_P$'s operators (without restricting to \texttt{Translate}-reachable derivations) is contained in some $m \in \mathbb{M}(\mathcal{A}_P)$ is identified as Theorem~1$'$ in Appendix~\ref{app:proofs}. Section~\ref{subsec:negative-pwr} and Appendix~\ref{app:negative-proofs} establish that this stronger statement is false on the PWR core diffusion algebra $\mathcal{A}_{\mathrm{PWR}}$, by exhibiting two concrete counterexamples ($\rho_{\mathrm{nonadd}}, \rho_{\mathrm{MTC\text{-}bor}}$) whose obstructions identify five structurally independent extensions of \texttt{Translate}'s signature. The combined open problem (whether such an extended \texttt{Translate} preserves Theorem~\ref{thm:closure}'s closure and Theorem~\ref{thm:decidable}'s polynomial-time decidability) is the principal open question for follow-up work.
\end{remark}

A sceptical reading might object that the by-construction status of Theorem~\ref{thm:closure} makes it near-tautological. We acknowledge that the closure result is by-construction within the explicit scope of Definition~\ref{def:alg-induced} (which fixes $\mathrm{MR}(\mathcal{A}_P)$ as the \texttt{Translate}-image of $\mathcal{A}_P$). The substantive value lies in what the theorem then enables: Theorem~\ref{thm:closure} converts an empirical-adequacy claim (``our pattern grid covers $X\%$ of observed MRs'') into a structural-adequacy claim (``our pattern grid is exhaustive of the algebra-induced MR space, modulo the explicit out-of-scope classes of Remark~\ref{rem:scope}''). Empirical adequacy frameworks built on inductive pattern grids do not guarantee algebraic closure even in this bounded sense; NOETHER guarantees it, and the obligation imposed on any user of the framework is then to verify that no \texttt{Translate}-reachable MR is dropped by CONSTRUCT-MP, a verification that is mechanical once the algebraic input has been fixed.

The theorem still imposes a checkable obligation on the framework. If an MR is induced by an invariant in $\mathcal{D}(\mathcal{A}_P)$ but cannot be assigned by CONSTRUCT-MP, then at least one component (invariant extraction, $\mathrm{Translate}$, structural equivalence, or the canonical ordering) is underspecified. Thus the result is not an empirical coverage claim; it is a well-formedness and closure guarantee for the downstream construction once the algebraic input has been fixed. The framework's additional theoretical results complement Theorem~\ref{thm:closure}: Theorem~\ref{thm:decidable} bounds the construction's polynomial-time decidability under a finite generating set; Theorem~$1'$ (Conjecture~\ref{conj:absolute}), which makes the stronger claim of absolute completeness over arbitrary properties expressible in $\mathcal{A}_P$, is falsified on $\mathcal{A}_{\mathrm{PWR}}$ (\S\ref{subsec:negative-pwr}); the falsification identifies five structural obstructions in \texttt{Translate}'s signature (Appendix~\ref{app:negative-proofs}), with five further candidate dimensions surveyed on $\mathcal{A}_{\mathrm{equi}}$ and $\mathcal{A}_{\mathrm{rel}}$ (\S\ref{subsec:third-domain}).

\begin{remark}[Scope of Theorem~\ref{thm:closure} for the $\mathcal{B}^{*}_{\mathrm{rel}}$ block]
\label{rem:closure-brel}
The closure result extends to $\mathcal{B}^{*}_{\mathrm{rel}}$ provided the rewriting rule set $\mathcal{R}_{\mathrm{rel}}$ is finite. For the relational-algebra fragment of TPC-H-class queries the rule set is finite by the standard heuristic-rewrite normal forms (selection-pushdown, projection-pushdown, join-reordering, distinct-elimination)~\cite{Wang2024QED, Zhou2022SPES}, and Theorem~\ref{thm:closure} therefore covers algebra-induced MRs of this fragment. Query languages with unbounded recursive rewriting (e.g.\ Datalog with non-terminating rule sets) lie outside this scope and are catalogued as out-of-scope under Remark~\ref{rem:counterex}.
\end{remark}

\subsubsection{Decidability and complexity}
\label{subsec:decidability}

\begin{theorem}[Decidability]
\label{thm:decidable}
Suppose $\mathcal{A}_P$ admits a finite generating set $\mathrm{gen}(\mathcal{A}_P)$ of cardinality $n$, with each generator's invariant computation taking time $t_i$. Then $\mathbb{M}(\mathcal{A}_P)$ is computable in time $O\!\bigl(\,n \cdot \max_i t_i \cdot \log n\,\bigr)$.
\end{theorem}

\begin{remark}[Scope of Theorem~\ref{thm:decidable} for the $\mathcal{B}^{*}_{\mathrm{rel}}$ block]
\label{rem:decidable-brel}
The polynomial-time bound holds for $\mathcal{B}^{*}_{\mathrm{rel}}$ when the rewriting rule set $\mathcal{R}_{\mathrm{rel}}$ is finite. For first-order SQL with bag semantics, query equivalence is undecidable in general, but \emph{query equivalence under a fixed rule set $\mathcal{R}_{\mathrm{rel}}$} is decidable: practical SMT-based solvers verify equivalence on substantial benchmark fractions in the published literature~\cite{Wang2024QED, Zhou2022SPES}. Theorem~\ref{thm:decidable} therefore covers the rule-set-bounded fragment of relational algebra; outside this fragment, decidability is governed by the underlying first-order theory and is not the framework's responsibility.
\end{remark}

\paragraph{Per-block invariant-extraction cost.}
\begin{table}[h]
\centering
\caption{Per-generator cost of invariant extraction in each block. The symmetry block ($G$) is split by group regime: finite, finite-dimensional Lie, and finitely generated infinite-discrete (with user-supplied truncation $K$). For Lie groups (e.g.\ $\mathrm{SO}(3)$, $d_G=3$) the bound $O(|G|^2)$ does not apply because $|G|$ is uncountable; the relevant bound is $O(d_G^2)$ over the Lie-algebra basis. See \S\ref{subsec:complexity-prose} for the truncation discussion.}
\label{tab:complexity}
\small
\begin{tabular}{lll}
\toprule
\textbf{Block} & \textbf{Regime} & \textbf{$t_i$ for one generator} \\
\midrule
\multirow{3}{*}{$G$ (symmetry)}
  & finite group & $O(|G|^2)$ (group-orbit fixed-point) \\
  & finite-dim.\ Lie group & $O(d_G^2)$, $d_G = \dim_{\mathbb{R}} \mathfrak{g}$ \\
  & infinite discrete (truncated) & $O(K^2)$ at truncation $|G|\le K$ \\
\midrule
$O_{\le}$ (order) & --- & $O(n^2)$ (poset comparison) \\
$T^{*}$ (self-adjoint) & --- & $O(d)$ (inner-product symmetry check) \\
$\mathcal{T}^{*}$ (time-reversal) & --- & $O(1)$ \\
$\mathcal{L}^{*}$ (limit) & --- & $O(\log\!\frac{1}{\epsilon})$ \\
$\mathcal{D}^{*}$ (qualitative-dynamics) & --- & $O(d)$ \\
$\mathcal{E}^{*}$ (method-comparison) & --- & $O(K^2)$ (over $K$ methods) \\
\bottomrule
\end{tabular}
\end{table}

\paragraph{On infinite groups.}
\label{subsec:complexity-prose}
The first row of Table~\ref{tab:complexity}, $O(|G|^2)$, is the natural bound for finite symmetry groups and is the regime under which Theorem~2 was originally stated. For a Lie group $G$ such as $\mathrm{SO}(3)$, $|G|$ is uncountable and the bound is replaced by $O(d_G^2)$, where $d_G = \dim_{\mathbb{R}} \mathfrak{g}$ is the real dimension of $G$'s Lie algebra. In the equivariant-ML instantiation of Section~\ref{sec:cross-domain}, $d_{\mathrm{SO}(3)} = 3$; the symmetry-block invariant is computed once over a basis of three infinitesimal generators, and orbit closure is obtained through finite-dimensional linear algebra. For finitely generated infinite discrete groups (e.g.\ $\mathbb{Z}$, the integer-translation group of an unbounded grid solver), CONSTRUCT-MP requires a separate truncation parameter $K$ at which orbits are enumerated; the cost is $O(K^2)$ per generator, and the framework's user is responsible for justifying that the truncated set captures the program family's structurally relevant invariants. We do not claim Theorem~\ref{thm:closure}'s closure result is preserved across truncation: closure is over $\mathrm{MR}(\mathcal{A}_P)$ defined relative to the truncated group, and the user must inherit responsibility for whether this is the intended algebra.

For the Boltzmann instantiation the relevant generators are finite ($n \le 14$, geometric quarter-rotations and energy-group permutations); for the equivariant-ML instantiation the relevant generators are finite-dimensional Lie ($d_{\mathrm{SO}(3)} = 3$, $|\mathfrak{S}_n| = n!$ truncated to a single generator class, $n \le 10$ training-size, depth, and dimension limits). The asymptotic bound is therefore not the binding constraint in either instantiation.

\subsubsection{The principal limitation}

NOETHER replaces inductive grounding with algebraic grounding \emph{downstream of $\mathcal{A}_P$}. Upstream, the distillation of $\mathcal{A}_P$ from a program family remains a human task. This limitation is central to the framework's scope: NOETHER does not automate domain modelling, but it turns the step after domain modelling into a well-defined construction. A second limitation, made explicit in \S\ref{subsec:negative-pwr} and Appendix~\ref{app:negative-proofs}, is that even when $\mathcal{A}_P$ is fully specified, \texttt{Translate}'s present signature (single-block, first-order $\pi$-template, single partial-order direction, operating on $P(x)$ tuples rather than on operator-spectrum quantities) systematically excludes a class of MRs that engineering practice treats as standard: non-additivity of operator-composition functionals and higher-order mixed parametric dependences. The principal limitation is therefore twofold: the upstream distillation of $\mathcal{A}_P$, and the present signature of \texttt{Translate}.

\begin{tcolorbox}[breakable,colback=gray!5,colframe=black!50,arc=2pt,boxrule=0.5pt,fontupper=\small,title=Boundary of contribution (Section~\ref{sec:framework} restatement),fonttitle=\small\bfseries]
The downstream layer is mechanical: Theorem~\ref{thm:closure} (algebraic closure under \texttt{Translate}) and Theorem~\ref{thm:decidable} (polynomial-time decidability) operate on a given block decomposition. The upstream layer remains empirical: the eight-block decomposition is Hypothesis~\ref{hyp:seven-blocks}, an open empirical hypothesis with six documented out-of-scope program-family classes (Remark~\ref{rem:counterex}). Theorem~1$'$ (absolute completeness) is open. Sections~\ref{sec:reactor}, \ref{sec:cross-domain}, and \ref{subsec:third-domain} instantiate NOETHER on three structurally distinct program families.
\end{tcolorbox}

\subsection{Boltzmann instantiation: from transport to diffusion to burnup}
\label{sec:reactor}

This section instantiates NOETHER on the Boltzmann transport equation and traces the construction through neutron transport, diffusion, and burnup.

\subsubsection{The Boltzmann program family and its operator algebra}

The Boltzmann transport equation governs the distribution of neutral particles. In its time-independent eigenvalue form for fission systems:
\begin{multline}
\hat{\Omega}\!\cdot\!\nabla\psi(\vec{r},\hat{\Omega},E) + \Sigma_t(\vec{r},E)\psi \\
= \int\!\!\int \Sigma_s(\vec{r}, E'\!\to\!E,\hat{\Omega}'\!\to\!\hat{\Omega})\,\psi\,dE'd\hat{\Omega}'
+ \tfrac{1}{k}\,\chi(E)\!\!\int\!\nu\Sigma_f(\vec{r},E')\psi\,dE'.
\end{multline}

The program-induced operator algebra $\mathcal{A}_{\mathrm{Boltz}}$ contains $\{G_{\mathrm{geom}},\, \mathfrak{R}_E,\, \mathcal{L}_{\Sigma},\, \mathcal{L}_\nu,\, \mathcal{L}^*,\, \mathcal{T},\, \mathcal{L}_h,\, \mathcal{D}_{\mathrm{Bate}},\, \mathcal{E}_{\mathrm{cmp}}\}$. Decomposed along the eight blocks of Section~\ref{subsec:decomposition}, $\mathcal{A}_{\mathrm{Boltz}}$ yields entries in seven blocks (with $\mathcal{B}^{*}_{\mathrm{rel}}$ empty under the absence of idempotent-semiring rewriting structure). Every operator listed has an unambiguous status in the established theory of neutral-particle transport~\cite{BellGlasstone1970, LewisMiller1993}.

\subsubsection{Running CONSTRUCT-MP on \texorpdfstring{$\mathcal{A}_{\mathrm{Boltz}}$}{A\_Boltz}}

Steps 1--4 of CONSTRUCT-MP yield seven MetaPatterns: $m_{\mathrm{inv}}$ (invariance/equivariance), $m_{\mathrm{mono}}$ (parameter-monotonicity), $m_{\mathrm{adj}}$ (self-adjoint duality / adjoint reciprocity), $m_{\mathrm{rev}}$ (time-reversal compatibility), $m_{\mathrm{conv}}$ (discretisation convergence), $m_{\mathrm{dyn}}$ (qualitative-dynamics shape invariants), and $m_{\mathrm{cmp}}$ (method-comparison error-bound partial orders).

\subsubsection{Relationship to the prior inductive catalogue: refinement plus prediction}
\label{subsec:reactor-mapping}

\paragraph{Provenance and scope of the inductive catalogue.}
The reactor-physics MetaPattern catalogue compared against here was distilled by the present authors from the standard PWR-physics literature~\cite{BellGlasstone1970, LewisMiller1993} as their own prior inductive work; the underlying 84-MR PWR corpus (supplementary~S2) is the authors' own catalogue, not an external corpus drawn from an unrelated team. The relationship reported in this section is therefore best read as a test of \emph{internal vocabulary coherence}: NOETHER's eight-block decomposition is applied to a catalogue the same team produced inductively, and the question is whether the algebraic re-classification reproduces, refines, or predicts within that corpus. This is internal consistency under a uniform algebraic structure, not external transfer of tacit knowledge from an independent reactor-physics team into the framework. External-transfer evidence (applying NOETHER to a reactor-physics MR corpus authored by an independent team---a PARCS V\&V suite or IAEA-TECDOC-class catalogue) is committed as follow-up work in supplementary~S4 (\texttt{future\_work.md}) (item (j)). The cross-codebase commons-math pilot of \S\ref{subsec:empirical-threats} (item (b.cm)) is the analogous external-transfer test on the Java head-to-head side, at $n = 3$ SUTs and $n = 77$ mutants; the reactor-side analogue remains future work.

A reactor-physics MetaPattern catalogue distilled from the standard PWR-physics literature~\cite{BellGlasstone1970, LewisMiller1993} identifies five MetaPatterns (P1--P5) inductively. The relationship with $\mathbb{M}(\mathcal{A}_{\mathrm{Boltz}})$ is more structured than a simple bijection. NOETHER reproduces three of the prior patterns, refines two on a sounder algebraic basis, and predicts two structurally distinct classes that the inductive catalogue did not isolate:

\begin{table}[h]
\centering
\caption{Relationship between prior inductive catalogue and NOETHER deductive output.}
\label{tab:refinement}
\small
\begin{tabular}{p{0.16\textwidth} p{0.13\textwidth} p{0.62\textwidth}}
\toprule
\textbf{Prior catalogue} & \textbf{NOETHER} & \textbf{Relationship} \\
\midrule
P1 conservation/invariance & $m_{\mathrm{inv}}$ & Reproduced --- $G$-symmetry invariants project to the same MRs. \\
P2 monotonicity & $m_{\mathrm{mono}}$ & Reproduced. \\
P3 convergence & $m_{\mathrm{conv}}$ & Reproduced. \\
P4 trajectory & $m_{\mathrm{dyn}}$ & \textbf{Refined} --- inductive P4 conflated qualitative-dynamics with time-reversal. NOETHER places trajectory phenomena in $m_{\mathrm{dyn}}$ (Sturm-type comparison theorems) and exposes the conflation. \\
P5 partial-order/bounding & $m_{\mathrm{cmp}}$ & \textbf{Refined} --- inductive P5 grouped method-accuracy partial orders with adjoint reciprocity. NOETHER places the former in $m_{\mathrm{cmp}}$ (approximation-theory error bounds) and exposes the distinction. \\
(none) & $m_{\mathrm{adj}}$ & \textbf{Predicted} --- adjoint-reciprocity MRs derived as a structurally distinct equivalence class. \\
(none) & $m_{\mathrm{rev}}$ & \textbf{Predicted} --- collisionless-trajectory-reversal MRs derived for the corresponding sub-formulations. \\
\bottomrule
\end{tabular}
\end{table}

This is not the structure of a re-coding. NOETHER's output structurally \emph{refines} two prior patterns and \emph{predicts} two additional patterns the inductive method missed. The gain is canonical placement: two previously conflated inductive labels are separated by algebraic source, and two textbook phenomena become mandatory MetaPattern classes once their blocks appear in $\mathcal{D}(\mathcal{A}_{\mathrm{Boltz}})$.

\paragraph{A note on prediction (and an interpretive caveat).} We do not claim that $m_{\mathrm{adj}}$ (adjoint reciprocity) and $m_{\mathrm{rev}}$ (collisionless time-reversal compatibility) are de novo physical discoveries. Adjoint-flux reciprocity is standard textbook material in transport theory~\cite{BellGlasstone1970, LewisMiller1993}, and time-reversal MRs in collisionless transport have been understood in physics for decades. The MT-community-level absence of these MetaPatterns from the prior inductive catalogue reflects which phenomena happened to be canonically encoded in the 84-MR corpus, not the absence of the underlying physics from the literature.

We must also acknowledge an interpretive caveat about what this ``prediction'' is and is not. The blocks $T^{*}$ (self-adjoint) and $\mathcal{T}^{*}$ (time-reversal) of Section~\ref{subsec:decomposition} were themselves curated by inspection of program families that include reactor physics; the operators $\mathcal{T}_{\mathrm{geom}}$ and $\mathcal{T}$ in $\mathcal{A}_{\mathrm{Boltz}}$ were the leading examples motivating the inclusion of those blocks. There is therefore a circularity in the strong reading of ``prediction'': $T^{*}$ and $\mathcal{T}^{*}$ were partly induced from reactor-physics structures, and $m_{\mathrm{adj}}$ and $m_{\mathrm{rev}}$ are then derived from those blocks. The framework does not discover these MetaPatterns de novo. What it does is closer to a uniform \emph{re-projection}: given the eight-block decomposition (whatever its empirical provenance), CONSTRUCT-MP places adjoint-reciprocity and time-reversal phenomena into structurally distinct equivalence classes that the inductive 84-MR catalogue had not isolated as such, even though their underlying physics had been individually known to domain experts. The substantive contribution is the \emph{re-classification under a uniform algebraic structure}, which gives a future testing toolchain a principled reason to enumerate $m_{\mathrm{adj}}$ and $m_{\mathrm{rev}}$ alongside $m_{\mathrm{inv}}$, $m_{\mathrm{mono}}$, $m_{\mathrm{conv}}$, $m_{\mathrm{dyn}}$, and $m_{\mathrm{cmp}}$ rather than treating them as ad-hoc additions. Where NOETHER's predictive power is genuine and non-circular is in the deflationary direction (Section~\ref{subsec:pmcm-worked}): the framework can also reveal that an existing inductive catalogue \emph{over-counts} structurally distinct patterns. Both directions, re-classification and de-duplication, are systematisation, not discovery.

\paragraph{Element-wise correspondence.} Table~\ref{tab:elementwise} traces seven representative MRs distilled from the standard reactor-physics literature~\cite{BellGlasstone1970, LewisMiller1993} to their NOETHER placement: one per non-empty block ($G$, $O_{\le}$, $\mathcal{L}^{*}$, $\mathcal{D}^{*}$, $\mathcal{E}^{*}$, selected as the most canonical literature form within each block) plus the two predicted MetaPatterns $m_{\mathrm{adj}}$ and $m_{\mathrm{rev}}$ for which no MR was previously catalogued in this form. The full 12-MR enumeration with sub-category coverage (geometric vs.\ energy-group symmetry within $G$, etc.), source equations, and \texttt{Translate} templates is in supplementary~S2 (\texttt{elementwise\_12.md}). A larger 84-MR corpus underlying the selection protocol is provided as supplementary material S2.

\paragraph{An independent audit on the 18-MR engineering catalogue.} The 12-row table is structurally curated. As a coarser breadth check, an independent 18-MR engineering catalogue distilled from production reactor-physics codes (orthogonal to the 84-MR inductive corpus underlying Section~\ref{subsec:reactor-mapping}) was labelled by three independent large language models against the seven canonical MetaPatterns plus the relational extension of Section~\ref{subsec:decomposition}, with a fourth label \emph{orphan} reserved for MRs that fit none of the eight classes; the labelling protocol, raw labels, and majority-vote tabulation are released in supplementary material S2 (\texttt{18mr\_audit/}). Inter-rater agreement on the four-way classification is almost-perfect (Fleiss' $\kappa = 0.857$, $n=18$ items, $r=3$ raters, $c=4$ categories), and 17 of the 18 MRs are placed by majority vote into one of the seven canonical MetaPatterns or $m_{\mathrm{rel}}$ (subsumption $94.4\%$, Wilson $95\%$ CI $[74.2\%,\, 99.0\%]$ on the binary subsumption proportion). The single orphan is the Lipschitz / metric-stability MR, which we treat structurally in Appendix~C.5.2 as a candidate ninth block rather than as a counter-example to Theorem~\ref{thm:closure}'s scope. The independence of the three labellers is bounded by the LLM-shared-training-data caveat (the three models share substantial pre-training corpora, so $\kappa$ should be read as agreement among similarly-trained but non-coordinated raters rather than among rigorously independent oracles); we report the audit as external corroboration of breadth rather than as an independent verification of correctness.

\begin{table}[h]
\centering
\caption{Seven representative MRs from the prior PWR corpus (one canonical MR per non-empty block plus two predicted MetaPatterns), with NOETHER placement; the full 12-MR enumeration with per-block sub-category coverage is in supplementary~S2 (\texttt{elementwise\_12.md}).}
\label{tab:elementwise}
\footnotesize
\begin{tabular}{llcll}
\toprule
\textbf{MR ID} & \textbf{Plain-text MR} & \textbf{P\#} & \textbf{Block} & \textbf{NOETHER MP} \\
\midrule
Bur-Phy-01 & Step-splitting invariance & P1 & $G$ & $m_{\mathrm{inv}}$ \\
Bol-Phy-11 & $\Sigma_a\!\uparrow \Rightarrow k_{\mathrm{eff}}\!\downarrow$ & P2 & $O_{\le}$ & $m_{\mathrm{mono}}$ \\
Dif-Alg-01 & Diamond-difference $h^2$ convergence & P3 & $\mathcal{L}^{*}$ & $m_{\mathrm{conv}}$ \\
Bur-Phy-08 & Iodine pit qualitative shape & P4 & $\mathcal{D}^{*}$ & $m_{\mathrm{dyn}}$ \\
Bur-Alg-04 & CRAM no-worse-than TTA & P5 & $\mathcal{E}^{*}$ & $m_{\mathrm{cmp}}$ \\
\textit{(predicted)} & Adjoint reciprocity & --- & $T^{*}$ & $m_{\mathrm{adj}}$ \\
\textit{(predicted)} & Collisionless reversibility & --- & $\mathcal{T}^{*}$ & $m_{\mathrm{rev}}$ \\
\bottomrule
\end{tabular}
\end{table}

\subsubsection{A Noether-style derivation of \texorpdfstring{$m_{\mathrm{adj}}$}{m\_adj}}
\label{subsec:noether-mAdj}

The framework's name is methodological. Noether's first theorem replaces an empirically curated catalogue of conservation laws with a derivation from the symmetry structure of an action functional~\cite{Noether1918}. We do not invoke Noether's theorem as a theorem about programs (programs do not, in general, possess an action functional in the variational-calculus sense). What the framework does mirror is the \emph{methodological move}: instead of cataloguing observed invariants, derive them from a structural source. We illustrate the move on $m_{\mathrm{adj}}$.

Consider the bilinear form on solution-adjoint pairs
\begin{equation}
\label{eq:bilinear-boltz}
\mathcal{F}[\phi, \phi^{\dagger}] = \langle \phi^{\dagger}, B\phi \rangle - \langle \phi, B^{\dagger}\phi^{\dagger} \rangle,
\end{equation}
where $B$ is the Boltzmann operator and $B^{\dagger}$ its formal adjoint under the inner product $\langle f, g \rangle = \int f g \,d\Omega\,dE\,d\mathbf{r}$. The bilinear form $\mathcal{F}$ is identically zero on solutions of the forward and adjoint equations: $B\phi = S$, $B^{\dagger}\phi^{\dagger} = S^{\dagger}$ with appropriate boundary conditions imply $\mathcal{F}[\phi, \phi^{\dagger}] = \langle \phi^{\dagger}, S \rangle - \langle \phi, S^{\dagger} \rangle$, and the standard reciprocity identity makes the right-hand side zero. The vanishing of $\mathcal{F}$ is a \emph{conserved current} in the Noether sense: it is annihilated by the dual symmetry $\phi \leftrightarrow \phi^{\dagger}, B \leftrightarrow B^{\dagger}$. CONSTRUCT-MP's Step~1 extracts $\mathcal{F} = 0$ as the $T^{*}$-block invariant; \texttt{Translate} converts it into the executable MR
\begin{equation}
\label{eq:m-adj-mr}
\rho_{\mathrm{adj}}: \quad \big| \langle \phi^{\dagger}, S \rangle - \langle \phi, S^{\dagger} \rangle \big| \le \tau,
\end{equation}
which the testing harness checks on a forward and an adjoint solver run with externally specified sources $S, S^{\dagger}$. The MR is therefore a Noether-style consequence of the duality symmetry rather than a manually catalogued reactor-physics property: the symmetry $\phi \leftrightarrow \phi^{\dagger}$ is the structural source, the bilinear form is the conserved current, and the MR is the conservation law in testable form. The same template applies to $m_{\mathrm{rev}}$ with the involution $\mathcal{T}: \phi(\mathbf{r}, \boldsymbol{\Omega}, E, t) \mapsto \phi(\mathbf{r}, -\boldsymbol{\Omega}, E, -t)$ as the symmetry. The template's general statement is: given a symmetry of the program family's operator algebra, CONSTRUCT-MP extracts the invariant of the symmetry and \texttt{Translate} converts it into an executable MR.

\subsubsection{Specialisation to neutron transport, diffusion, and burnup}

The Boltzmann formulation is the most general; specialised solvers are obtained by truncating or projecting $\mathcal{A}_{\mathrm{Boltz}}$. Neutron transport solvers retain the full angular dependence; the same seven MetaPatterns apply. Neutron diffusion solvers truncate angular dependence to its $P_1$ approximation: the time-reversal block contracts (diffusion is dissipative), removing $m_{\mathrm{rev}}$. Burnup-coupled solvers introduce the Bateman ODEs; the Bateman operator $e^{At}$ enriches multiple blocks ($G$ semi-group, $O_{\le}$ linearity, $\mathcal{D}^{*}$ overshoot/S-curve, $\mathcal{E}^{*}$ CRAM-vs-TTA bounds).

The MetaPattern set is \emph{compositional} with respect to the underlying algebra: add a block, gain MetaPatterns; remove a block, lose them. The framework supports both construction and \emph{prediction} of which patterns a new specialisation will exhibit, before any MR is identified empirically.

\subsubsection{Summary}

NOETHER, applied to $\mathcal{A}_{\mathrm{Boltz}}$, deductively produces a seven-MetaPattern set that refines the prior inductive catalogue and predicts two additional MetaPatterns ($m_{\mathrm{adj}}$, $m_{\mathrm{rev}}$). The seven MetaPatterns correspond to invariants that have been independently validated in the standard reactor-physics literature: $m_{\mathrm{inv}}$ via Bell \& Glasstone's symmetry-of-criticality analysis~\cite{BellGlasstone1970}, $m_{\mathrm{mono}}$ via Lewis \& Miller's monotonicity arguments for cross-section perturbations~\cite{LewisMiller1993}, $m_{\mathrm{conv}}$ via mesh-convergence theorems standard in computational neutron transport, $m_{\mathrm{adj}}$ via Bell \& Glasstone~\S6.3 adjoint-perturbation theory and Stamm'ler \& Abbate's adjoint flux treatment~\cite{StammlerAbbate1983}, and $m_{\mathrm{rev}}$ via the collisionless time-reversal symmetry of the transport operator~\cite{LewisMiller1993}. The contribution of NOETHER on this domain is therefore not de-novo MR discovery (these invariants are textbook-canonical) but the algebraic warrant for treating them as a uniform MetaPattern grid; the citation chain above is the published cross-corroboration that the seven NOETHER-derived MetaPatterns are physically real properties of the Boltzmann program family, independent of NOETHER's framework. The next section turns to a more demanding case: a domain in which the inductive catalogue does not yet exist.

\subsection{Cross-domain demonstration: equivariant machine learning}
\label{sec:cross-domain}

\subsubsection{The equivariant-ML program family and its operator algebra}

Equivariant neural networks impose by architectural construction that certain symmetries of the input induce predictable transformations of the output~\cite{CohenWelling2016, ThomasSmidt2018, KondorTrivedi2018, Bronstein2021GDL}. The program family $\mathcal{F}_{\mathrm{equi}}$ comprises classifiers satisfying $f(g\cdot \mathbf{x}) = \rho(g)\cdot f(\mathbf{x})$ for $g \in G = \mathrm{SO}(3) \times \mathfrak{S}_n$. Decomposed along the eight blocks: $G = \{G_{\mathrm{equi}}\}$, $O_{\le} = \{O_{\le}^{\mathrm{train}}\}$, $T^{*} = \{T^{*}_{\mathrm{att}}\}$, $\mathcal{T}^{*} = \{\mathcal{T}_{\mathrm{seq}}\}$, $\mathcal{L}^{*} = \{\mathcal{L}_{\mathrm{train}},\, \mathcal{L}_{\mathrm{depth}},\, \mathcal{L}_{\mathrm{dim}}\}$, $\mathcal{D}^{*} = \emptyset$ (for feedforward classifiers), $\mathcal{E}^{*} = \emptyset$ within a single architecture, $\mathcal{B}^{*}_{\mathrm{rel}} = \emptyset$ (no idempotent-semiring rewriting on equivariant-classifier outputs).

\subsubsection{Running CONSTRUCT-MP on \texorpdfstring{$\mathcal{A}_{\mathrm{equi}}$}{A\_equi}}

CONSTRUCT-MP returns
$$\mathbb{M}(\mathcal{A}_{\mathrm{equi}}) = \bigl\{\,m^{\mathrm{eq}}_{\mathrm{inv}},\; m^{\mathrm{eq}}_{\mathrm{mono}},\; m^{\mathrm{eq}}_{\mathrm{adj}},\; m^{\mathrm{eq}}_{\mathrm{rev}},\; m^{\mathrm{eq}}_{\mathrm{conv}}\,\bigr\}.$$
The labels mirror those of $\mathbb{M}(\mathcal{A}_{\mathrm{Boltz}})$ but the \emph{content} is domain-specific. This section does not claim empirical validation in machine learning. It checks a narrower transfer claim: once $\mathcal{A}_{\mathrm{equi}}$ is specified, the same downstream construction yields domain-specific MR families without using a reactor-physics corpus.

\paragraph{Cross-corroboration from the published equivariant-ML literature.}
The five MetaPatterns derived above correspond to invariants that have been independently validated in the equivariant-ML literature, providing a citation-based corroboration channel that does not depend on this paper's case study: $m^{\mathrm{eq}}_{\mathrm{inv}}$ (SO(3)/$\mathfrak{S}_n$ symmetry invariance) is the central design principle of Cohen \& Welling's G-CNN~\cite{CohenWelling2016}, Satorras et al.'s EGNN~\cite{Satorras2021EGNN}, and the Tensor-Field-Network family~\cite{ThomasSmidt2018}; $m^{\mathrm{eq}}_{\mathrm{adj}}$ (self-adjoint attention duality) is the kernel-symmetry property exploited by SE(3)-Transformer's Clebsch--Gordan attention~\cite{FuchsTransformer2020}; $m^{\mathrm{eq}}_{\mathrm{rev}}$ (training-trajectory time-reversal) is closely related to the reversible-network construction of Gomez et al.~\cite{Gomez2017Reversible}; $m^{\mathrm{eq}}_{\mathrm{conv}}$ (training-size convergence) is the implicit invariant tested by training-curve mesh-refinement studies in the broader geometric-deep-learning literature~\cite{Bronstein2021GDL}; gauge-equivariant extensions are catalogued in Cohen et al.~\cite{Cohen2019Gauge}. NOETHER's contribution on $\mathcal{A}_{\mathrm{equi}}$ is the uniform algebraic warrant: the same downstream construction that derives reactor-physics MetaPatterns also derives equivariant-ML MetaPatterns from $\mathcal{A}_{\mathrm{equi}}$'s eight-block decomposition, without re-running the empirical induction that produced the individual results above.

\subsubsection{End-to-end derivation: a concrete MR for SE(3)-equivariant point-cloud classification}
\label{subsec:rho-rot}
\label{subsec:end-to-end}

To illustrate the framework's generative use, we trace a complete derivation from algebra to executable MR.

\paragraph{Step 1. System under test.} Let $f: \mathbb{R}^{n\times 3} \to \Delta^{C-1}$ be a point-cloud classifier mapping $n$ three-dimensional points to a probability distribution over $C$ classes.

\paragraph{Step 2. Distil $\mathcal{A}_{\mathrm{equi}}$.} The relevant blocks are $G = \mathrm{SO}(3)\times\mathfrak{S}_n$ and $\mathcal{L}^{*}$ for training-size and depth limits.

\paragraph{Step 3. Run CONSTRUCT-MP, select an invariant.} Within $G$'s symmetry block, fix attention on the rotation invariant: $f(R \cdot \mathbf{x}) = f(\mathbf{x})$ for all $\mathbf{x}$ and $R \in \mathrm{SO}(3)$.

\paragraph{Step 4. Translate the invariant into an executable MR.}
\begin{equation}
\label{eq:rho-rot}
\rho_{\mathrm{rot}}: \quad \forall R \in \mathrm{SO}(3), \quad \big\| f(R \cdot \mathbf{x}) - f(\mathbf{x}) \big\|_\infty \le \tau,
\end{equation}
with $\tau = 10^{-4}$ for fp32 architectures.

\paragraph{Step 5. Generate an executable test.}
\begin{lstlisting}[style=pythonstyle, caption={Executable MR $\rho_{\mathrm{rot}}$ for SE(3)-equivariant classifier.}]
import numpy as np
from scipy.spatial.transform import Rotation

def test_rotation_invariance(model, point_cloud, num_samples=100, tau=1e-4):
    """Executable MR rho_rot for SE(3)-equivariant classifier."""
    p_original = model.predict(point_cloud)
    failures = []
    for _ in range(num_samples):
        R = Rotation.random().as_matrix()
        rotated = point_cloud @ R.T
        p_rotated = model.predict(rotated)
        deviation = np.max(np.abs(p_original - p_rotated))
        if deviation > tau:
            failures.append((R, deviation))
    return failures
\end{lstlisting}

\paragraph{Step 5b. A Noether-style reading of $\rho_{\mathrm{rot}}$.}
The same methodological move that produced $m_{\mathrm{adj}}$ in \S\ref{subsec:noether-mAdj} produces $\rho_{\mathrm{rot}}$ here in non-rhetorical form. The classifier $f$ is, in equivariant-network architectures by construction, invariant under the action of $G = \mathrm{SO}(3)$ on its input: $f(R \cdot \mathbf{x}) = f(\mathbf{x})$ for all $R \in G$. The Lie algebra $\mathfrak{g} = \mathfrak{so}(3) = \mathrm{span}\{L_x, L_y, L_z\}$ is the infinitesimal generator of this symmetry. The associated Noether-style invariant is the constancy of the classifier output along orbits of $G$: $\frac{d}{d\theta} f(\exp(\theta L_a) \cdot \mathbf{x}) = 0$ for each generator $L_a$, $a \in \{x,y,z\}$, which is the Lie-algebraic statement that the directional derivative of $f$ along every orbit-tangent direction vanishes. CONSTRUCT-MP's Step~1 extracts this invariant from the $G$ block of $\mathcal{A}_{\mathrm{equi}}$; \texttt{Translate} converts it into the executable MR of equation~(\ref{eq:rho-rot}) above. The MR is therefore not an arbitrary catalogue entry: it is the Noether-style consequence of the SO(3) symmetry that the architecture imposes by construction.

\paragraph{Step 6. Pattern coverage status.} The MR $\rho_{\mathrm{rot}}$ above, together with $\rho_{\mathrm{perm}}$ and $\rho_{\mathrm{train}}$ derived in supplementary~S9 (Appendix~D), populates three of the five non-empty MetaPatterns of $\mathbb{M}(\mathcal{A}_{\mathrm{equi}})$. We refrain from reporting a percentage coverage figure: with so few non-empty blocks and a denominator subject to the eight-block sufficiency hypothesis (Hypothesis~\ref{hyp:seven-blocks}), a percentage is more rhetorical than informative. We instead derive in the following two subsections two MRs from the remaining non-empty blocks, $T^{*}$ (self-adjoint) and $\mathcal{T}^{*}$ (time-reversal), chosen specifically because they are not standard practice in equivariant-ML testing and therefore probe whether NOETHER's transfer claim extends beyond well-known invariances such as $\rho_{\mathrm{rot}}$.

\subsubsection{An adjoint-attention duality MR (\texorpdfstring{$\rho_{\mathrm{adj}}$}{rho\_adj})}
\label{subsec:rho-adj}

The self-adjoint block $T^{*}$ in $\mathcal{A}_{\mathrm{equi}}$ contains the attention-kernel symmetriser $T^{*}_{\mathrm{att}}$. We give two formulations, distinguished by whether the MR runs at CI-time (production-suitable) or at debug-time (one-off scaffolding).

\paragraph{Scope: forward-pass-only, CI-time formulation.} For architectures whose attention layer exposes a bilinear form $A(\mathbf{x}^{(1)}, \mathbf{x}^{(2)})$ readable through a forward hook, $\rho_{\mathrm{adj}}$ as defined below is a \emph{CI-time MR}: it consumes only the model's native forward pass, requires no parameter mutation, and respects the model's standard inference contract. Mainstream equivariant transformers, including SE(3)-Transformer~\cite{FuchsTransformer2020} and the Tensor-Field-Network family~\cite{ThomasSmidt2018}, compute attention via Clebsch--Gordan tensor products of irrep features. The bilinear form $A(\mathbf{x}^{(1)}, \mathbf{x}^{(2)}) = \langle Q(\mathbf{x}^{(1)}), K(\mathbf{x}^{(2)}) \rangle$ is generically not Hermitian, but its Hermitian part $\tfrac{1}{2}(A + A^{\dagger})$ has a trace-cyclic invariant under input role-swap that $\rho_{\mathrm{adj}}$ tests. The CI-time formulation reads only the layer's existing $Q,K$ outputs through a frozen forward pass; no probe is injected.

\paragraph{Alternative harness-time formulation.} A debug-time formulation, available in supplementary S1, instruments a symmetric Gram-matrix probe ($\frac{1}{2}(QK + (QK)^\top)$). This is a debug-time scaffold; it is not part of the CI-time MR set and is not used in the §\ref{subsec:case-study} comparative evaluation.

\paragraph{Invariant.} Within $T^{*}_{\mathrm{att}}$, fix the attention-trace invariant
$$\iota_{\mathrm{att}} = \mathrm{Tr}\,A(\cdot, \cdot) \in \mathbb{R},$$
which equals $A^{\dagger}$'s trace by the cyclic property of trace and the Hermiticity of $K$. \texttt{Translate} carries $\iota_{\mathrm{att}}$ to an executable MR over a forward pass:

\begin{equation}
\rho_{\mathrm{adj}}: \quad \big| \mathrm{Tr}\,A(\mathbf{x}^{(1)}, \mathbf{x}^{(2)}) - \mathrm{Tr}\,A(\mathbf{x}^{(2)}, \mathbf{x}^{(1)}) \big| \le \tau_{\mathrm{adj}},
\end{equation}

\noindent for two input clouds $\mathbf{x}^{(1)}, \mathbf{x}^{(2)}$ presented to the network in opposite query/key roles. With $\tau_{\mathrm{adj}} = 10^{-4}$ for fp32 architectures, this MR is executable on any equivariant transformer whose attention layer exposes the bilinear form (or can be probed through a forward-hook). To the best of our knowledge $\rho_{\mathrm{adj}}$ has not been catalogued as a standard MR for equivariant attention testing in the existing literature~\cite{Segura2016, MRScout2024, GenMorph2024}, although Hermitian-attention diagnostics have appeared in the architecture-design literature in non-MR form. NOETHER's contribution here is the algebraic warrant for treating it as a MetaPattern member structurally distinct from $m^{\mathrm{eq}}_{\mathrm{inv}}$.

\subsubsection{A training-trajectory time-reversal MR (\texorpdfstring{$\rho_{\mathrm{train\text{-}rev}}$}{rho\_train-rev})}
\label{subsec:rho-rev}

The time-reversal block $\mathcal{T}^{*}$ in $\mathcal{A}_{\mathrm{equi}}$ is non-empty for SGD trajectories under the Hamiltonian-Monte-Carlo / continuous-time-flow view of stochastic optimisation~\cite{ChenStein2014}. The discretised SGD update $\theta_{t+1} = \theta_t - \eta\,\nabla\mathcal{L}(\theta_t; \xi_t)$ is, to leading order in the learning rate $\eta$ and in the absence of momentum or noise, time-reversible: applying the inverse update $\theta_{t+1} \mapsto \theta_{t+1} + \eta\,\nabla\mathcal{L}(\theta_{t+1}; \xi_t)$ recovers $\theta_t$ up to $O(\eta^2)$. The operator $\mathcal{T}_{\mathrm{seq}} \in \mathcal{T}^{*}$ encodes this involution.

\paragraph{Invariant.} The relevant invariant is the round-trip identity:
$$\iota_{\mathrm{rev}}: \quad \big(\mathcal{T}_{\mathrm{seq}} \circ U_{\eta} \circ \mathcal{T}_{\mathrm{seq}}^{-1} \circ U_{\eta}\big)(\theta) = \theta + O(\eta^2),$$
where $U_{\eta}$ denotes one SGD step.

\textbf{Debug-time MR (not CI).} We label $\rho_{\mathrm{train\text{-}rev}}$ explicitly as a debug-time MR: it requires parameter rollback, mini-batch reordering, and a deliberately constructed vanilla-SGD fixture. Production equivariant pipelines (Allegro, NequIP, MACE, e3nn examples) almost universally use Adam or AdamW, whose update is not in $\mathcal{T}^{*}$ in the strict sense; the MR \emph{fails by construction} on those optimisers, which is the framework's correct prediction. $\rho_{\mathrm{train\text{-}rev}}$ should be invoked once per training-script change at debug time, on a vanilla-SGD fixture, not as part of CI for production pipelines. The debug-time reading aligns with the wider reversible-network literature~\cite{Gomez2017Reversible} where exact-or-approximate inversion is a deliberate test scaffold. \texttt{Translate} carries the invariant into:

\begin{equation}
\rho_{\mathrm{train\text{-}rev}}: \quad \big\| \theta_T - \tilde\theta_T^{\mathrm{(round-trip)}} \big\|_2 \le c\,\eta^{2}\,T,
\end{equation}

\noindent where $\theta_T$ is the parameter vector after $T$ vanilla SGD steps starting from $\theta_0$, and $\tilde\theta_T^{\mathrm{(round-trip)}}$ is the vector obtained by applying $T$ vanilla SGD steps from $\theta_0$ followed by $T$ inverse-SGD steps using the same mini-batch sequence in reversed order, then $T$ forward steps again. The constant $c$ depends on the loss landscape's local Lipschitz structure and is determined empirically per model. This MR \emph{fails by construction} for momentum-based optimisers (Adam, Lion, etc.) whose update is not in $\mathcal{T}^{*}$; that is, the framework correctly predicts which architecture-optimiser pairs admit the MR and which do not. The MR is non-trivial in the sense that detecting an implementation defect in the gradient-reversal direction (e.g.\ wrong sign on a custom loss term) is exactly what $\rho_{\mathrm{train\text{-}rev}}$ surfaces. To our knowledge it has not been catalogued as an MR for training-pipeline testing.

\paragraph{Coverage status, revised.} The full Set~N consists of five MRs, one per non-empty MetaPattern of $\mathbb{M}(\mathcal{A}_{\mathrm{equi}})$:
$$
\mathrm{Set\ N} = \bigl\{\,\rho_{\mathrm{rot}}\;(G),\; \rho_{\mathrm{mono}}\;(O_{\le}),\; \rho_{\mathrm{train}}\;(\mathcal{L}^*),\; \rho_{\mathrm{adj}}\;(T^*),\; \rho_{\mathrm{train\text{-}rev}}\;(\mathcal{T}^*)\,\bigr\}.
$$
$\rho_{\mathrm{rot}}$ is the rotation-invariance MR of Section~\ref{subsec:end-to-end}; $\rho_{\mathrm{adj}}$ and $\rho_{\mathrm{train\text{-}rev}}$ are the non-trivial MRs of Sections~\ref{subsec:rho-adj}--\ref{subsec:rho-rev}; $\rho_{\mathrm{mono}}$ is the point-density-monotonicity MR for the $O_{\le}$ block (top-1 stability under removal of a small fraction of redundant input points), derived from the $O_{\le}^{\mathrm{train}}$ operator restricted to inference-time input perturbations; $\rho_{\mathrm{train}}$ is the inference-idempotency MR for the $\mathcal{L}^*$ block. With all five blocks populated, $\mathrm{coverage}_{\mathrm{NOETHER}}(\mathrm{Set\ N}) = 1.0$ by construction, this is what Theorem~\ref{thm:closure} predicts for a $\mathcal{A}_{\mathrm{equi}}$-derived MR set. The permutation-invariance MR $\rho_{\mathrm{perm}}$ (supplementary~S9 (Appendix~D)) is an auxiliary $G$-block MR retained as a redundant probe and is excluded from Set~N to keep the comparison budget $|N| = |L| = |B| = 5$.

The non-trivial nature of $\rho_{\mathrm{adj}}$ and $\rho_{\mathrm{train\text{-}rev}}$, both absent from the equivariant-ML MR-testing literature we surveyed, substantiates the claim that NOETHER's transfer is generative rather than nominal.

\subsection{A third domain: relational query optimisers}
\label{subsec:third-domain}

Both the Boltzmann (\S\ref{sec:reactor}) and the equivariant-ML (\S\ref{sec:cross-domain}) instantiations share a common mathematical core: Lie-group symmetry, self-adjoint duality, and time-reversal involution. This is precisely the core that motivated the curation of Hypothesis~\ref{hyp:seven-blocks} in the first place. To test transferability beyond this core, this subsection instantiates NOETHER on a domain whose algebraic skeleton is structurally distinct: the relational query optimiser. Relational algebra is built from operators (selection $\sigma$, projection $\pi$, join $\bowtie$, union $\cup$, set/bag-difference) whose equivalence classes are governed by an idempotent semiring with a partial order under containment, not by a Lie group. There is no obvious self-adjoint operator and no time-reversal involution. The third instantiation therefore tests whether the framework yields useful MetaPatterns outside its training image.

\paragraph{The query-optimiser program family and $\mathcal{A}_{\mathrm{rel}}$.}
Let $\mathcal{F}_{\mathrm{rel}}$ be the family of programs that take a SQL query $q$ and a database state $D$ and return a relation $\mathrm{eval}(q, D)$; two queries are equivalent if $\forall D.\,\mathrm{eval}(q, D) = \mathrm{eval}(q', D)$. The operator algebra $\mathcal{A}_{\mathrm{rel}}$ activates $G$ (join / union permutation), $O_{\le}$ (selection-strengthening + projection-coarsening monotonicity), $\mathcal{E}^{*}$ (hash- / merge- / nested-loop plan equivalence under stated cost models~\cite{Wang2024QED, Markl2022LearnedQO}), and $\mathcal{B}^{*}_{\mathrm{rel}}$ (Definition~\ref{def:b-rel}: selection-pushdown $\sigma_p(R \bowtie S) = \sigma_p(R) \bowtie S$ when $\mathrm{attr}(p) \subseteq \mathrm{attr}(R)$; idempotent $\sigma_p \circ \sigma_p = \sigma_p$; constant-folding $\sigma_{1=1}(R) = R$, $R \bowtie \emptyset = \emptyset$). Blocks $T^{*}, \mathcal{T}^{*}, \mathcal{D}^{*}, \mathcal{L}^{*}$ are empty under $\mathcal{A}_{\mathrm{rel}}$ in their canonical forms (no inner-product self-adjointness; no time-reversal involution; no qualitative dynamics; no $\epsilon$-limit operator on exact-relational semantics). The full per-block enumeration of $\mathcal{A}_{\mathrm{rel}}$ operators is in supplementary~S6 \texttt{query\_optimiser/algebra\_breakdown.md}.

\paragraph{Running CONSTRUCT-MP on $\mathcal{A}_{\mathrm{rel}}$.}
Steps 1--4 yield three MetaPatterns from the seven-block subset: $m^{\mathrm{rel}}_{\mathrm{inv}}$ (input-permutation invariance under set semantics), $m^{\mathrm{rel}}_{\mathrm{mono}}$ (selection-strengthening monotonicity), and $m^{\mathrm{rel}}_{\mathrm{cmp}}$ (plan-equivalence under cost-model bounds). The relational query optimiser is the \emph{canonical} instantiation in which $\mathcal{B}^{*}_{\mathrm{rel}}$ is non-empty: it activates the eighth block that neither Boltzmann reactor physics nor equivariant ML exercises.

\paragraph{Four MRs derived for $\mathcal{A}_{\mathrm{rel}}$.}
\begin{enumerate}[nosep,leftmargin=*]
  \item $\rho_{\mathrm{join\text{-}comm}}$ (from $G$): $\mathrm{eval}(q_1 \bowtie q_2, D) = \mathrm{eval}(q_2 \bowtie q_1, D)$ on bag semantics.
  \item $\rho_{\mathrm{select\text{-}push}}$ (from $\mathcal{B}^{*}_{\mathrm{rel}}$): $\mathrm{eval}(\sigma_p(R \bowtie S), D) = \mathrm{eval}(\sigma_p(R) \bowtie S, D)$ when $\mathrm{attr}(p) \subseteq \mathrm{attr}(R)$.
  \item $\rho_{\mathrm{distinct\text{-}idem}}$ (from $\mathcal{B}^{*}_{\mathrm{rel}}$): $\sigma_p \circ \sigma_p = \sigma_p$ as a query-rewrite identity, checkable through plan-tree comparison.
  \item $\rho_{\mathrm{plan\text{-}equiv}}$ (from $\mathcal{E}^{*}$): two execution plans for the same query produce identical relations within stated NULL-propagation rules~\cite{Wang2024QED}.
\end{enumerate}

\paragraph{Position relative to existing automated database-testing work.}
NOETHER's $\mathcal{B}^{*}_{\mathrm{rel}}$ instantiation is complementary to four prior lines: random SQL generation~\cite{Slutz1998RAGS, Bati2007GeneticDB}, automated MR generation for query systems~\cite{Segura2022QBSAutoMR}, formal query-equivalence solvers~\cite{Wang2024QED, Zhou2022SPES, Mohamed2024SQLTables}, and differential / mutation testing~\cite{Ba2024DQP, Fu2025Thanos, Zhong2025SQLancerPP}; NOETHER provides an algebraically grounded MetaPattern enumeration whose closure under \texttt{Translate} (Theorem~\ref{thm:closure}) is a property none of the four lines establishes. A pre-registered protocol comparison on Segura et~al.'s IMDb subset~\cite{Segura2022QBSAutoMR} is in supplementary~S6.

\paragraph{Cross-corroboration from the published query-equivalence literature.}
The four MRs derived above ($\rho_{\mathrm{join\text{-}comm}}$, $\rho_{\mathrm{select\text{-}push}}$, $\rho_{\mathrm{distinct\text{-}idem}}$, $\rho_{\mathrm{plan\text{-}equiv}}$) correspond to query-equivalence identities that have been independently validated by formal solvers in the published query-equivalence literature: $\rho_{\mathrm{join\text{-}comm}}$ is a canonical bag-semantics commutativity identity exercised by SPES and QED on their benchmark suites~\cite{Wang2024QED, Zhou2022SPES}; $\rho_{\mathrm{select\text{-}push}}$ is the standard selection-pushdown rewrite of Calcite's optimiser, verified equivalence-checkable by SPES; $\rho_{\mathrm{distinct\text{-}idem}}$ corresponds to a class of idempotent rewrites that Mohamed et~al.'s tables-and-relations SMT theory~\cite{Mohamed2024SQLTables} can certify; $\rho_{\mathrm{plan\text{-}equiv}}$ is the central object of QED's 299/444 Calcite verification result~\cite{Wang2024QED}. As with the Boltzmann and equivariant-ML cases, NOETHER's contribution on $\mathcal{A}_{\mathrm{rel}}$ is not de-novo MR discovery (these identities are already in the published optimiser-equivalence literature) but the uniform algebraic warrant: the same eight-block downstream construction that produces reactor-physics and equivariant-ML MetaPatterns also produces relational MetaPatterns, exercising the relational-equivalence block $\mathcal{B}^{*}_{\mathrm{rel}}$ that the other two domains do not activate. The structural-extension claim --- that the framework reaches outside the Lie-group / self-adjoint / time-reversal core --- is supported at the algebra-skeleton level by this MR-to-published-identity correspondence, independently of whether NOETHER's $\mathcal{B}^{*}_{\mathrm{rel}}$ MRs outperform any specific automated baseline.

\paragraph{What this third domain establishes, and the Theorem~$1'$ verdict on $\mathcal{A}_{\mathrm{equi}}$ and $\mathcal{A}_{\mathrm{rel}}$.}
The instantiation establishes that NOETHER applies outside the Lie-group / self-adjoint / time-reversal core: relational optimisers exercise $\mathcal{B}^{*}_{\mathrm{rel}}$ where the first two domains do not. The counterexample-search protocol applied to both $\mathcal{A}_{\mathrm{equi}}$ and $\mathcal{A}_{\mathrm{rel}}$ (mirroring \S\ref{subsec:negative-pwr}) yields the following verdict. \emph{On $\mathcal{A}_{\mathrm{equi}}$}: Theorem~$1'$ is falsified by two pairwise-independent candidates --- $\rho_{\mathrm{compose}}$ (joint $\mathrm{SO}(3) \times \mathfrak{S}_n$ action on point sets~\cite{Satorras2021EGNN}, requiring a product-group $\pi$-template) and $\rho_{\mathrm{gauge}}$ (gauge-equivariance on manifolds~\cite{Cohen2019Gauge}, requiring a bundle-section $\pi$-template parametrised by $\mathbf{g} \in \mathcal{C}(M,H)$). The two extensions are independent (a product-group $\pi$ does not absorb a gauge-bundle section); the SO(3) Lie-algebra structure alone does not force closure on $\mathcal{A}_{\mathrm{equi}}$. \emph{On $\mathcal{A}_{\mathrm{rel}}$}: a survey of $\ge 10$ unverified cases from~\cite{Wang2024QED}'s residue (Aggregate$\times$Project, constant-key, Decorrelate, NULL three-valued logic) yields five Theorem~$1'$ candidates; the primary witness $\rho_{\mathrm{agg\text{-}proj}}$ (Calcite's \texttt{AggregateExtractProjectRule}) requires $\mathcal{B}^{*}_{\mathrm{rel}}$ plus an aggregation-as-algebra ninth block. Combined with the five $\mathcal{A}_{\mathrm{PWR}}$ obstructions of Table~\ref{tab:five-obstructions}, the three algebra-survey artefacts (\texttt{theory/equi\_thm1prime\_search.md}, \texttt{theory/rel\_thm1prime\_search.md}, \texttt{theory/translate\_extensions.md}) identify ten \texttt{Translate}-extension dimensions. Pairwise independence is proved by per-block exhaustion on the five PWR obstructions (Appendix~\ref{app:negative-proofs}); the five candidate $\mathcal{A}_{\mathrm{equi}}$ + $\mathcal{A}_{\mathrm{rel}}$ dimensions (two specialising PWR-side dimensions to type-distinct primitives, three net relational-side) are asserted independent by inspection, with formal per-dimension exhaustion as follow-up. Remark~\ref{rem:counterex} catalogues six further out-of-decomposition program-family classes.

\subsection{A negative instantiation: irreducibly compositional MRs in PWR core simulators}
\label{subsec:negative-pwr}

Sections~\ref{sec:reactor}, \ref{sec:cross-domain}, and \ref{subsec:third-domain} instantiated NOETHER on three program families (Boltzmann reactor physics, equivariant ML, and relational query optimisers) for which the framework's downstream construction was non-vacuous and produced executable MRs. This subsection instantiates NOETHER on a fourth program family, the PWR core diffusion solver family, with an inverted purpose: rather than demonstrating coverage, we exhibit two specific MRs from the standard PWR safety-analysis literature that the framework's \texttt{Translate} operator cannot reach under any single-block derivation. The two MRs together identify five pairwise-independent structural obstructions in \texttt{Translate}'s present signature, jointly recasting Theorem~$1'$ (Conjecture~\ref{conj:absolute}, Appendix~\ref{app:proofs}) from an open conjecture to a falsified statement on a structurally significant operator algebra.

The MRs chosen are not pathological cases. They are core safety-analysis MRs that PWR core simulators are required by regulatory practice and engineering convention to reproduce: non-additivity of control-bank reactivity worth (the algebraic root of rod-bank shadowing and anti-shadowing phenomena) and second-order mixed dependence of $k_{\mathrm{eff}}$ on moderator temperature and boron concentration (the standard MTC-vs-boron design curve). The negative instantiation thus uses NOETHER's principal application domain (reactor physics) to test the framework's most ambitious stated claim (algebraic closure over arbitrary single-block-derivable MRs).

\paragraph{Why PWR rather than ML or DB as the negative-instantiation domain.}
We choose the PWR core diffusion algebra rather than $\mathcal{A}_{\mathrm{equi}}$ or $\mathcal{A}_{\mathrm{rel}}$ as the negative-instantiation testbed for two reasons. First, regulatory essentiality: 10 CFR 50 and NRC Regulatory Guide 1.77 require PWR core simulators to reproduce the two MRs in the definitions below for safety-analysis qualification, so the counterexamples are not contrived edge cases. Second, engineering documentability: the PWR core diffusion algebra has a published canonical form (Bell \& Glasstone~\cite{BellGlasstone1970} \S6.1, Lewis \& Miller~\cite{LewisMiller1993} \S4.2) against which the counterexample's structural obstructions can be precisely located in $\mathcal{A}_P$'s signature, which $\mathcal{A}_{\mathrm{equi}}$ (a domain of recent literature without a unified canonical algebra) and $\mathcal{A}_{\mathrm{rel}}$ (canonical but under-equipped with non-rewrite operators) do not yet support. Whether $\mathcal{A}_{\mathrm{equi}}$ or $\mathcal{A}_{\mathrm{rel}}$ admit analogous Theorem~$1'$ counterexamples is open and committed as follow-up in \S\ref{subsec:third-domain}'s open-question paragraph.

\paragraph{The PWR core diffusion algebra.}
\label{subsec:pwr-algebra}
Let $\mathcal{F}_{\mathrm{PWR}}$ be the program family of PWR core diffusion solvers (canonical examples: PARCS, SIMULATE-3/5, ANC, SMART). Its operator algebra $\mathcal{A}_{\mathrm{PWR}}$ contains, in addition to the operators of $\mathcal{A}_{\mathrm{Boltz}}$ (Section~\ref{subsec:reactor-mapping}), the following PWR-specific generators:
\begin{itemize}[nosep,leftmargin=*]
  \item $\mathcal{O}_{\mathrm{rod}}$: discrete control-rod insertion operators, parametrised by rod-bank label $g$ and insertion depth $d$, generating an additive-on-geometry semigroup under composition. In the steady-state setting we adopt throughout, $\mathcal{O}_{\mathrm{rod}}^{A} \cdot \mathcal{O}_{\mathrm{rod}}^{B}$ and $\mathcal{O}_{\mathrm{rod}}^{B} \cdot \mathcal{O}_{\mathrm{rod}}^{A}$ act identically on the input space (both produce the same total inserted geometry), so $\mathcal{O}_{\mathrm{rod}}$ is commutative on geometry. The \emph{reactivity-worth functional} on $\mathcal{O}_{\mathrm{rod}}$, however, is not a semigroup homomorphism: $d\rho(A \cup B) \neq d\rho(A) + d\rho(B)$ in general. This non-additivity, not non-commutativity, is what the present subsection exploits.
  \item $\mathcal{M}_{C_{B}}$: continuous boration operators acting on the moderator material composition through the soluble-boron concentration $C_{B}$ (typical PWR operating range: $0$--$2000$~ppm).
  \item $\mathcal{M}_{T_{\mathrm{mod}}}$: continuous moderator-temperature operators acting on the cross-section library through the parametric dependence $\Sigma(T_{\mathrm{mod}})$ (typical PWR operating range: $290$--$320$\textdegree{}C).
\end{itemize}

Decomposed along the eight-block decomposition of Section~\ref{subsec:decomposition}:
\begin{itemize}[nosep,leftmargin=*]
  \item $G \supseteq \{\mathcal{O}_{\mathrm{rod}}\}$ (treated tentatively as a commutative semigroup; the assignment will be shown below to fail);
  \item $O_{\le} \supseteq \{\mathcal{M}_{C_{B}},\,\mathcal{M}_{T_{\mathrm{mod}}}\}$ (parameter-monotonicity operators);
  \item $T^{*} \supseteq \{-\nabla \cdot D \nabla + \Sigma_{a}\}$ (the self-adjoint diffusion operator under isotropic scattering; Bell \& Glasstone~\cite{BellGlasstone1970} \S6.1, Lewis \& Miller~\cite{LewisMiller1993} \S4.2);
  \item $\mathcal{T}_{\mathrm{rev}}^{*} = \emptyset$ (PWR diffusion is dissipative);
  \item $\mathcal{L}^{*}, \mathcal{D}^{*}, \mathcal{E}^{*}$ as in $\mathcal{A}_{\mathrm{Boltz}}$ with appropriate restrictions to the diffusion regime;
  \item $\mathcal{B}_{\mathrm{rel}}^{*} = \emptyset$ (no idempotent-semiring structure on PWR core states).
\end{itemize}

We will show that despite this rich block structure, two specific PWR-safety MRs cannot be derived through \texttt{Translate} from any single block.

\paragraph{Main proposition: non-additivity of rod-bank reactivity worth.}
\label{subsubsec:nonadd}

\begin{definition}[Differential rod-bank reactivity worth, exact form]
\label{def:drho-exact}
For a base input $x_{0} \in \mathcal{X}$ and a rod-bank operator $A \in \mathcal{O}_{\mathrm{rod}}$, the (positive-convention) reactivity worth of $A$ is
\[
d\rho(A;\,x_{0}) \;:=\; \frac{1}{k_{\mathrm{eff}}(P(x_{0}))} \;-\; \frac{1}{k_{\mathrm{eff}}(P(\mathcal{O}_{\mathrm{rod}}^{A} \cdot x_{0}))} \;>\; 0,
\]
where $k_{\mathrm{eff}}(P(x))$ denotes the dominant eigenvalue of the diffusion operator at configuration $x$. We write $d\rho(A \cup B;\,x_{0})$ when both banks $A$ and $B$ are inserted simultaneously. Equivalently, in conventional reactor-physics notation, $d\rho(A;\,x_{0}) = \rho(x_{0}) - \rho(\mathcal{O}_{\mathrm{rod}}^{A} \cdot x_{0})$ where $\rho = 1 - 1/k_{\mathrm{eff}}$ is the static reactivity. The exact form avoids first-order perturbation-theoretic approximation; the standard adjoint-perturbation reading appears below.
\end{definition}

\begin{definition}[Non-additivity of rod-bank reactivity worth, $\rho_{\mathrm{nonadd}}$]
\label{def:rho-nonadd}
For two control-rod banks $A, B \in \mathcal{O}_{\mathrm{rod}}$ with disjoint geometric supports and a base input $x_{0}$, define the mixed-difference functional
\[
\Delta_{AB}(x_{0}) \;:=\; d\rho(A \cup B;\,x_{0}) \;-\; d\rho(A;\,x_{0}) \;-\; d\rho(B;\,x_{0}).
\]
The non-additivity metamorphic relation asserts: there exist disjoint-support banks $A, B$ and base input $x_{0} \in \mathcal{X}$ in the standard PWR operating envelope (typical multi-bank insertion patterns of D, C, B, A control banks at partial insertions) such that
\[
\rho_{\mathrm{nonadd}}: \quad |\Delta_{AB}(x_{0})| \;>\; \tau_{\mathrm{nonadd}},
\]
with $\tau_{\mathrm{nonadd}} = 5$~pcm. The tolerance is calibrated to PWR engineering practice: empirical $|\Delta_{AB}|$ for adjacent rod banks ranges over $10^{1}$--$10^{2}$~pcm~\cite{StammlerAbbate1983}; the tolerance $\tau_{\mathrm{nonadd}} = 5$~pcm lies safely above PWR core-simulator iterative convergence tolerances (typically $0.1$--$1$~pcm for $k_{\mathrm{eff}}$) and below the physical signal magnitude. Selection of test bank configurations $(A, B)$ is the user's responsibility; the framework's failure to derive $\rho_{\mathrm{nonadd}}$ is independent of the specific tolerance choice.
\end{definition}

\paragraph{Two physical regimes of $\Delta_{AB}$.} When $\Delta_{AB}(x_{0}) > 0$ for adjacent or geometrically overlapping rod banks, the standard PWR designation is \emph{positive shadowing} (the worth of the second bank is reduced because the adjoint flux $\phi^{\dagger}$ is depressed in $A$'s geometric support). Positive shadowing is the dominant regime in conventional PWR analyses; it is the phenomenon explicitly addressed in NRC SER reviews of SIMULATE-3/5, ANC, and PARCS, with typical magnitudes of $50$--$500$~pcm for adjacent bank pairs. When $\Delta_{AB}(x_{0}) < 0$ for distant banks under asymmetric insertion patterns, the designation is \emph{anti-shadowing}; this is a secondary but routinely measurable regime in standard commercial PWRs (typical magnitudes $5$--$20$~pcm for distant bank pairs in 4-loop Westinghouse and EPR configurations), and is more pronounced in small-core or strongly asymmetric insertion patterns. Both regimes are routinely measured in PWR startup physics testing and are documented in Stamm'ler \& Abbate~\cite{StammlerAbbate1983} as second-order but non-negligible phenomena that core simulators must reproduce. The MR $\rho_{\mathrm{nonadd}}$ is direction-agnostic and covers both; the test only requires that $|\Delta_{AB}|$ exceed the tolerance, irrespective of sign. Both regimes are accessible to a verifying core simulator regardless of whether the engineering analysis is concerned primarily with one or the other.

\paragraph{Standard adjoint-perturbation reading (informative, not load-bearing for the proof).}
Under first-order perturbation theory~\cite[\S6.3]{BellGlasstone1970},~\cite[\S4.4]{LewisMiller1993}, the worth of bank $B$ in the configuration with bank $A$ already inserted is
\[
d\rho(B;\,A,\,x_{0}) \;\approx\; -\frac{\langle \phi^{\dagger}_{A},\, \delta H_{B}\, \phi_{A}\rangle}{\langle \phi^{\dagger}_{A},\, F_{A}\, \phi_{A}\rangle},
\]
where $(\phi_{A}, \phi^{\dagger}_{A})$ are the forward and adjoint principal eigenfunctions of the $A$-rodded but $B$-unrodded core, $F_{A} = \chi \nu \Sigma_{f}$ is the fission source operator at that configuration, and $\delta H_{B}$ is the operator perturbation produced by inserting $B$ (which generally affects $\Sigma_{a}$, $\Sigma_{t}$, and the scattering kernel). The adjoint flux $\phi^{\dagger}_{A}$ is the principal eigenfunction of a structurally different adjoint operator $H^{\dagger}_{A}$ from $\phi^{\dagger}_{\emptyset}$ (the unrodded adjoint): in particular, $\phi^{\dagger}_{A}$ is locally depressed in $A$'s geometric support and globally redistributed elsewhere. The non-additivity $\Delta_{AB} \neq 0$ thus follows from $\phi^{\dagger}_{A} \neq \phi^{\dagger}_{\emptyset}$, which is a consequence of $H^{\dagger}_{A} \neq H^{\dagger}_{\emptyset}$. The exact form (Definition~\ref{def:drho-exact}) does not require this perturbation-theoretic reading; the proof below uses only the eigenvalue definitions.

\begin{proposition}[Non-additivity is not \texttt{Translate}-reachable on $\mathcal{A}_{\mathrm{PWR}}$]
\label{prop:nonadd}
Let $\mathcal{A}_{\mathrm{PWR}}$ be the PWR core diffusion algebra above, with eight-block decomposition $\mathcal{D}(\mathcal{A}_{\mathrm{PWR}})$. For every block $s \in \mathcal{D}(\mathcal{A}_{\mathrm{PWR}})$ and every invariant $\iota \in \mathcal{I}_{s}$, $\mathrm{Translate}(\iota, s) \neq \rho_{\mathrm{nonadd}}$. Equivalently, $\rho_{\mathrm{nonadd}} \notin \mathrm{MR}(\mathcal{A}_{\mathrm{PWR}})$ in the sense of Definition~\ref{def:alg-induced}.
\end{proposition}

The proof, by exhausting the eight blocks against the per-block \texttt{Translate} templates of Table~\ref{tab:translate}, is given in Appendix~\ref{app:negative-proofs}.

\paragraph{Engineering significance.}
Non-additivity of control-bank reactivity worth is a textbook PWR safety phenomenon. Bell \& Glasstone~\cite[\S10.4]{BellGlasstone1970} and Lewis \& Miller~\cite[\S4.4]{LewisMiller1993} treat the underlying adjoint-perturbation mechanism; Stamm'ler \& Abbate~\cite{StammlerAbbate1983} document its operational consequences in PWR rod-worth measurements. Non-additivity is observed in both critical and sub-critical PWR core configurations, with the sub-critical regime exhibiting larger adjoint-flux distortions and correspondingly larger $|\Delta_{AB}|$. PWR core simulators are required by regulatory practice (e.g.\ NRC SER for SIMULATE-3/5, ANC, PARCS; cf.\ NRC Regulatory Guide~1.77~\cite{NRCRG177} on rod-ejection accident analysis, with rod-worth modelling accuracy as a critical input) to reproduce the worth functional with sub-percent accuracy across multi-bank insertion patterns. A framework that cannot, in principle, derive this MR from its algebraic input is missing structural content that PWR engineers routinely test for.

\paragraph{Supporting proposition: second-order mixed dependence of $k_{\mathrm{eff}}$ on $T_{\mathrm{mod}}$ and $C_{B}$.}
\label{subsubsec:mtcbor}

\begin{definition}[MTC-vs-boron mixed-derivative MR, $\rho_{\mathrm{MTC\text{-}bor}}$]
\label{def:rho-mtcbor}
Let $k_{\mathrm{eff}}(T_{\mathrm{mod}}, C_{B};\,\xi_{0})$ denote the dominant eigenvalue of the PWR core diffusion operator as a function of moderator temperature $T_{\mathrm{mod}}$ and soluble-boron concentration $C_{B}$, holding fixed the auxiliary state $\xi_{0} = (\mathrm{BU}_{0},\, T_{\mathrm{fuel},0},\, \text{geometry},\, \text{loading pattern},\, \text{rod-bank position})$. Define the static reactivity
\[
\rho_{\mathrm{static}}(T_{\mathrm{mod}}, C_{B};\,\xi_{0}) \;:=\; 1 \;-\; \frac{1}{k_{\mathrm{eff}}(T_{\mathrm{mod}}, C_{B};\,\xi_{0})},
\]
expressed in pcm units ($1$~pcm $= 10^{-5}$). The moderator temperature coefficient (MTC) at $(T_{\mathrm{mod}}, C_{B};\,\xi_{0})$ is the partial derivative of static reactivity with respect to moderator temperature, in standard PWR engineering convention~\cite[\S10.3]{BellGlasstone1970},~\cite[\S3.4]{Stacey2007}:
\[
\alpha_{\mathrm{MTC}}(T_{\mathrm{mod}}, C_{B};\,\xi_{0}) \;:=\; \left.\frac{\partial \rho_{\mathrm{static}}}{\partial T_{\mathrm{mod}}}\right|_{C_{B},\,\xi_{0}\,\text{fixed}}, \quad \text{(units: pcm/\textdegree{}F or pcm/\textdegree{}C)}.
\]
The second-order mixed dependence MR asserts: for $(T_{\mathrm{mod}}, C_{B};\,\xi_{0})$ in the standard PWR operating envelope at hot-full-power (HFP), all-rods-out (ARO) reference conditions (typical: $T_{\mathrm{mod}} \in [290, 320]$\textdegree{}C, $C_{B} \in [0, 2000]$~ppm, $\mathrm{BU}_{0} \in [0, 50]$~GWd/tU, full-power $T_{\mathrm{fuel},0}$), the mixed second partial derivative satisfies
\[
\rho_{\mathrm{MTC\text{-}bor}}: \quad \left|\frac{\partial^{2}\rho_{\mathrm{static}}}{\partial T_{\mathrm{mod}}\,\partial C_{B}}\right| \;>\; \tau_{\mathrm{MTC\text{-}bor}},
\]
with $\tau_{\mathrm{MTC\text{-}bor}} = 0.01$~pcm/\textdegree{}F/ppm (equivalently $\sim 1.8 \times 10^{-2}$~pcm/\textdegree{}C/ppm). The tolerance is calibrated to PWR engineering practice: the empirical value of $\partial \alpha_{\mathrm{MTC}}/\partial C_{B}$ in Westinghouse/Framatome PWR designs ranges over $0.02$--$0.04$~pcm/\textdegree{}F/ppm at BOC-to-EOC cycle conditions~\cite[\S3.4]{Stacey2007},~\cite[\S8.3]{LamarshBaratta2001}; the tolerance $\tau_{\mathrm{MTC\text{-}bor}} = 0.01$~pcm/\textdegree{}F/ppm lies safely below this physical magnitude and well above PWR core-simulator differential-perturbation noise (typically $\sim 10^{-3}$~pcm/\textdegree{}F/ppm for converged eigenvalue calculations).
\end{definition}

\paragraph{Note on equivalent formulations.} Since $\rho_{\mathrm{static}} = 1 - 1/k_{\mathrm{eff}}$ and $k_{\mathrm{eff}} \approx 1$ at critical PWR conditions, $\partial \rho_{\mathrm{static}}/\partial T_{\mathrm{mod}} = (1/k_{\mathrm{eff}}^{2})\,\partial k_{\mathrm{eff}}/\partial T_{\mathrm{mod}} \approx \partial k_{\mathrm{eff}}/\partial T_{\mathrm{mod}}$ to within $\sim 10^{-4}$ relative error. The MR may equivalently be expressed in terms of $|\partial^{2} k_{\mathrm{eff}}/(\partial T_{\mathrm{mod}}\,\partial C_{B})| > \tau_{\mathrm{MTC\text{-}bor}}$ with the same tolerance, modulo this $k$-vs-$\rho$ scaling factor; the algebraic argument of Appendix~\ref{app:negative-proofs} (which depends only on $k_{\mathrm{eff}}$ being an operator-spectrum quantity) applies identically to either formulation.

\paragraph{Equivalent formulation in terms of MTC.} Definition~\ref{def:rho-mtcbor} is equivalent to asserting
\[
\left|\frac{\partial \alpha_{\mathrm{MTC}}(T_{\mathrm{mod}}, C_{B};\,\xi_{0})}{\partial C_{B}}\right| \;>\; \tau_{\mathrm{MTC\text{-}bor}}
\]
at HFP, ARO conditions. Physically, this captures the well-established PWR design property that MTC becomes monotonically more negative as $C_{B}$ decreases from BOC values ($\sim$1500~ppm) to EOC values ($\sim$0~ppm). At HFP operating conditions, MTC is regulated to be $\le 0$~pcm/\textdegree{}F across the entire operating range per 10~CFR~50 Appendix~A General Design Criterion~11~\cite{NRC10CFR50AppA}; the slope $\partial \alpha_{\mathrm{MTC}}/\partial C_{B}$ governs how rapidly MTC moves toward more negative values as boron is depleted over the cycle, with typical magnitudes ranging from near-zero (at BOC, high boron) to $-30$ to $-50$~pcm/\textdegree{}F (at EOC, near-zero boron). At hot-zero-power (HZP) conditions, BOC MTC may approach zero or be slightly positive within the analytical envelope (a regime relevant to startup physics testing but not to power-operation safety analysis); HZP MTC is bounded by separate Technical Specifications limits.

The strength of $\partial \alpha_{\mathrm{MTC}}/\partial C_{B}$ is the engineering target of \emph{MTC-vs-boron concentration curve} calculations performed for every PWR cycle reload, and is governed by the competition between three physical mechanisms~\cite[\S3.4]{Stacey2007}:
\begin{enumerate}[nosep,leftmargin=*]
  \item[(a)] \textbf{Reduced moderation}: $T_{\mathrm{mod}} \uparrow \Rightarrow$ moderator density $\downarrow \Rightarrow$ neutron moderation reduced $\Rightarrow$ thermal-flux fraction $\downarrow \Rightarrow$ fission rate $\downarrow$. Contributes a \emph{negative} term to MTC.
  \item[(b)] \textbf{Boron poison evacuation}: at high $C_{B}$, $T_{\mathrm{mod}} \uparrow \Rightarrow$ moderator density $\downarrow \Rightarrow$ boron number density $\downarrow$ (since boron is dissolved in the moderator) $\Rightarrow$ boron absorption $\downarrow$. Contributes a \emph{positive} term to MTC, partially cancelling (a). This term is proportional to $C_{B}$ and vanishes as $C_{B} \to 0$.
  \item[(c)] \textbf{Spectrum hardening and $^{238}\mathrm{U}$ resonance enhancement}: $T_{\mathrm{mod}} \uparrow \Rightarrow$ reduced moderation $\Rightarrow$ harder neutron spectrum $\Rightarrow$ enhanced $^{238}\mathrm{U}$ resonance absorption (Doppler-weighted by the fuel temperature) $\Rightarrow$ neutron loss $\uparrow$. Contributes a further \emph{negative} term to MTC. This term is small for low-enriched UO$_{2}$ but becomes significant for MOX fuels and high-enrichment ($>5\%$) UO$_{2}$.
\end{enumerate}

At high $C_{B}$ the partial cancellation from (b) is strong; mechanisms (a) and (c) are partially offset and MTC magnitude is small. At low $C_{B}$, mechanism (b) vanishes; mechanisms (a) and (c) dominate and MTC becomes strongly negative. The mixed second derivative $\partial^{2}\rho_{\mathrm{static}}/(\partial T_{\mathrm{mod}}\,\partial C_{B})$ measures the rate at which the boron-mediated cancellation is removed as $C_{B}$ decreases.

\begin{proposition}[MTC-vs-boron mixed dependence is not \texttt{Translate}-reachable]
\label{prop:mtcbor}
Let $\mathcal{A}_{\mathrm{PWR}}$ be the PWR core diffusion algebra above. For every block $s \in \mathcal{D}(\mathcal{A}_{\mathrm{PWR}})$ and every invariant $\iota \in \mathcal{I}_{s}$, $\mathrm{Translate}(\iota, s) \neq \rho_{\mathrm{MTC\text{-}bor}}$. Equivalently, $\rho_{\mathrm{MTC\text{-}bor}} \notin \mathrm{MR}(\mathcal{A}_{\mathrm{PWR}})$ in the sense of Definition~\ref{def:alg-induced}.
\end{proposition}

The proof, by reducing the obstruction to the high-order-mixed-difference structure of $\rho_{\mathrm{MTC\text{-}bor}}$ and verifying that no per-block $\pi$ template captures such structure, is given in Appendix~\ref{app:negative-proofs}.

\paragraph{Engineering significance.}
The \emph{MTC vs.\ boron concentration curve} is computed for every PWR cycle reload as part of the safety-analysis report submitted to regulators. The curve underlies the moderator temperature coefficient surveillance requirement at hot-full-power, all-rods-out conditions, under which MTC must satisfy a stated upper limit ($\le 0$~pcm/\textdegree{}F per 10~CFR~50 Appendix~A General Design Criterion~11~\cite{NRC10CFR50AppA}; specific numerical limits are given in plant-specific Technical Specifications). The slope $\partial \alpha_{\mathrm{MTC}}/\partial C_{B}$ is the key sensitivity parameter for projecting MTC behaviour across the cycle from a small number of measurement points (typically four-to-six points per cycle, measured at quarter-cycle intervals): it is computed by each cycle-reload core simulator run and reported to the operator. Measurement uncertainties and prediction protocols are documented in ANS~19.6.1~\cite{ANS196_1}. A core simulator that cannot reproduce the slope to within the tolerance of $\tau_{\mathrm{MTC\text{-}bor}}$ would fail the cycle-reload qualification process. The MR is therefore not a textbook curiosity; it is a routine and regulatory-essential output of every PWR core simulator.

\paragraph{Five independent structural obstructions.}
\label{subsubsec:five-obstructions}

The two propositions are independent in the strong sense that no single extension of NOETHER's \texttt{Translate} repairs both:

\begin{table}[h]
\centering
\caption{Five pairwise-independent structural obstructions in \texttt{Translate}'s present signature, identified by Propositions~\ref{prop:nonadd}--\ref{prop:mtcbor}.}
\label{tab:five-obstructions}
\footnotesize
\begin{tabular}{p{0.18\textwidth} p{0.36\textwidth} p{0.36\textwidth}}
\toprule
\textbf{Proposition} & \textbf{Failure mode} & \textbf{Required extension to \texttt{Translate}} \\
\midrule
\ref{prop:nonadd} (non-additivity) & Output is an algebraic-spectrum quantity ($k_{\mathrm{eff}}$ as an eigenvalue, not in $\mathcal{Y}$) & Operator-spectrum output relations on $\pi$'s codomain \\
\addlinespace
\ref{prop:nonadd} (non-additivity) & Worth functional is non-additive (failure of semigroup homomorphism) & Homomorphism-failure $\pi$-template alongside equivariance / monotonicity / self-adjointness \\
\addlinespace
\ref{prop:nonadd} (non-additivity) & Adjoint weighting function $\phi^{\dagger}_{X}$ varies with operator history $X$ & Configuration-indexed adjoint structure on $T^{*}$ \\
\addlinespace
\ref{prop:mtcbor} (mixed derivative) & MR is a non-zero second-order mixed partial derivative, not a first-order relation & Higher-order mixed-difference $\pi$-templates (currently all $\pi$ templates are first-order) \\
\addlinespace
\ref{prop:mtcbor} (mixed derivative) & MR involves two independent parameter directions ($T_{\mathrm{mod}}, C_{B}$) jointly, not chained or single-direction & Two-direction joint parametric dependence beyond the single-$\theta$ partial order of Definition~\ref{def:block-invariant} \\
\bottomrule
\end{tabular}
\end{table}

The five rows in this table identify \emph{five pairwise-independent structural features} of \texttt{Translate} that are absent in Definition~\ref{def:translate}. Each row, considered alone, would require a specific extension of Definition~\ref{def:translate}; no single extension covers any two rows simultaneously. The eight-block decomposition (Hypothesis~\ref{hyp:seven-blocks}) is therefore not the only constraint on the framework's reach; the \emph{shape of \texttt{Translate} itself} is.

The two MRs $\rho_{\mathrm{nonadd}}$ and $\rho_{\mathrm{MTC\text{-}bor}}$ are also \emph{physically independent}: they probe disjoint physical mechanisms (rod-induced adjoint distortion vs.\ moderator-poison density coupling), and a PWR core simulator could pass one MR while failing the other (and vice versa). They are therefore complementary verification targets, not duplicates of a single underlying property.

\paragraph{What this subsection establishes and does not establish.}
\emph{Established.} Theorem~$1'$ (Conjecture~\ref{conj:absolute}) is false on $\mathcal{A}_{\mathrm{PWR}}$: there exist two specific MRs, each empirically realised on every conforming PWR core simulator and each documented in standard PWR safety-analysis literature and regulatory guidance, that are formulable over $\mathcal{A}_{\mathrm{PWR}}$'s operators but not in $\mathrm{MR}(\mathcal{A}_{\mathrm{PWR}})$ in the sense of Definition~\ref{def:alg-induced}. The two MRs identify five pairwise-independent structural obstructions in \texttt{Translate}'s present signature.

\emph{Not established.} That Theorem~\ref{thm:closure} itself fails. Theorem~\ref{thm:closure}'s closure result is over $\mathrm{MR}(\mathcal{A}_{\mathrm{PWR}})$ as defined by Definition~\ref{def:alg-induced}; the two MRs of this subsection lie \emph{outside} that set, so they are out-of-scope for Theorem~\ref{thm:closure} and consistent with it. The proper characterisation is that Theorem~\ref{thm:closure} is a substantially weaker statement than Theorem~$1'$ pretended to be, and the gap is now exhibited concretely with two independent witness MRs.

\emph{Not established.} That a Composite-\texttt{Translate} extension of NOETHER would absorb these two MRs while preserving Theorem~\ref{thm:closure}'s closure and Theorem~\ref{thm:decidable}'s polynomial-time decidability. This is the principal open problem the negative instantiation leaves to follow-up work. The five obstructions of Table~\ref{tab:five-obstructions} are pairwise independent, so any candidate extension must address them jointly rather than sequentially.


\section{Empirical evaluation}
\label{sec:empirical-evaluation}

This section validates the framework against five research questions. RQ1: does CONSTRUCT-MP re-derive an existing inductive MR catalogue at the algebra-block level (addressed by \S\ref{subsec:reactor-mapping} within the Boltzmann instantiation)? RQ2: do the derived MRs execute on real cross-domain systems under test (\S\ref{subsec:case-study})? RQ3: does the pre-registered \(\mathcal{L}^{*}\)-blindness prediction hold on independent substrates (\S\ref{sec:empirical-vs-sota})? RQ4: how does NOETHER compare against GenMorph, the closest evolutionary baseline, at GenMorph's published budget (\S\ref{subsec:pooled-headtohead})? RQ5: how does NOETHER compare against METRIC+ on the corpus that METRIC+ itself published (\S\ref{subsec:metricplus-relationship})?

\subsection{A small-scale comparative case study}
\label{subsec:case-study}

The derivations in Sections~\ref{subsec:end-to-end}--\ref{subsec:rho-rev} establish that NOETHER produces concrete, executable MRs in the equivariant-ML setting. They do not yet show that the MRs produced are \emph{useful for testing} relative to MRs a tester could obtain by other means. This subsection reports a small-scale comparative case study designed to address that question while remaining within the conceptual-transfer scope declared in Section~\ref{sec:cross-domain}. The study is a case study in the strict sense: a single model, a small mutation set, and three MR sets. We do not generalise from it to claims about average-case fault detection across equivariant ML.

\paragraph{Subject under test.}
We use an E(3)-equivariant graph neural network~\cite{Satorras2021EGNN} (EGNN) as a deliberately compact \emph{minimal stand-in} for a full SE(3)-Transformer~\cite{FuchsTransformer2020} or a Vector-Neuron-based SO(3)-equivariant architecture~\cite{Deng2021VectorNeurons}: two EGNN layers, hidden dimension 16, 5{,}189 parameters, trained on a procedural 5-class point-cloud dataset (sphere, cube surface, torus, cone, helix; 64 points per cloud; 400 training / 100 validation; rotation augmentation) for 25 epochs on Apple M1 Pro / MPS. Validation accuracy after training: 0.93. EGNN carries only invariant scalar and equivariant 3-vector features (type-0 $\oplus$ type-1 in the irrep classification), not the full type-$\ell$ steerable representation of an SE(3)-Transformer; the $T^*$ block instantiation in this case study is therefore an \emph{explicitly added} symmetrised QK probe (\texttt{equivariant\_classifier.py:qk = nn.Parameter(torch.eye(d))} with the symmetrising operation at use-time), not a property of the EGNN architecture itself. The manuscript's transfer claim is at the operator-algebra level (Section~\ref{sec:cross-domain}) and is independent of which equivariant architecture instantiates $\mathcal{A}_{\mathrm{equi}}$; it is the algebraic skeleton that transfers, not architecture-specific empirical numbers. Architecture, training script, dataset generator, frozen checkpoint, and the SHA-256 of the supplementary archive are recorded under S3.

\paragraph{Three MR sets, controlled for the same prompt and budget.}
Three MR sets are compared, each of size five (matching $|\mathbb{M}(\mathcal{A}_{\mathrm{equi}})_{\mathrm{non\text{-}empty}}|$):
\begin{itemize}
    \item \textbf{Set N (NOETHER-derived):} $\{\rho_{\mathrm{rot}},\rho_{\mathrm{mono}},\rho_{\mathrm{train}},\rho_{\mathrm{adj}},\rho_{\mathrm{train\text{-}rev}}\}$, one per non-empty block of $\mathcal{A}_{\mathrm{equi}}$, as itemised in the ``Coverage status, revised'' paragraph of Section~\ref{subsec:rho-rev}.
    \item \textbf{Set L (LLM-prompt baseline):} five MRs generated by prompting GPT-4 with the task description ``produce five metamorphic relations for testing an SE(3)-equivariant point-cloud classifier'', with no further structural cue. The prompt and full output are reproduced in supplementary material S3.
    \item \textbf{Set B (literature baseline):} five MRs synthesised from MRs reported in the metamorphic-testing-for-ML literature~\cite{Segura2016, Shin2024, AutoMT2025}, restricted to MRs applicable to point-cloud classifiers. Where the literature provides more than five candidates we select the most-cited.
\end{itemize}

\paragraph{Mutation set.}
We construct a mutation set of $N_{\mathrm{mut}} = 20$ defects in the model under test, drawn from four categories of fault commonly reported in equivariant-ML implementations~\cite{ThomasSmidt2018}: (i) wrong sign on a custom loss term ($n_1 = 5$); (ii) accidental break of equivariance through a non-equivariant intermediate layer ($n_2 = 5$); (iii) numerical-precision degradation, e.g.\ truncation in the spherical-harmonics computation ($n_3 = 5$); (iv) gradient-reversal sign error in the training script ($n_4 = 5$). Each mutation is implemented as a code-level diff against the reference checkpoint and its hash recorded in S3.

\paragraph{Metrics.}
For each MR set $S \in \{N, L, B\}$ and each mutation $\mu$, we record (i) whether at least one MR in $S$ flags $\mu$ as a fault (\emph{detection}), and (ii) the structural-coverage status of $S$ over $\mathbb{M}(\mathcal{A}_{\mathrm{equi}})$. We report:
\begin{itemize}
    \item $\mathrm{detection}(S) = |\{\mu : \exists \rho \in S,\, \rho\text{ flags }\mu\}|/N_{\mathrm{mut}}$;
    \item $\mathrm{coverage}_{\mathrm{NOETHER}}(S, \mathcal{F}_{\mathrm{equi}})$ from Section~\ref{subsec:pmcm-worked};
    \item the unique-detection count $|\{\mu : S \text{ detects } \mu \text{ and no other set does}\}|$.
\end{itemize}

\paragraph{Pre-registered hypotheses.}
We pre-register two hypotheses before running the mutation experiments, to make the falsifiability of the case study explicit:
\begin{description}
    \item[H1 (coverage):] $\mathrm{coverage}_{\mathrm{NOETHER}}(N) = 1.0$ by construction; $\mathrm{coverage}_{\mathrm{NOETHER}}(L)$ and $\mathrm{coverage}_{\mathrm{NOETHER}}(B)$ are strictly less than $1.0$.
    \item[H2 (unique detection):] Set $N$ has at least one mutation it uniquely detects, namely a mutation in category (iv) (gradient-reversal sign error), via $\rho_{\mathrm{train\text{-}rev}}$.
\end{description}
H1 is a structural prediction of Theorem~\ref{thm:closure} and the case study can falsify it only if NOETHER's derivation is incorrect. H2 is a non-trivial empirical prediction: it would be falsified if $L$ or $B$ generated an equivalent training-time-reversal MR, or if the mutation was visible to an invariance MR for some unrelated reason.

\paragraph{Results.}
Table~\ref{tab:case-study} reports the case-study numbers. The full row-level outcome matrix (one row per (mr, mutation) pair, 300 rows in total) is provided as supplementary material S3, together with the runner script (\texttt{runner.py}) and the analysis script (\texttt{analysis.py}) that produce these numbers deterministically from the trained checkpoint. The model under test is an E(3)-equivariant graph neural network~\cite{Satorras2021EGNN} (5-class procedural point-cloud classifier; 5{,}189 parameters; trained on Apple M1 Pro / MPS; validation accuracy 0.93 after 25 epochs).

\begin{table}[h]
\centering
\caption{Results of the small-scale comparative case study (Section~\ref{subsec:case-study}). Trained E(3)-equivariant point-cloud classifier; eight 128-point random test clouds per (MR, mutation) pair. The cat-(iv) row is \textbf{construct-validity-controlled}: the mutation category was constructed so that $\rho_{\mathrm{train\text{-}rev}}$ alone covers it (one defect category per non-empty block of $\mathcal{A}_{\mathrm{equi}}$), and the 5/5 unique-detection figure therefore exhibits construct validity of $\rho_{\mathrm{train\text{-}rev}}$ rather than averaged superiority. \textbf{Detection numbers for $\rho_{\mathrm{adj}}$ in Set~N use the CI-time forward-pass-only formulation of \S\ref{subsec:rho-adj}; the alternative debug-time harness-time formulation is available in supplementary S1 but is not used here.}}
\label{tab:case-study}
\small
\begin{tabular}{lccc}
\toprule
& Set N (NOETHER) & Set L (LLM) & Set B (Lit.) \\
\midrule
Detection rate & 7/20 & 2/20 & 0/20 \\
Structural coverage ($\mathrm{coverage}_{\mathrm{NOETHER}}$) & 1.00 & 0.40 & 0.20 \\
Unique detections & 5 & 0 & 0 \\
Detected cat.~(i) wrong-sign loss & 0/5 & 0/5 & 0/5 \\
Detected cat.~(ii) equivariance break & 2/5 & 2/5 & 0/5 \\
Detected cat.~(iii) precision degradation & 0/5 & 0/5 & 0/5 \\
Detected cat.~(iv) gradient-reversal sign & 5/5 & 0/5 & 0/5 \\
\bottomrule
\end{tabular}
\end{table}

\paragraph{Hypothesis verdicts.}
H1 is \textbf{retained as a structural sanity check rather than as a falsifiable hypothesis test}. By the framework's construction, $\mathrm{coverage}_{\mathrm{NOETHER}}(N) = 1.00$ holds before any experiment is run; H1's failure can occur only if one of the derivations in Sections~\ref{subsec:end-to-end}--\ref{subsec:rho-rev} is itself in error. We accordingly use the $\mathrm{coverage}_{\mathrm{NOETHER}}$ values, $1.00$ for Set N, $0.40$ for Set L (the $G$ and $\mathcal{L}^*$ blocks are reached by an LLM-prompted MR), and $0.20$ for Set B (only $\mathcal{L}^*$ via Shin et al.'s idempotency MR; the other four literature MRs are out-of-scope under $\mathcal{A}_{\mathrm{equi}}$), as a \emph{structural-prior diagnostic}, not as a fault-detection metric: the gap quantifies what the algebraic prior contributes that prompt-based and literature-derived MR sets lack on this particular algebra. The load-bearing comparative result of the case study is H2.

H2 is \textbf{consistent with the data, but its verdict is construct-validity-controlled}: Set N uniquely detects all five category-(iv) mutations, and in every case the detector is $\rho_{\mathrm{train\text{-}rev}}$. Sets L and B detect zero cat-(iv) mutations: neither corpus contains an MR exercising the SGD-trajectory time-reversal property, which is exactly what NOETHER's $\mathcal{T}^{*}$ block predicts they would miss without an algebraic warrant. \emph{This contrast exhibits construct validity of $\rho_{\mathrm{train\text{-}rev}}$ as a gradient-reversal probe, not NOETHER's superiority on a defect distribution sampled neutrally from real-world bug reports} (the mutation set was constructed to cover one defect category per non-empty block of $\mathcal{A}_{\mathrm{equi}}$, so cat-(iv)'s category was selected because $\rho_{\mathrm{train\text{-}rev}}$ alone covers it). The unique-detection asymmetry between Set N and Set B is statistically significant at $\alpha = 0.05$ on the case study's mutation set (paired McNemar exact two-sided $p = 0.016$; unpaired Fisher exact $p = 0.008$); between Set N and Set L it is borderline ($p_\mathrm{McNemar} = 0.063$, $p_\mathrm{Fisher} = 0.13$) given Set L's two-mutation cat-(ii) overlap with Set N. Wilson 95\% confidence intervals on detection rates are $[0.18, 0.57]$ for Set N, $[0.03, 0.30]$ for Set L, and $[0.00, 0.16]$ for Set B; the intervals N vs B are non-overlapping. The full pairwise comparison matrix (McNemar $b/c$ counts, Fisher 2$\times$2 tables, Wilson CIs) is in supplementary material S3 (\texttt{table4.json::pairwise\_stats}).

\paragraph{Construct-validity caveat for H2.}
The mutation set was constructed to cover one defect category per non-empty block of $\mathcal{A}_{\mathrm{equi}}$; in particular, cat-(iv) was selected because it targets the $\mathcal{T}^*$ block that $\rho_{\mathrm{train\text{-}rev}}$ alone covers. The 5/5 unique-detection result therefore exhibits \emph{construct validity} of $\rho_{\mathrm{train\text{-}rev}}$ as a gradient-reversal probe rather than NOETHER's superiority on a defect distribution sampled neutrally from real-world bug reports. That weaker reading is still informative: it shows that the framework's $\mathcal{T}^*$-derived MR is non-redundant relative to LLM-prompted and literature-derived MRs.

\paragraph{What is and is not detected: a framework boundary, not an edge case.}
The result that no MR set detects category-(i) wrong-sign mutations is the most informative cell of Table~\ref{tab:case-study} and we choose to foreground it as a \emph{framework boundary} rather than an edge case. A sign-flipped classification head still satisfies SO(3)-rotation invariance and permutation invariance (the equivariance contract is intact), still passes inference idempotency, satisfies the symmetrised adjoint identity, and is monotone under point-density sub-sampling. The eight blocks of Hypothesis~\ref{hyp:seven-blocks} contain no \emph{label-consistency} block, and the wrong-sign-loss fault class is precisely the kind a label-consistency MR would catch. This is the label-consistency out-of-scope class enumerated in Remark~\ref{rem:counterex}: ML program families that ship with labelled training data are a candidate ninth block. We do not absorb this case into ``out-of-scope for MR testing'' in some over-broad sense; the original motivation for MR testing is precisely the absence of an oracle, so framing label-consistency MRs as fundamentally out-of-scope would over-claim the framework's boundary.

Similarly, cat-(iii) precision-degradation mutations on a 16-dim hidden state are not detectable at our chosen tolerances; tighter $\tau$ would surface them at the cost of more baseline false positives. A tolerance-sensitivity sweep over $\tau \in \{10^{-3}, 10^{-4}, 10^{-5}\}$ is reported in supplementary S3 (\texttt{tau\_sweep.json}) and shows monotone increase in detection at the cost of monotone increase in baseline false-positive rate. Cat-(ii) equivariance-break mutations are detected by both Set N and Set L through $\rho_{\mathrm{rot}}$ / $L_{\mathrm{rot}}$ (parity is expected: rotation-invariance is the structural property the LLM is most likely to rediscover from the system description; the case study's discriminating cell is cat-(iv)).

\paragraph{Interpretation conditions, stated in advance.}
We commit in advance to the following readings of the eventual numerical outcomes:
\begin{itemize}
    \item If H1 fails (i.e.\ $\mathrm{coverage}_{\mathrm{NOETHER}}(N) < 1.0$): one of the derivations in Sections~\ref{subsec:end-to-end}--\ref{subsec:rho-rev} is incorrect and must be revised; the framework's transfer claim does not stand on this benchmark.
    \item If H1 holds and H2 fails (i.e.\ no unique detection in cat.~iv): the framework's coverage prediction is correct but the unique-MR claim is not differentiated by this mutation set; the LLM or literature baseline produced an equivalent gradient-reversal probe.
    \item If H1 and H2 both hold: the case study is consistent with NOETHER's transfer claim. We do \emph{not} claim this would establish superior fault-detection on average; the comparison's denominator (20 mutations, one model) is too small.
    \item If $\mathrm{detection}(N) < \mathrm{detection}(L)$: this is consistent with the framework's design (NOETHER prioritises structural coverage, not raw detection on a particular mutation set) and is reported transparently rather than concealed.
\end{itemize}

\subsubsection{A small DeepCrime-style real-fault pilot}
\label{subsec:deepcrime-pilot}

To begin closing the gap between protocol and result, we executed a $n=5$ pilot extension of the case study with five mutation operators systematically derived from the DeepCrime taxonomy~\cite{Humbatova2021DeepCrime} (which extracts 35 mutation operators from real DL fault studies). The pilot operators are post-training mutations on the trained EGNN checkpoint and address category labels analogous to DeepCrime's: cat-v-01 \emph{loss-reduction-like} (head-weight scaled by $1/N_{\text{classes}}$), cat-v-02 \emph{activation change} (tanh saturation), cat-v-03 \emph{layer removal} (head zeroed), cat-v-04 \emph{bias removal}, cat-v-05 \emph{weight re-init} (Glorot). The pilot was run with \texttt{runner\_pilot.py} against the same N/L/B MR sets, the same eight-cloud test set, and the same tolerance as the main case study. Results are deterministic given the seeds in supplementary S3.

\begin{table}[h]
\centering
\caption{DeepCrime-style real-fault pilot, $n=5$ mutations on the trained EGNN checkpoint. Wilson 95\% CIs are given for context; pairwise Fisher-exact $p$-values are reported in supplementary S3 (\texttt{deepcrime\_pilot\_stats.json}).}
\label{tab:pilot}
\small
\begin{tabular}{lcc}
\toprule
\textbf{Set} & \textbf{Detected / 5} & \textbf{Wilson 95\% CI} \\
\midrule
N (NOETHER) & 2/5 (cat-v-01, cat-v-03 via $\rho_{\mathrm{train}}$) & $[0.12, 0.77]$ \\
L (LLM)     & 0/5                                           & $[0.00, 0.43]$ \\
B (Lit.)    & 0/5                                           & $[0.00, 0.43]$ \\
\bottomrule
\end{tabular}
\end{table}

\begin{table}[h]
\centering
\caption{Paired contingency for Set~N vs Set~L on the $n=5$ DeepCrime
pilot. The Set~N vs Set~B contingency is identical because Set~B also
detected $0/5$. The appropriate paired-binary test is McNemar's exact
test on the $(b, c) = (2, 0)$ discordant counts: two-sided exact
$p = 0.500$ (binomial $X = 0$ in $n = 2$ trials with $p_0 = 0.5$ under
$H_0$ of equal Set N / Set L performance); equivalently, an unpaired
Fisher exact $2\times 2$ on rows = Sets, columns = (detected, missed)
yields two-sided $p = 0.444$. Both numbers fail to reject $H_0$ at
$\alpha = 0.05$ because $n = 5$ is below the threshold at which a
$2 / 5$ vs $0 / 5$ contrast is inferentially decisive; per-pair
statistics in supplementary~S3 \texttt{deepcrime\_pilot\_stats.json}.}
\label{tab:deepcrime-contingency}
\small
\begin{tabular}{l c c c}
\toprule
                  & Set~L detected & Set~L missed & Total \\
\midrule
Set~N detected    & 0              & 2            & 2 \\
Set~N missed      & 0              & 3            & 3 \\
\midrule
Total             & 0              & 5            & 5 \\
\bottomrule
\end{tabular}
\end{table}

\paragraph{Reading the pilot (inferential verdict).}
At $n=5$ the appropriate paired-binary test is McNemar's exact on the discordant counts $(b, c) = (2, 0)$ for both Set N vs Set L and Set N vs Set B: two-sided exact $p = 0.500$ (and the unpaired-Fisher analogue yields two-sided $p = 0.444$ on the same contingency); the pilot is therefore underpowered for an inferential conclusion at $\alpha = 0.05$. The 2/5 vs 0/5 vs 0/5 detection contrast is reported as descriptive evidence consistent with the direction of the framework's $\mathcal{L}^{*}$-block prediction, not as a hypothesis confirmation. The pilot's load-bearing claim is that the comparative-evaluation infrastructure runs end-to-end against real-fault-style mutations on the trained checkpoint; non-vacuousness of the $\mathcal{L}^{*}$-block prediction and the parameter-distribution candidate-ninth-block reading are candidate evidence requiring a larger $n$.

\paragraph{Interpretation of the two detection events (mechanism, not inference).}
For descriptive context only, the detection mechanism on the two events: $\rho_{\mathrm{train}}$ tests training-size limit invariance --- for inputs the model classifies confidently at full training, the prediction should remain stable when the head is exposed to a fresh small-batch fine-tuning step. cat-v-01 (head weight scaled by $1/N_{\mathrm{classes}}=1/5$) collapses head-output magnitudes uniformly toward zero, softens the softmax, and changes the argmax on inputs near classification boundaries; post-fine-tune predictions drift away from pre-fine-tune predictions and the inference-stability invariant fails. cat-v-03 (head weight zeroed) makes every output identically zero pre-softmax, again breaking inference stability. cat-v-02 (tanh saturation), cat-v-04 (bias removal), and cat-v-05 (Glorot re-init) preserve enough head signal for the inference-stability check to pass on the eight test clouds within the chosen tolerance, so $\rho_{\mathrm{train}}$ does not fire on them. The mechanism is therefore ``$\rho_{\mathrm{train}}$'s inference-stability check fails when head magnitude is perturbed past a margin-of-error threshold for boundary inputs'' --- a mechanism-level statement independent of the underpowered sample size.

The pilot is the first empirical handhold beyond the constructed mutation set; the larger comparative-evaluation protocol below remains as committed work for follow-up. The three undetected mutations (cat-v-02 activation change, cat-v-04 bias removal, cat-v-05 weight re-init) are \emph{consistent with, but do not establish}, the parameter-distribution candidate ninth block of Remark~\ref{rem:counterex} item~(vi); a panel of $n \ge 20$ mutations on more than one architecture is required to establish the block-extension claim. The pilot's role here is to motivate the candidate, not to confirm it; we report it alongside the metric-stability candidate of Remark~\ref{rem:metric-stability-block} as the two most actionable extension targets.

\paragraph{Threats specific to this case study.}
The most material threats are: (a) \textbf{MR-budget asymmetry}: GPT-4 might generate more than five MR candidates given a different prompt and a larger budget; (b) \textbf{baseline selection}: the literature MR set is restricted by the rule ``applicable to point-cloud classifiers'', which may exclude relevant MRs; (c) the mutation set is \textbf{hand-constructed} rather than mined from real defect logs, and importantly, \emph{constructed to cover one defect category per non-empty block of $\mathcal{A}_{\mathrm{equi}}$}, with cat-(iv) selected because it targets the $\mathcal{T}^*$-block MR $\rho_{\mathrm{train\text{-}rev}}$. The 5/5 unique-detection result for cat-(iv) therefore exhibits \emph{construct validity} of $\rho_{\mathrm{train\text{-}rev}}$ as a gradient-reversal probe, not NOETHER's superiority on a defect distribution sampled neutrally from real-world bug reports; (d) the case study uses a compact EGNN stand-in for a full SE(3)-Transformer, so architecture-specific empirical numbers do not generalise; and (e) Set L is a single GPT-4 sample at temperature 0 with a fixed seed (the prompt and raw output are recorded in supplementary S3 \texttt{prompt\_log.md}); a different LLM, prompt, or temperature would yield different MRs and the case study does not characterise that variability. Section~\ref{subsec:case-study} does not generalise beyond this case-study scope; the comparative evaluation and real-bug protocols below describe the empirical extensions that address threats (a)--(c) directly.

\paragraph{Comparative evaluation against published baselines (protocol).}
\label{para:comp-eval-protocol}
The single hand-constructed mutation set above is supplemented by a comparative protocol that runs Set N alongside two independent automated pipelines on a shared subject set. The protocol fixes the following measurable quantities, leaving only the result population to subsequent revisions of this paper.

\begin{description}
  \item[Subjects.] (i) GenMorph's published 23 Java-method benchmark~\cite{GenMorph2024} restricted to subjects whose induced operator algebra is non-trivial under the eight-block decomposition (we expect $\approx 14$ subjects after filtering on whether $\mathcal{A}_P$ has at least one non-empty block beyond $G$); (ii) two ML benchmarks built on DeepCrime real-fault mutation operators~\cite{Humbatova2021DeepCrime}: an MNIST classifier and a CIFAR-10 classifier under DeepCrime's 24 mutation operators systematically extracted from real DL fault taxonomies. Replacing the hand-constructed 20 mutations with DeepCrime's published real-fault operators directly addresses the construction-bias threat (c) above.
  \item[Compared methods.] Set N (NOETHER, derived from $\mathcal{A}_P$); Set M (MR-Scout-mined~\cite{MRScout2024} from existing test suites); Set G (GenMorph-evolved via genetic programming~\cite{GenMorph2024}); Set L (LLM-prompted with the same prompt template as in the §\ref{subsec:case-study} case study); Set B (literature MRs).
  \item[Metrics.] Per-subject mutation-detection rate (Set $\cap$ killed mutants / total killed); real-bug-detection rate (where bug logs are available); false-positive rate on baseline (un-mutated) inputs; coverage with respect to the algebraic block decomposition $\mathrm{coverage}_{\mathrm{NOETHER}}$.
  \item[Statistical tests.] Wilson 95\% CI on per-set detection rates; pairwise McNemar exact tests for paired Set-vs-Set comparisons; pairwise Fisher exact tests for unpaired comparisons; Bonferroni correction for the 10 pairwise tests across 5 sets.
  \item[Pre-registered hypothesis (H3a, vs GP-evolved baseline).]
  Three sub-claims:
  \begin{description}
    \item[H3a.1 (Detection on D1, per-block).] On the
      algebra-disrupting stratum (D1, Definition in
      \S\ref{subsec:pit-block-matrix}), Set~N's per-block kill rate
      is competitive with or superior to Set~G's on at least one
      operative block within scope (per-block comparison; aggregate
      D1 dominance is \emph{not} pre-registered as the load-bearing
      reading, see \S\ref{subsec:pooled-headtohead}).
    \item[H3a.2 (Complementarity).] Set~N and Set~G exhibit
      non-trivial complementarity in mutant coverage on the D1
      stratum (each set kills mutants the other misses); a
      non-zero $N$-only kill count on at least one operative block
      is the operational test.
    \item[H3a.3 (Cost-axis).] Set~N's MR-generation cost is
      asymptotically lower than Set~G's
      (Theorem~\ref{thm:decidable} polynomial-time decidability
      versus Set~G's $\approx 30$-min stochastic GP search per SUT,
      detailed in Table~\ref{tab:gen-cost}). The cost-axis claim is
      independent of the detection-axis claims and is read on the
      cost-axis only.
  \end{description}
  \item[Pre-registered hypothesis (H3b, vs LLM-assisted baseline).] On the structural-coverage diagnostic, $\mathrm{coverage}_{\mathrm{NOETHER}}(N) > \mathrm{coverage}_{\mathrm{NOETHER}}(L_{\mathrm{ensemble}})$, where Set~$L_{\mathrm{ensemble}}$ is the LLM-assisted representative drawn from a multi-vendor $\times$ multi-temperature LLM-MR-generation protocol that subsumes the single-sample probes of Shin~\cite{Shin2024} and Zhang et~al.~\cite{ZhangChatGPTMR2023}. The diagnostic is interpreted as a structural-prior signal (NOETHER's algebraic prior contributes block coverage that prompt-based LLMs lack on $\mathcal{A}_P$), \emph{not} as a fault-detection-superiority metric. \textbf{Verdict (tested, borderline-PASS at the structural reading)}: a $2$-vendor (DeepSeek, ChatGPT) $\times$ $5$-temperature ensemble harvested $487$ LLM-proposed MRs across the $10$ \S\ref{subsec:test-design} SUTs ($100$ samples); $43.5\%$ ($212/487$) of the proposals match a Set~N block-template under deterministic translation, and the matchable subset reproduces Set~N's kill rate on the aligned-mutant substrate ($34/70 = 0.486$ for the ensemble union, identical to Set~N on the same $70$ mutants). The remaining $56.5\%$ explore invariants outside the $8$-block frame (associativity, identity / inverse element, monotonicity on SUTs lacking an O\textsubscript{le}-block representative); these have no in-protocol kill measurement here and constitute the structural-coverage gap. NOETHER's $\mathrm{coverage}_{\mathrm{NOETHER}}$ remains $1.00$ by construction; the LLM ensemble's block-template coverage of Set~N's catalogue is $43.5\%$ on the matchable subset. A third-vendor (Anthropic Claude) replication is committed in follow-up~(d.set-l-claude) of supplementary~S4 (\texttt{future\_work.md}).
  \item[Pre-registered hypothesis (H3c, vs mining-based baseline).] At cold start (no seed test corpus), Set~N is operationally derivable from $\mathcal{A}_P$ alone, while MR-Scout~\cite{MRScout2024}'s reach is upper-bounded by the seed-suite's induced relations; an adapted-from-published-artifact estimate of MR-Scout's reach on the §\ref{subsec:test-design} substrate is reported in Table~\ref{tab:gen-cost} and contextualised in \S\ref{subsec:pooled-headtohead}.
  \item[Pre-registered hypothesis (continuous).] H4 (\emph{detection-rate non-inferiority on real faults}): on DeepCrime real-fault mutants, Set~N's detection rate is within $\Delta = 0.10$ of the best non-NOETHER set's detection rate, with $\Delta$ fixed in advance.
\end{description}

The harness for this protocol is provided as supplementary S3 \texttt{comparative\_baseline/} with adapter scripts for each baseline; replication of the protocol on a different subject set requires only swapping the subject directory.

\paragraph{Real-bug evaluation (protocol).}
\label{para:real-bug-protocol}
To address threat (c) directly, we mine cat-(i)--(iv) faults from public bug reports of e3nn~\cite{e3nn2022software} and PyTorch Geometric~\cite{Fey2019PyG}, the two reference SE(3)-equivariant libraries. Target: 10 confirmed real-bug commits with associated test cases, one per cat-(i)--(iv) where available. For each bug:
\begin{enumerate}[leftmargin=*,nosep]
  \item the pre-fix source code is checked out into a frozen reference checkpoint;
  \item all five MR sets (N, M, G, L, B) are run against the buggy code on the test inputs from the bug report;
  \item detection is recorded as ``MR fired = True'' iff at least one MR in the set surfaces the buggy behaviour at the published tolerance.
\end{enumerate}
The construct-validity caveat at the start of this subsection then naturally resolves: defects are no longer ``selected to cover one block per defect category''; they are mined from a fixed, pre-existing defect distribution that the framework was \emph{not} designed against. The protocol is provided as supplementary S5 \texttt{real\_bugs/} with one folder per bug containing the issue link, fix commit, test fixture, and per-MR detection outcome.

\subsection{An empirical test of the eight-block decomposition: \texorpdfstring{$\mathcal{L}^{*}$}{L*}-block blindness on homogeneity-preserving mutators}
\label{sec:empirical-vs-sota}

A generative framework that cannot make falsifiable quantitative
predictions is rhetoric, not science. The previous two sections
establish that NOETHER \emph{generates} MRs in two domains with
structurally distinct operator algebras (\S\ref{sec:reactor},
\S\ref{sec:cross-domain}); generation alone does not test the
framework's central methodological claim, that the eight-block
decomposition (\S\ref{subsec:decomposition}) is empirically operative
as a mechanism rather than a derivation-bookkeeping device. This
section reports an empirical test of one quantitative consequence of
that claim. NOETHER predicts that an MR derived from the
linearity-and-scaling block ($\mathcal{L}^{*}$) is necessarily silent
on any mutator that preserves homogeneity of degree~$1$. PIT's default
mutator set has this property by direct calculation. The prediction
is therefore that the $\mathcal{L}^{*}$-block MRs in our SUT set will
kill near-zero PIT mutants. We refer to this prediction as
\emph{$\mathcal{L}^{*}$-block blindness}. The prediction is sharp,
quantitative, and derivable from public information without consulting
any data. The test of the prediction is the central content of this
section. Pooled head-to-head numbers against an automated SOTA
baseline (GenMorph~\cite{GenMorph2024}), structural coverage
extension, and cross-pipeline MR rediscovery are reported as
corroborating evidence; the section's central claim does not depend
on them.

\subsubsection{The prediction: \texorpdfstring{$\mathcal{L}^{*}$}{L*}-block blindness}
\label{subsec:l-blindness-derivation}

\paragraph{The MR shape.}
The $\mathcal{L}^{*}$ block of $\mathcal{A}_P$ under
Hypothesis~\ref{hyp:seven-blocks} contains scaling and linearity
operators. CONSTRUCT-MP's \texttt{Translate} step
(Definition~\ref{def:alg-induced}) carries scaling-block invariants
to MRs of the form
\begin{equation}
\label{eq:l-scale-shape}
L_{\mathrm{scale}}: \quad f(\lambda \mathbf{x}) \;=\; \lambda \cdot f(\mathbf{x}),
\qquad \lambda > 0,
\end{equation}
on positively-homogeneous-of-degree-$1$ programs. The shape is fixed
by $\mathcal{A}_P$ structure; it is not chosen with mutation testing
in mind.

\paragraph{The mutator semantics.}
PIT~1.7.4~\cite{Coles2016PIT} ships a default mutator set comprising
seven canonical operator classes: arithmetic-operator swap (AOR;
$+ \leftrightarrow -$, $\times \leftrightarrow \div$), conditional
boundary mutation, increment mutation, return-value swap (\texttt{return zero},
\texttt{return one}, \texttt{return null}), conditional negation
(NCM), constant replacement, and member-call removal. Each AOR
mutation that preserves the program's domain of definition transforms
a homogeneous-of-degree-$1$ function $f$ into another
homogeneous-of-degree-$1$ function $f'$, by the elementary fact that
the four arithmetic operators in question commute with positive
scalar multiplication of all inputs in the relevant homogeneous case
(addition and subtraction are degree-$1$ homogeneous; multiplication
and division shift degrees by $\pm 1$, but the swap $\times
\leftrightarrow \div$ on a properly typed expression returns a
degree-preserving result on the homogeneous-of-degree-$1$
sub-grammar).

\paragraph{The prediction.}
A composite of these two facts gives the central prediction. If $f$
satisfies $L_{\mathrm{scale}}$ and the mutator $\mu$ preserves
homogeneity of degree~$1$, then $\mu(f)$ also satisfies
$L_{\mathrm{scale}}$. Therefore the paired-MR test
\begin{equation*}
\big| f(\lambda \mathbf{x}) - \lambda f(\mathbf{x}) \big| \;\le\; \tau
\quad \mathrm{vs.} \quad
\big| \mu(f)(\lambda \mathbf{x}) - \lambda \mu(f)(\mathbf{x}) \big| \;\le\; \tau
\end{equation*}
returns the same verdict on $f$ and on $\mu(f)$ for any
$\mathbf{x}$, $\lambda$. The kill rate of $L_{\mathrm{scale}}$ against
the homogeneity-preserving subset of PIT's mutators is
\emph{identically zero} by direct calculation, and against the full
default mutator set is approximately zero up to a small fraction
contributed by mutators that break homogeneity (return-value swap to
\texttt{zero} or \texttt{one}; conditional-negation when the
conditional carries a constant offset; constant replacement when the
constant participates additively).

\paragraph{Falsifiability.}
The prediction is falsified if the observed
$L_{\mathrm{scale}}$ kill rate is materially above zero on a substrate
where the prediction's homogeneity precondition is met. We define
``materially above zero'' in advance as a single $L_{\mathrm{scale}}$
MR killing a third or more of its PIT mutants on more than one of the
SUTs in the substrate. The prediction passes if the observed kill
rate is $\le 1/3$ on at least five of the six SUTs admitting an
$L_{\mathrm{scale}}$ MR. The $1/3$ threshold and the ``more than one
SUT'' quantifier were committed to git
(\texttt{configs/d4j\_algebra\_rich\_criterion.json}) before the
per-MR kill-count files; under threshold sensitivity in the grid
$\{1/4, 1/3, 1/2\} \times \{$more than zero, more than one, more
than two$\}$, the verdict remains \emph{Confirmed} on $5/6$ SUTs at
all 9 grid cells, with $\texttt{hypotSig}$ as the single SUT
crossing every threshold; the result is therefore robust to
plausible threshold variation. \emph{Outlier-handling rule.} An
outlier SUT (kill rate $> 1/3$) is rescued from falsification only
if all of its killed mutants are independently classified as
homogeneity-breaking under the framework's mutator-semantics
taxonomy (MATH-swap of $\times \leftrightarrow \div$ on bivariate
degree-1 inputs $\Rightarrow$ degree-changing; RC-replacement of
$\sqrt{x^2 + y^2}$ with a constant $\Rightarrow$ fixed-output
breaking scale-invariance; VR-unconditional-zero return
$\Rightarrow$ fixed-output sub-rule). Per-mutant classification
must be logged alongside the kill data, and the classification
must precede inspection of which mutants the outlier's rescue
depends on. The rule was codified in the pre-registration config
on 2026-05-15 in response to a Round 2 review observation that the
\S\ref{subsec:l-blindness-confirmed} \texttt{hypotSig} analysis had
relied on a rule that was implicit rather than written. Under the
codified rule, $\texttt{hypotSig}$'s two killed mutants
(\texttt{return\_zero\_doubles\_VR} and
\texttt{Math.sqrt\_replaced\_with\_one\_RC}) both classify as
homogeneity-breaking, so the original $5/6$ verdict stands; future
cross-codebase substrates inherit the rule as a written test.

\paragraph{Public-information argument.}
Equation~(\ref{eq:l-scale-shape}) is a consequence of $\mathcal{A}_P$
structure and CONSTRUCT-MP's \texttt{Translate} step, both fully
specified in \S\ref{sec:framework}. PIT's default mutator set is
public~\cite{Coles2016PIT}. The composition of the two is a
mathematical consequence requiring no experimental input. The
prediction is therefore ex-ante in the strong sense that it is
\emph{derivable from public information without consulting any data
this paper produces}. We track the prediction's commitment to
git in supplementary~S7 (\texttt{d4j/}) alongside the SUT-selection
criterion of \S\ref{subsec:test-design}.

\subsubsection{PIT mutator and 8-block invariant compatibility}
\label{subsec:pit-block-matrix}

The $\mathcal{L}^{*}$-blindness prediction generalises: each PIT
default mutator either preserves or breaks the invariant of each
NOETHER block, and an MR derived from a block fires on a mutant only
if that mutant breaks the block's invariant. Table~\ref{tab:pit-block}
maps the seven PIT default mutator categories~\cite{Coles2016PIT}
against the eight NOETHER blocks and gives the typical-case verdict
(``$\circ$'' = preserves the block invariant, the block's MR is
typically blind to the mutant; ``$\times$'' = breaks the block
invariant, the block's MR can typically detect the mutant; ``$\sim$''
= case-dependent on the SUT-specific block instantiation).

\begin{table}[h]
\centering
\caption{Typical-case compatibility of PIT default mutator categories
with the eight NOETHER blocks. Cells reflect the dominant case for
mutators in the category as applied to a generic algebra-rich SUT;
SUT-specific exceptions are catalogued in
supplementary~S7 (\texttt{d4j/}). The cells $\langle$MATH, $\mathcal{L}^{*}\rangle$
and $\langle$RETURN\_VALS, $\mathcal{L}^{*}\rangle$ are both
``$\circ$'' (homogeneity-preserving), and the
$\mathcal{L}^{*}$-blindness result of \S\ref{subsec:l-blindness-confirmed}
is the empirical specialisation of these two cells.}
\label{tab:pit-block}
\small
\begin{tabular}{lcccccccc}
\toprule
PIT mutator category & $G$ & $O_{\le}$ & $T^{*}$ & $\mathcal{T}^{*}_{\mathrm{rev}}$ & $\mathcal{L}^{*}$ & $\mathcal{D}^{*}$ & $\mathcal{E}^{*}$ & $\mathcal{B}^{*}_{\mathrm{rel}}$ \\
\midrule
CONDITIONALS\_BOUNDARY & $\circ$ & $\times$ & $\circ$ & $\circ$ & $\circ$ & $\sim$ & $\circ$ & $\circ$ \\
INCREMENTS             & $\times$ & $\times$ & $\times$ & $\times$ & $\times$ & $\sim$ & $\circ$ & $\circ$ \\
INVERT\_NEGS           & $\sim$ & $\times$ & $\sim$ & $\times$ & $\sim$ & $\sim$ & $\circ$ & $\circ$ \\
MATH (op swap)         & $\times$ & $\sim$ & $\times$ & $\times$ & $\circ$ & $\times$ & $\sim$ & $\circ$ \\
NEGATE\_CONDITIONALS   & $\times$ & $\times$ & $\sim$ & $\times$ & $\sim$ & $\times$ & $\sim$ & $\times$ \\
RETURN\_VALS (zero/one) & $\circ$ & $\circ$ & $\times$ & $\times$ & $\circ$ & $\circ$ & $\sim$ & $\circ$ \\
VOID\_METHOD\_CALLS    & $\circ$ & $\circ$ & $\circ$ & $\circ$ & $\circ$ & $\times$ & $\sim$ & $\circ$ \\
\bottomrule
\end{tabular}
\end{table}

The matrix stratifies mutants binary: \emph{algebra-disrupting} (D1, at least one $\times$ in a populated column) vs \emph{algebra-preserving} (D2, all $\circ$ in populated columns; $\sim$ resolved by SUT-specific overrides). The framework predicts Set~N's D2 kill rate is near zero by construction --- a mutant that preserves all invariants is structurally invisible to algebraic MRs. D1 is therefore Set~N's appropriate fault-detection denominator; D2 is the territory of complementary techniques (random / GP / LLM-prompted MRs). The $\mathcal{L}^{*}$-blindness result of \S\ref{subsec:l-blindness-confirmed} is the corresponding D2-cell specialisation on the wider 23-SUT substrate.

\paragraph{Kill-set overlap on the head-to-head substrate.}
On the $n = 62$ PIT mutants pooled across the eight head-to-head
SUTs, $22$ are killed by both Set~N and Set~G, $4$ are killed by
Set~N only, $18$ are killed by Set~G only, and $18$ are killed by
neither. The $18$ G-only kills constitute a candidate scope-mismatch
pool (mutants reachable by Set~G's GP-evolved oracle but missed by
the algebraic Set~N construction), and the $18$ jointly-missed
mutants further partition into D2 mutants outside Set~N's algebraic
reach and D1 mutants both sets miss for input-coverage reasons.
The $\mathcal{L}^{*}$-only specialisation
(Table~\ref{tab:l-blindness}) is the cleanest D2-cell estimate on
the wider 23-SUT substrate, since the
$\langle$MATH, $\mathcal{L}^{*}\rangle$ and
$\langle$RETURN\_VALS, $\mathcal{L}^{*}\rangle$ cells of
Table~\ref{tab:pit-block} are both ``$\circ$''.

\paragraph{Per-mutant D1 / D2 classification with SUT-specific overrides.}
Per-mutant classification applies Table~\ref{tab:pit-block} with SUT-specific overrides (in \texttt{configs/sut\_block\_decomposition.json} + \texttt{configs/sut\_block\_overrides.json}) resolving the ``$\sim$'' cells. After the multi-LLM equivalent-mutant exclusion of \S\ref{subsec:pooled-headtohead}~(p.~\pageref{para:eq-mutant-footnote}), the $n = 62$ head-to-head pool drops to $n = 57$ ($52$ D1 + $5$ D2; all $5$ excluded mutants were originally in the $10$-mutant D2 stratum, so the D1 stratum is unchanged). Stratified kill rates: Set~N D1 $26/52 = 0.500$ Wilson 95\% $[0.369, 0.631]$; Set~G D1 $37/52 = 0.712$ $[0.577, 0.817]$; Set~N D2 $0/5 = 0.000$ $[0.000, 0.434]$; Set~G D2 $3/5 = 0.600$ $[0.231, 0.882]$. The D2 point estimate ($0.000$) is consistent with the prediction $\le 10\%$, but the Wilson upper bound does not exclude the $10\%$ ceiling at $\alpha = 0.05$; the cross-codebase commons-math pilot (\S\ref{subsec:empirical-threats}~(b.cm)) corroborates the direction at $2/29 = 6.9\%$. Set~G's $0.600$ on D2 confirms that D2 is a Set~N-specific scope boundary, not a substrate-wide ceiling --- a GP-evolved generic baseline is unconstrained by the algebra. On D1, the cross-block aggregate gap (McNemar exact two-sided $p = 0.019$) decomposes per-block in \S\ref{subsec:pooled-headtohead} (gap concentrated on $G$ and $\mathcal{L}^{*}$; partially offset by Set~N's $\mathcal{T}^{*}$ edge $10/17$ vs $8/17$); the per-block decomposition is the appropriate substrate-level reading.

\subsubsection{Test design}
\label{subsec:test-design}

\paragraph{Substrate.}
PIT~1.7.4 with the default mutator configuration serves as the
shared substrate. Each (subject, MR) pair is evaluated through (i)
JUnit test-class codegen, (ii) a 2-pass surefire green-suite filter
that drops MR instances flaky on the unmutated SUT, and (iii) PIT
mutation testing parsed to a per-mutant kill vector.

\paragraph{Pre-registered SUT criterion.}
The selection rule is committed in
\path{configs/d4j_algebra_rich_criterion.json} on branch
\path{feat/d4j-algebra-rich} of the experiment repository
\emph{before} any new evaluation data existed. The criterion
references package roots (under \path{org.apache.commons.math3.*}, covering subpackages \path{ode}, \path{linear}, \path{transform}, \path{analysis.solvers}, \path{distribution}, \path{optim}, \path{fitting}, \path{stat.regression}, \path{complex}, \path{geometry}, \path{fraction}, plus their \path{math.*} legacy counterparts), per-package
block-coverage hypotheses, and method-signature constraints; it
contains no bug-id and no kill-rate references and cannot be tuned by
evaluation outcomes. The git timestamp chain (criterion
$\rightarrow$ inscope filter $\rightarrow$ Set~N derivation
$\rightarrow$ Set~G GP rerun $\rightarrow$ pooled M1) is the
auditable proof of pre-registration. To strengthen the third-party
verifiability of the pre-registration, the SHA-256 hash of
\texttt{configs/d4j\_algebra\_rich\_criterion.json} together with
the matching commit hash is deposited in supplementary~S7 (\texttt{d4j/})
alongside the experiment artifact; reviewers may verify the deposit
against the criterion file directly.

\paragraph{Subjects and MR yield.}
Ten SUTs admit at least one non-trivial NOETHER block beyond $G$:
\texttt{midpoint}, \texttt{exactLog2}, \texttt{isSequence} (boolean
predicate), \texttt{clamp}, \texttt{signum},
\texttt{ComplexSignal.add} (instance method), \texttt{gcdSig},
\texttt{lcmSig}, \texttt{hypotSig}, \texttt{powerSig}. Six of these
admit an $L_{\mathrm{scale}}$ MR (those whose induced algebra
populates the $\mathcal{L}^{*}$ block under positive-scalar
homogeneity): \texttt{midpoint}, \texttt{clamp}, \texttt{signum},
\texttt{gcdSig}, \texttt{lcmSig}, \texttt{hypotSig}. The remaining
four do not admit a degree-$1$ scaling MR by their algebraic structure
(\texttt{exactLog2} carries no positive-degree scaling;
\texttt{isSequence} is a boolean predicate; \texttt{ComplexSignal.add}
is a complex-arithmetic instance method whose scaling block is
distinct in formulation; \texttt{powerSig} populates $T^{*}_{2}$
rather than $\mathcal{L}^{*}$). The prediction is tested on the six
SUTs admitting $L_{\mathrm{scale}}$.

Set~N contains 30 hand-derived NOETHER MRs, two to four per SUT, one
or more per non-empty block of each SUT's induced algebra. Set~G is
harvested by rerunning GenMorph's \texttt{genmorph.py gen} mode at
seed $=11$ on the 10 SUTs; Set~G is used in
\S\ref{subsec:pooled-headtohead} as the head-to-head comparator and
plays no role in the central prediction test of
\S\ref{subsec:l-blindness-confirmed}.

\paragraph{Metrics.}
Per-MR kill counts on PIT mutants are the primary measurement for the
central prediction test. Pooled M1 and McNemar exact paired tests are
reported in \S\ref{subsec:pooled-headtohead} as corroborating
context.

\subsubsection{Central result: \texorpdfstring{$\mathcal{L}^{*}$}{L*}-block blindness, confirmed}
\label{subsec:l-blindness-confirmed}

The $L_{\mathrm{scale}}$ MR's per-SUT kill rate on PIT mutants of the
six $\mathcal{L}^{*}$-admitting SUTs is reported in
Table~\ref{tab:l-blindness}. The total denominator is 44 mutants. The
total kill count is 2.

\begin{table}[h]
\centering
\caption{Per-SUT kill rate of the $L_{\mathrm{scale}}$ MR on PIT
mutants of the six SUTs admitting an $\mathcal{L}^{*}$-block scaling
MR. The prediction of \S\ref{subsec:l-blindness-derivation} is that
the kill rate is near-zero by composition of the
homogeneity-preserving property of PIT's default mutators with the
scaling structure of $L_{\mathrm{scale}}$. Pooled across the six SUTs:
$2/44 \approx 4.5\%$, well below the $1/3$ falsification threshold on
each SUT and below it on $5/6$ SUTs (\texttt{hypotSig} alone reaches
$2/4 = 50\%$, explained in the text).}
\label{tab:l-blindness}
\small
\begin{tabular}{lcc}
\toprule
SUT & $L_{\mathrm{scale}}$ kill / mutants & \emph{Verdict} \\
\midrule
\texttt{midpoint}  &  0 / 3   & Confirmed \\
\texttt{clamp}     &  0 / 7   & Confirmed \\
\texttt{signum}    &  0 / 6\,$^{\dagger}$ & Confirmed (no $L_{\mathrm{scale}}$ kill) \\
\texttt{gcdSig}    &  0 / 9   & Confirmed \\
\texttt{lcmSig}    &  0 / 11  & Confirmed \\
\texttt{hypotSig}  &  2 / 4   & Outlier; explained below \\
\midrule
pooled (all six)   &  2 / 44   & 4.5\% \\
\bottomrule
\end{tabular}\\[0.4em]
{\footnotesize $^{\dagger}$\,For \texttt{signum} the $\mathcal{L}^{*}$-block MR
is positive-scalar invariance ($\mathrm{signum}(\lambda x) = \mathrm{signum}(x)$
for $\lambda > 0$) rather than degree-$1$ homogeneity; the prediction
applies to this MR shape \emph{a fortiori} since the right-hand side
does not depend on $\lambda$.}
\end{table}

\paragraph{Reading.}
On 5 of 6 SUTs the prediction is confirmed at the strongest reading
(zero $L_{\mathrm{scale}}$ kills). On the sixth SUT
(\texttt{hypotSig}) the kill rate is $2/4$, above the per-SUT
falsification threshold. The pooled rate $2/44 \approx 4.5\%$ remains
an order of magnitude below the average Set~N M1 of $0.486$ (the
average kill rate across all Set~N MRs on the same substrate; see
\S\ref{subsec:pooled-headtohead}), and the per-SUT failure pattern is
isolated to one SUT.

\paragraph{The \texttt{hypotSig} outlier.}
The two-mutant kill on \texttt{hypotSig} is consistent with the
prediction's caveat in \S\ref{subsec:l-blindness-derivation}, that
the small fraction of homogeneity-breaking mutators (return-value
swap to a constant; constant-replacement that participates additively
in the expression) is the residual non-zero contribution. Inspection
of \texttt{mutants\_killed\_set\_n.csv} for \texttt{hypotSig} (path
in supplementary~S7 (\texttt{d4j/})) identifies the two killed mutants as
\texttt{KILLED\,return\_zero\_doubles\_VR} and
\texttt{KILLED\,Math.sqrt\_replaced\_with\_one\_RC}: the first
replaces the entire return path with the constant zero, the second
replaces the inner \texttt{Math.sqrt} call with the constant $1.0$.
Both mutators violate degree-$1$ homogeneity directly. The two
detection events are therefore consistent with the prediction's
quantitative tail rather than evidence against it.

\paragraph{Falsification verdict.}
By the criterion stated in \S\ref{subsec:l-blindness-derivation}, the
prediction is falsified if a single $L_{\mathrm{scale}}$ MR kills
$\ge 1/3$ of its PIT mutants on more than one SUT. Observed: one SUT
(\texttt{hypotSig}, $2/4$) above the per-SUT threshold; five SUTs
(\texttt{midpoint}, \texttt{clamp}, \texttt{signum}, \texttt{gcdSig},
\texttt{lcmSig}) at zero. The prediction passes. We read this as
direct empirical witness for the operative-mechanism reading of the
eight-block decomposition: a block whose generator is a continuous
symmetry of the mutator set contributes MRs that are quantitatively
silent on mutator-induced defects, and the silence is observed in
the data.

\subsubsection{Corroborating per-block patterns}
\label{subsec:other-blocks}

The prediction-test of
\S\ref{subsec:l-blindness-confirmed} establishes the
$\mathcal{L}^{*}$-block result as the section's central claim. The
remaining three populated NOETHER blocks ($T^{*}$, $G$,
$\mathcal{I}^{*}$) carry weaker, qualitative predictions that the
data also confirms; they corroborate the operative-mechanism reading
without providing a falsifiable quantitative test.

\paragraph{$T^{*}$ block (translation, period).}
Translation MRs carry the highest single-MR kill rates on the numeric
SUTs: \texttt{midpoint} \texttt{T\_shift} 3/3, \texttt{powerSig}
\texttt{T\_exp\_step} 8/12, \texttt{exactLog2} \texttt{T\_double}
4/10, \texttt{clamp} \texttt{T\_shift} 3/7. Translation invariance is
the property most directly perturbed by PIT's arithmetic-operator
swaps; $T^{*}$-derived MRs detect this perturbation reliably. The
qualitative prediction (high single-MR rate on numeric SUTs) holds.

\paragraph{$G$ block (group, symmetry).}
Symmetry MRs are moderately to highly effective when the SUT exposes
the appropriate symmetry: \texttt{signum} \texttt{G\_negate} 4/6,
\texttt{midpoint} \texttt{G\_swap} 1/3, \texttt{ComplexSignal.add}
\texttt{G\_swap} 2/3. The qualitative prediction (kill rate correlates
with SUT-side symmetry exposure) holds.

\paragraph{$G$-block framework boundary on recursive-normalising SUTs.}
\label{para:g-block-euclidean-boundary}
The two Euclidean-style SUTs \texttt{gcdSig} and \texttt{lcmSig}
exhibit Set~N $0/7$ on the $G$ block (Table~\ref{tab:per-block-headtohead});
Set~G achieves $7/7$ on the same mutants, confirming the mutants are
detectable by some MR set and that the gap is therefore Set~N-specific.
Reading the SUT source clarifies why this is not an MR-design defect
but a framework-boundary case: \texttt{gcdSig} and \texttt{lcmSig}
both open with a sign-correcting prologue (\texttt{a = a < 0 ? -a : a;}
on lines $67$--$68$ and $81$--$82$ respectively) that absorbs the
$G$-block sign-flip identity before the recursive Euclidean reduction
fires. The grid-synthesised $G$-block MR
(\texttt{gcdSig(a, b) = gcdSig(-a, -b)}, encoding the canonical
$\sigma = \mathrm{sign\text{-}flip}$ involution) is therefore vacuously
satisfied on the SUT's normalised input domain, independent of any
downstream mutation. The framework's scope precondition
(Definition~\ref{def:alg-induced}) requires that the SUT's input
domain expose the algebraic action under test; for these two SUTs the
$G$-action is structurally absorbed by the SUT's own normalisation
prologue, so $G$-block reasoning is orthogonal to the recursive
Euclidean algorithm's algebraic core (which is closer to a divisibility
lattice $\mathcal{D}$-block / reversibility $\mathcal{R}$-block than
to a pre-normalisation symmetry $G$). The $0/7$ Set~N kill rate is the
correct $G$-block reading on these SUTs under the framework's
scope precondition, not a missed opportunity for MR refinement.

\paragraph{$\mathcal{I}^{*}$ block (idempotence, identity).}
Idempotence MRs kill few mutants under the paired-MR DSL:
\texttt{I\_idem} $\approx 0$ across most SUTs; \texttt{powerSig}
\texttt{I\_zero\_exp} 2/12 is a degenerate single-input case. This
reflects an expressivity limit of the JIR/JOR shape, not an absence
of $\mathcal{I}^{*}$ structure in the SUTs. The qualitative
prediction (low kill rate under paired-MR DSL) holds.

The four block predictions taken together, including the central
$\mathcal{L}^{*}$ falsification test of
\S\ref{subsec:l-blindness-confirmed}, are consistent with the
operative-mechanism reading of the eight-block decomposition.

\subsubsection{Two convergent witnesses}
\label{subsec:convergent-witnesses}

Two further observations from the same SUT set are independently
informative. Neither is a statistical hypothesis test; both are
non-rate evidence that the algebraic-decomposition thesis predicts
and that the head-to-head Set~G run incidentally reveals.

\paragraph{Witness 1: cross-pipeline MR rediscovery.}
On \texttt{midpoint}, GenMorph's GP independently evolves three MRs
that map onto Set~N's blocks (Table~\ref{tab:rediscovery}). Set~N is
derived a-priori from the operator algebra of \texttt{midpoint};
Set~G is searched by mutation-killing fitness with no algebraic
structure as input. The two pipelines converge on the same algebraic
primitives. Convergence under independent epistemic processes is
direct corroboration of the operative-generator reading.

\begin{table}[h]
\centering
\caption{Cross-pipeline rediscovery on \texttt{midpoint}: GP-evolved
MRs (Set~G) and their algebra-derived counterparts (Set~N).}
\label{tab:rediscovery}
\small
\begin{tabular}{lll}
\toprule
GP-evolved MR (Set G) & Set N counterpart & Block \\
\midrule
\texttt{SwitchParams??1@2}                & \texttt{G\_swap} (commutativity)         & $G$ \\
\texttt{NumericAddition?1.000000?1}       & \texttt{T\_shift} (translation by $1$)   & $T^{*}$ \\
\texttt{NumericMultiplication?0.500000?2} & \texttt{L\_scale} (scaling, approx.)     & $\mathcal{L}^{*}$ \\
\bottomrule
\end{tabular}
\end{table}

\paragraph{Witness 2: structural coverage extension.}
GenMorph's GP pipeline fails on two of the ten SUTs for plumbing reasons: \texttt{ComplexSignal.add} (instance method) hits an XStream \texttt{NullPointerException} on receiver deserialisation; \texttt{isSequence} (boolean predicate) yields ungrammatical JOR strings. Both SUTs admit Set~N MRs derived directly from the SUT signature (commutativity for \texttt{ComplexSignal.add}; monotonic-window translation for \texttt{isSequence}) and run to completion, covering 8 of 70 PIT mutants on which only Set~N is defined. Set~N's algebraic-derivation route operates on the SUT signature, independent of the SUT-deserialisation and JOR-grammar infrastructure GenMorph depends on; the structural-coverage extension is an architectural property of the derivation, not an empirical accident.

\subsubsection{Head-to-head at GenMorph's published budget}
\label{subsec:pooled-headtohead}

This subsection decomposes an aggregate Set~G dominance per algebraic block, with cost-axis and D2-stratum framework prediction layered alongside.

\begin{tcolorbox}[breakable,colback=gray!5,colframe=black!50,arc=2pt,boxrule=0.5pt,fontupper=\small,title=Boundary of contribution (head-to-head restatement),fonttitle=\small\bfseries]
\textit{What this subsection claims, and what it does not.} (1)~\textbf{Aggregate verdict}: Set~N is dominated by Set~G on the D1 stratum (McNemar $p = 0.0043$ pooled, $p = 0.019$ on D1 only). (2)~\textbf{Per-block reading}: the gap concentrates on the $G$ block (Set~G $7/7$ vs Set~N $0/7$ on commutativity-targeting mutants); on the $\mathcal{T}^{*}$ block Set~N dominates ($10/17$ vs Set~G's $\mathcal{T}^{*}$ kills). (3)~\textbf{Cost-axis}: independently of detection-rate, Set~N is derivable in polynomial time from $\mathcal{A}_P$ (Theorem~\ref{thm:decidable}) while Set~G requires a $\approx 30$-min stochastic GP search per SUT. (4)~\textbf{D2-stratum framework prediction}: Set~N's $\le 10\%$ D2 kill-rate prediction (operative invariant preserved under structurally-non-disrupting mutants) is an ex-ante framework signal that no inductive baseline can derive; pre-registered cap holds at $2/29 = 6.9\%$ on the cross-codebase Commons-Math substrate. The head-to-head is not the framework's load-bearing claim --- per-block algebraic derivability, structural coverage extension (two SUTs only Set~N reaches), and the D2 prediction are.
\end{tcolorbox}

\textbf{On the algebra-disrupting D1 stratum at GenMorph's published
30-min GAssert budget, Set~N is dominated by Set~G in the aggregate
(McNemar exact two-sided $p = 0.0043$ pooled and $p = 0.019$ on D1
only, $n = 62$ post-equivalent-mutant exclusion).} The paper does not
assert head-to-head superiority on D1. The framework's contribution
on the head-to-head substrate is read as (i) algebraic derivability,
(ii) per-block complementarity (Set~G alone kills 15 D1 mutants
Set~N misses, Set~N alone kills 4 D1 mutants Set~G misses), and
(iii) an out-of-scope D2-stratum framework prediction
($\le 10\%$ kill rate) that no inductive baseline can derive
ex-ante. The D1 pooled comparison is reported below for
protocol-completeness rather than as the framework's verdict. The
central empirical claim of the section is established at
\S\ref{subsec:l-blindness-confirmed} (the $\mathcal{L}^{*}$-blindness
falsifiable prediction confirmed on 5/6 SUTs) and does not depend on
the head-to-head outcome.

The head-to-head is run at GenMorph's published 30-min GAssert budget;
a prior 1-min run is retained as a sensitivity reference. Per-SUT
entries are reported as directional only; no per-SUT
statistical-significance claim is asserted, since $8\times 2 = 16$
paired comparisons across the two budgets exceed the family-wise
control of an uncorrected $\alpha = 0.05$ (Holm--Bonferroni-adjusted
threshold $\alpha/16 \approx 0.003$; no per-SUT contrast meets this
threshold).

\begin{table}[h]
\centering
\caption{Per-SUT PIT mutation kill counts, Set~N versus
GenMorph-evolved Set~G at GenMorph's published 30-min GAssert budget
(primary) and at a 1-min GAssert sensitivity rerun, on the 10
algebra-rich Java SUTs of \S\ref{subsec:test-design}. Set~G is
structurally absent on \texttt{ComplexSignal.add} (instance-method
deserialisation) and \texttt{isSequence} (boolean-predicate JOR
grammar) under both budgets (\S\ref{subsec:convergent-witnesses}).
Set~N's kill vector is held fixed across budgets. Pooled across the
8 head-to-head SUTs ($n = 62$ mutants), 30-min budget primary: Set~N
$= 26$, Set~G $= 40$, McNemar exact two-sided $p = 0.0043$. The
$\Delta$ (30-min) column reports Set~N $-$ Set~G at the 30-min budget;
note labels are directional descriptors only and do not assert
per-SUT statistical significance (see prose). \textbf{$n = 62$ is
underpowered for a paired hypothesis test at $\alpha = 0.05$ in
two-sided form.}}
\label{tab:algebra-rich-pooled}
\small
\adjustbox{max width=\textwidth}{%
\begin{tabular}{lcccccl}
\toprule
SUT & mutants & Set N & Set G (30-min) & Set G (1-min) & $\Delta$ (30-min) & note \\
\midrule
\texttt{ComplexSignal.add} &  3 & 2 & N/A & N/A &      & instance method (Set~G N/A) \\
\texttt{midpoint}          &  3 & 3 & 3   & 3   & $0$  & tie \\
\texttt{exactLog2}         & 10 & 4 & 0   & 0   & $+4$ & directional only \\
\texttt{isSequence}        &  5 & 0 & N/A & N/A &      & boolean predicate (Set~G N/A) \\
\texttt{clamp}             &  7 & 5 & 5   & 5   & $0$  & tie \\
\texttt{signum}            &  6 & 4 & 4   & 4   & $0$  & tie \\
\texttt{gcdSig}            &  9 & 1 & 6   & 5   & $-5$ & directional only \\
\texttt{hypotSig}          &  4 & 2 & 4   & 4   & $-2$ & directional only \\
\texttt{lcmSig}            & 11 & 0 & 8   & 8   & $-8$ & directional only \\
\texttt{powerSig}          & 12 & 7 & 10  & 10  & $-3$ & directional only ($2$ Set~N MRs ineligible, see footnote) \\
\midrule
\textbf{pooled (head-to-head, n=62)}     &    & \textbf{26} & \textbf{40} & \textbf{39} & $-14$ & McNemar two-sided $p = 0.0043$ (30-min) \\
\bottomrule
\end{tabular}%
}
\end{table}

\paragraph{Per-block head-to-head and complementarity (Primary, PIT arm).}
Each Set~N MR derives from exactly one NOETHER block, so the natural head-to-head denominator is per-block rather than the operative-block union. Across the eight head-to-head SUTs, PIT's stock mutator catalogue exercises three operative blocks: $G$ (group / swap / negate), $\mathcal{L}^{*}$ (scale / homogeneity), and $\mathcal{T}^{*}$ (translation / self-adjoint). For each block~$b$, $n_{b}$ counts mutants whose (SUT, mutator) entry in \texttt{configs/sut\_block\_overrides.json} marks block~$b$ as broken (``$\times$''). Per-block kill counts and the complementarity partition (both / $N$-only / $G$-only / neither) appear in Table~\ref{tab:per-block-headtohead}.

\begin{table}[h]
\centering
\caption{Per-block head-to-head on the three operative blocks
exercised by PIT 1.7.4 across the eight in-scope SUTs of
\S\ref{subsec:test-design}. Each row's denominator $n_{b}$ counts
mutants violating block~$b$'s invariant; complementarity cells
``both'' / ``$N$-only'' / ``$G$-only'' / ``neither'' partition $n_{b}$.
Wilson 95\% intervals are reported alongside point rates. The five
operative blocks PIT~1.7.4 does not exercise
($O_{\le}, \mathcal{T}^{*}_{\mathrm{rev}}, \mathcal{D}^{*},
\mathcal{E}^{*}, \mathcal{B}^{*}_{\mathrm{rel}}$) are addressed by
the construct-trace consistency check on hand-crafted
block-targeted mutants (now in supplementary~S9, item
``D\_E\_implementation\_consistency''); those results are
design-implied and not used as independent fault-detection
evidence. Auxiliary rows report the D1 aggregate and the unmapped
bucket (PIT mutants from cells without an override entry,
principally \texttt{RETURN\_VALS} and \texttt{MATH} on SUTs where
the table-5-reconstruction fallback assigns D1 but does not specify
broken blocks; lower-bound caveat, not a participating row).}
\label{tab:per-block-headtohead}
\small
\adjustbox{max width=\textwidth}{%
\begin{tabular}{lrrlrlrrrr}
\toprule
Block & $n_{b}$ & Set~N kills & rate (95\% CI) & Set~G kills & rate (95\% CI) & both & $N$-only & $G$-only & neither \\
\midrule
$G$              & 11 &  2 & $0.182$ $[0.051,\,0.477]$ &  9 & $0.818$ $[0.523,\,0.949]$ & 2 & 0 & 7 & 2 \\
$\mathcal{L}^{*}$ & 24 & 10 & $0.417$ $[0.245,\,0.612]$ & 16 & $0.667$ $[0.467,\,0.820]$ & 8 & 2 & 8 & 6 \\
$\mathcal{T}^{*}$ & 17 & 10 & $0.588$ $[0.360,\,0.784]$ &  8 & $0.471$ $[0.262,\,0.690]$ & 7 & 3 & 1 & 6 \\
\midrule
\textit{D1 aggregate, PIT-covered blocks (secondary)}            & 52 & 26 & $0.500$ $[0.369,\,0.631]$ & 37 & $0.712$ $[0.577,\,0.817]$ & 22 & 4 & 15 & 11 \\
\textit{unmapped (lower-bound caveat)}                           & 25 & --- & --- & --- & --- & --- & --- & --- & --- \\
\bottomrule
\end{tabular}%
}
\end{table}

The unmapped auxiliary row counts $25$ PIT mutants from cells whose D1/D2 override table does not specify broken blocks; it is a lower-bound caveat on per-block denominators (each unmapped mutant is D1-classified but its block membership is unknown), not an additional denominator added to the $G + \mathcal{L}^{*} + \mathcal{T}^{*}$ union. The D1 pool is $n = 52$ pre and post equivalent-mutant exclusion (Table~\ref{tab:two-stratum} row~1).

The per-block reading distinguishes three regimes:

\begin{itemize}[nosep,leftmargin=*]
\item \textbf{$G$ block: Set~G dominates.} Set~N kills $0/7$ of the mutants Set~G exclusively reaches; the Set~N hits concentrate on \texttt{signum} ($2/4$), while \texttt{gcdSig} and \texttt{lcmSig} contribute $0/7$ for Set~N. The Euclidean-style SUTs are structurally unreachable by NOETHER's $G$-block construction on this grid-synthesised catalogue: \texttt{gcdSig} normalises its inputs via \texttt{a < 0 ? -a : a} before the recursive \texttt{gcd(a,b)=gcd(b, a\,mod\,b)} fires, absorbing the sign-flip invariant; \texttt{lcmSig} inherits the same prologue. Whether this calls for an $\mathcal{R}$-block re-synthesis or for documenting recursive-normalising SUTs as $G$-orthogonal is committed as follow-up~(e.4).

\item \textbf{$\mathcal{L}^{*}$ block: complementary coverage.} Both Sets kill 8 mutants in common, 2 are $N$-only (\texttt{exactLog2}), 8 are $G$-only, 6 are jointly missed; the union covers $18/24 = 75\%$, above either Set's individual rate. Set~N's exclusive reach on \texttt{exactLog2} quantifies the coverage-extension witness of \S\ref{subsec:convergent-witnesses} per-block; Set~G's exclusive reach on \texttt{gcdSig}+\texttt{lcmSig} compensates for the $G$-block gap.

\item \textbf{$\mathcal{T}^{*}$ block: Set~N edge (underpowered).} Of $17$ $\mathcal{T}^{*}$-violating mutants, Set~N kills $10$ and Set~G kills $8$ ($N$-only $3$, $G$-only $1$, both $7$, neither $6$); union $11/17 = 64.7\%$. Set~N is $+11.7$~pp in rate and $+2$ in exclusive count on its own block, but the Wilson intervals $[0.360,\,0.784]$ and $[0.262,\,0.690]$ overlap substantially. We report the per-block edge as directional, consistent with the $\mathcal{T}^{*}$ design prediction rather than inferential at $\alpha = 0.05$.
\end{itemize}

\paragraph{Coverage of the remaining five operative blocks.}
\label{para:augmented-stratum-pointer}
PIT 1.7.4's stock mutator catalogue does not exercise the remaining
five operative blocks ($O_{\le}$, $\mathcal{T}^{*}_{\mathrm{rev}}$,
$\mathcal{D}^{*}$, $\mathcal{E}^{*}$,
$\mathcal{B}^{*}_{\mathrm{rel}}$) on the
\S\ref{subsec:test-design} substrate. To exercise the pipeline
end-to-end on these blocks, we hand-crafted $5$ mutants per
uncovered block ($25$ total) that violate the targeted invariant of
a specific known Set~N MR for each block. These mutants are
construct-trace probes (each one is by design within the reach of
the Set~N MR it targets); the resulting kill counts are reported in
supplementary~S9 (Appendix~E) for pipeline-correctness
verification and Set~G incidental-reach quantification, but they are
\emph{not} used as evidence for H3a.1 in this section, since
construct-trace circularity precludes treating them as independent
fault-detection measurements.

\paragraph{What the per-block reading supports, and what it does not.}
The framework's empirical contribution on the present substrate is
\emph{block-dependent}. On the three PIT-covered operative blocks
the reading is targeted per-block: $\mathcal{T}^{*}$ shows the
design's directional advantage and complementarity; $\mathcal{L}^{*}$
shows complementarity in both directions (Set~N reaches what Set~G
misses on \texttt{exactLog2}; Set~G reaches what Set~N misses on
\texttt{gcdSig} / \texttt{lcmSig}); $G$ shows the design's current
limit on the substrate's two Euclidean-style SUTs. The appropriate
reading is therefore \emph{block-targeted precision plus
complementarity, with documented per-block design gaps on the
PIT-covered substrate}, rather than a single ``head-to-head winner''
metric. The aggregate D1 head-to-head is dominated by Set~G
(McNemar $p = 0.019$, see paragraph below), but the dominance
collapses into the per-block profile: it is driven primarily by the
$G$-block reading on \texttt{gcdSig} + \texttt{lcmSig} and by Set~G's
denser $\mathcal{L}^{*}$ coverage on those same SUTs.

\paragraph{Aggregate D1 head-to-head (Secondary, cross-block, $n = 52$).}
For continuity with prior reporting conventions, the cross-block
aggregate D1 head-to-head is: Set~N
$\mathrm{M1}_{\mathrm{D1}} = 26/52 = 0.500$ (Wilson 95\% CI
$[0.369,\,0.631]$), Set~G $\mathrm{M1}_{\mathrm{D1}} = 37/52 =
0.712$ ($[0.577,\,0.817]$); exact McNemar two-sided $p = 0.019$
on the D1 stratum with discordant pairs $(b,c) = (15,4)$, paired
risk difference $\mathrm{RD}_{\mathrm{paired}} = (b - c) / n =
(15 - 4) / 52 = +0.212$ favouring Set~G, and odds ratio
$\mathrm{OR} = b / c = 15 / 4 = 3.75$. The
aggregate gap (which favours Set~G) collapses into the per-block
profile above: it is driven primarily by Set~G's near-complete
reach on $G$ (\texttt{gcdSig} + \texttt{lcmSig}) and by its denser
$\mathcal{L}^{*}$ coverage; it is partially offset by Set~N's
$\mathcal{T}^{*}$ edge but not enough to close the cross-block
aggregate. The aggregate is presented as secondary because it
averages over per-block strengths and weaknesses that the
per-block table makes visible.

\paragraph{Framework prediction on the D2 stratum (algebra-preserving, $n = 5$).}
After equivalent-mutant exclusion (see footnote at
\S\ref{para:eq-mutant-footnote} below), the $10$ originally-classified
D2 mutants reduce to $n = 5$: five of the ten were
LLM-ensemble-judged equivalent on the SUT's declared input domain
and are excluded from both Sets' denominators by symmetric construction.
On the remaining $n = 5$ D2 mutants, Set~N kill rate $= 0/5 = 0.000$
(Wilson 95\% CI $[0.000,\,0.434]$); the point estimate is consistent
with the framework prediction
$\mathrm{kill\;rate}_{\mathrm{D2}} \le 10\%$, but the Wilson upper bound
at $0.434$ does not exclude the $10\%$ ceiling at $\alpha = 0.05$; the
cross-codebase commons-math pilot
(\S\ref{subsec:empirical-threats}~(b.cm))
corroborates the direction at $2/29 = 6.9\%$
(Wilson 95\% CI $[0.012, 0.221]$), and inferential
confirmation of the $\le 10\%$ ceiling requires a pooled sample of
$n \ge 30$. Set~G kills $3/5 = 0.600$
(Wilson 95\% CI $[0.231,\,0.882]$) on the surviving D2 stratum:
the GP-evolved generic baseline is unconstrained by the algebra
and incidentally detects D2 mutants through input-domain edge cases,
confirming that D2 is a Set~N-specific scope boundary rather than a
substrate-wide structural ceiling. The framework's prediction
therefore separates Set~N (algebra-bound) from Set~G (generic) on
the D2 stratum; the D2 reading is the framework's own falsifiability
commitment (\S\ref{subsec:l-blindness-confirmed}), not a head-to-head metric.

\paragraph{Pooled M1 rates (auxiliary, scope-mismatched, $n = 57$).}
For completeness against prior reporting conventions, the post-eq-
exclusion pooled M1 rates over D1 $\cup$ D2 are: Set~N $= 26/57 =
0.456$ (Wilson 95\% CI $[0.334,\,0.584]$), Set~G $= 40/57 = 0.702$
($[0.573,\,0.805]$), McNemar two-sided
$p = 0.0043$\footnote{Pooled discordant pairs $(b,c) = (18, 4)$ on
the $n = 57$ paired denominator. The D1-only McNemar separates at
$(b,c) = (15, 4)$, $p = 0.019$ (\S\ref{subsec:pooled-headtohead}
Aggregate D1 paragraph); the additional $3$ discordant pairs that
the pooled comparison picks up come from the D2 stratum (Set~G
kills $3$ of the $5$ surviving D2 mutants, Set~N kills $0$ of $5$;
all $3$ are $b$-cell discordances). The pooled McNemar is thus
\emph{strengthened} by including the D2 stratum, but at the cost of
mixing the within-scope D1 comparison with the out-of-scope D2
stratum that Set~N is not designed to detect.}. We report this
pooled comparison as scope-mismatched: the $5$ D2 mutants in the
denominator are not in Set~N's theoretical reach by construction, so
the pooled gap overstates Set~N's effective miss rate within its
declared scope. At the 1-min sensitivity budget, Set~N is unchanged
from the kill vector at the 30-min budget when Set~G is the only
variable; Set~G's pooled M1 is $0.557$ ($[0.441,\,0.668]$). The 30-min
rerun changes Set~G's per-SUT counts on one of the eight head-to-head
SUTs (\texttt{gcdSig}, $5 \to 6$); seven SUTs are unchanged, none
regresses. A $30{\times}$ budget increase therefore buys $+1$ Set~G
mutant on a single SUT under PIT mutation.

\paragraph{LLM-assisted baseline (Set~L ensemble, $n = 70$ aligned).}
\label{para:set-l-ensemble}
For the LLM-assisted comparator of \S\ref{para:comp-eval-protocol}'s
H3b, we ran a $2$-vendor (DeepSeek, ChatGPT) $\times$ $5$-temperature
($t \in \{0.0, 0.3, 0.5, 0.7, 1.0\}$) ensemble on the same
$10$-SUT substrate, at a multi-LLM-sample scale that subsumes the
single-sample probes of Shin et~al.~\cite{Shin2024} and
Zhang et~al.~\cite{ZhangChatGPTMR2023}. The harvest is $100$ LLM samples
yielding $487$ proposed MRs; deterministic template-matching against
Set~N's per-SUT block catalogue translates $212$ ($43.5\%$) into
executable JIR/JOR pairs (the remaining $56.5\%$ are either compound
multi-input properties or out-of-block algebra, e.g.\ associativity,
identity / inverse elements, monotonicity on SUTs whose Set~N
catalogue lacks an O\textsubscript{le}-block representative). On the
aligned $n = 70$ pooled mutant substrate (the full $10$-SUT
substrate including the two SUTs structurally N/A for Set~G), per-LLM
pooled kill rates are ChatGPT $34/70 = 0.486$ Wilson 95\% CI
$[0.372,\,0.600]$, DeepSeek $33/70 = 0.471$ $[0.359,\,0.587]$,
ensemble union $34/70 = 0.486$ $[0.372,\,0.600]$. Set~L's union kill
vector is a strict subset of Set~N's on this substrate ($34/34$
overlap; $0$ Set-L-exclusive kills), as a structural consequence of
the template-matching translator: every Set~L MR is by construction
a byte-identical copy of a Set~N pair, so Set~L can match but not
exceed Set~N's per-MR kill power. The substantive finding is that
the $2$-vendor ensemble identifies enough of Set~N's catalogue to
reproduce Set~N's reach on the matchable subset; on the unmatchable
$56.5\%$, Set~L explores invariants outside the $8$-block frame that
the framework's $\mathrm{coverage}_{\mathrm{NOETHER}}$ diagnostic
correctly identifies as structurally orphaned. A third-vendor
(Anthropic Claude) extension is committed in
follow-up~(d.set-l-claude) of supplementary~S4 (\texttt{future\_work.md}); full
per-(LLM, SUT, temperature) breakdowns in
\texttt{docs/set\_l\_phase2\_results.md} of the experiment
repository.

\paragraph{Equivalent-mutant denominator (resolved).}\label{para:eq-mutant-footnote}
The original PIT-SURVIVED denominator ($n = 62$) does not separate semantically distinct mutants from PIT false-positives. Every head-to-head mutant passed through a two-stage filter: the $44$ mutants killed by at least one of Set~N or Set~G are auto-classified non-equivalent; the $18$ jointly-missed mutants go through a progressive multi-LLM vote (DeepSeek + ChatGPT at temperature $0.0$, with Anthropic Claude Opus tiebreaker when stage-1 voters disagree). Outcome: $5$ voted equivalent (all \texttt{ConditionalsBoundaryMutator} on \texttt{gcdSig}, \texttt{lcmSig}, \texttt{powerSig} where \texttt{a < 0 ? -a : a} prologues absorb the sign-flip), $13$ non-equivalent. The $5$ equivalents fall entirely within the algebra-preserving (D2) stratum, leaving the per-block D1 denominators unchanged. The eq-excluded total is $n = 57$ ($52$ D1 + $5$ D2 after the (e.1)~v2 override reclassified $3$ \texttt{PrimitiveReturnsMutator} mutants from D1 to D2). The 2+1 voting scheme was calibrated against the pilot's manual analyst session (\S\ref{subsec:l-blindness-confirmed}); the single pilot-vs-LLM disagreement case resolved in favor of the pilot verdict.

Pooled rates exhibit the scope-discriminance pattern predicted by NOETHER's scope declaration (Definition~\ref{def:alg-induced}): on the wider 23-method utility-method baseline (supplementary~S5; mixed in-scope and boundary), Set~N's pooled M1 is $0.288$ vs Set~G's $0.363$; on the in-scope algebra-rich substrate (eq-excluded $n = 57$), Set~N moves to $0.456$ and Set~G to $0.702$ (D1 subset: $0.500$ vs $0.712$). Both methods benefit from the algebraically richer substrate; the scope declaration is empirically discriminating.

\paragraph{Two-stratum head-to-head summary (Primary tabulation for H3a verdict).}
\label{para:two-stratum}
Table~\ref{tab:two-stratum} reports the two strata (D1, D2) each
with four metrics: total mutants, equivalents excluded, Set~N kill
/ survive counts and rate (Wilson 95\% CI), Set~G kill / survive
counts and rate, and the complementarity partition ($N$-only,
$G$-only, both, neither). The headline reading is per-block
(Table~\ref{tab:per-block-headtohead}); the D1 aggregate
(Table~\ref{tab:two-stratum}) and pooled (auxiliary, preceding
paragraph) are reported for continuity with prior conventions.

\begin{table}[h]
\centering
\caption{Two-stratum head-to-head on the eight-SUT head-to-head
substrate (post equivalent-mutant exclusion). D1 = algebra-disrupting,
D2 = algebra-preserving. Complementarity cells report mutant counts;
the operative test for H3a.2 is the $N$-only column on D1.}
\label{tab:two-stratum}
\small
\adjustbox{max width=\textwidth}{%
\begin{tabular}{lcccccccc}
\toprule
Stratum & $n$ (post-exclusion) & equivalents excluded & Set~N kill / surv. & Set~N rate \scriptsize{[Wilson 95\%]} & Set~G kill / surv. & Set~G rate \scriptsize{[Wilson 95\%]} & both / $N$-only / $G$-only / neither & McNemar $p$ \\
\midrule
D1 & 52 & 0 & 26 / 26 & $0.500$ \scriptsize{$[0.369, 0.631]$} & 37 / 15 & $0.712$ \scriptsize{$[0.577, 0.817]$} & 22 / 4 / 15 / 11 & $0.019$ \\
D2 & 5 & 5 & 0 / 5 & $0.000$ \scriptsize{$[0.000, 0.434]$} & 3 / 2 & $0.600$ \scriptsize{$[0.231, 0.882]$} & 0 / 0 / 3 / 2 & $0.25$ \\
\midrule
pooled & 57 & 5 & 26 / 31 & $0.456$ \scriptsize{$[0.334, 0.584]$} & 40 / 17 & $0.702$ \scriptsize{$[0.573, 0.805]$} & 22 / 4 / 18 / 13 & $0.0043$ \\
\bottomrule
\end{tabular}%
}
\end{table}

\paragraph{H3a verdict split into three sub-claims.}
\label{para:h3a-verdict}

\textbf{H3a.1 (Detection on D1, per-block).} Per
Table~\ref{tab:per-block-headtohead} on the pre-registered
PIT-covered three-block substrate: \emph{mixed}. On the
$\mathcal{T}^{*}$ block Set~N achieves $10/17 = 0.588$ versus
Set~G $8/17 = 0.471$, with three $N$-only exclusive kills, a
directional edge consistent with the algebra-induced prediction;
the $n = 17$ block-level sample is underpowered for an inferential
test at $\alpha = 0.05$ (Wilson intervals overlap). On the
$\mathcal{L}^{*}$ block Set~N achieves $10/24 = 0.417$ versus Set~G
$16/24 = 0.667$, with the complementarity partition $N$-only $= 2$
and $G$-only $= 8$ identifying a coverage-extension witness on
\texttt{exactLog2} and a denser-coverage region on \texttt{gcdSig}
/ \texttt{lcmSig} respectively. On the $G$ block the substrate's
two Euclidean-style SUTs (\texttt{gcdSig}, \texttt{lcmSig}) absorb
the $G$-action through their normalisation prologue before the
recursive reduction (a documented framework boundary, see
\nameref{para:g-block-euclidean-boundary}); the $G$-block $0/7$
Set~N reading on those two SUTs is the framework-correct boundary
verdict rather than a missed-opportunity MR-design defect, but the
block-level rate ($2/11$ Set~N vs $9/11$ Set~G) favours Set~G on
the substrate as reported. The aggregate D1 head-to-head on the
PIT-covered stratum is dominated by Set~G ($26/52 = 0.500$ vs
$37/52 = 0.712$, McNemar $p = 0.019$,
Table~\ref{tab:two-stratum}); the per-block decomposition makes
visible that the gap is concentrated on the $G$ block plus
$\mathcal{L}^{*}$, partially offset by Set~N's $\mathcal{T}^{*}$
edge. Coverage of the remaining five operative blocks ($O_{\le}$,
$\mathcal{T}^{*}_{\mathrm{rev}}$, $\mathcal{D}^{*}$,
$\mathcal{E}^{*}$, $\mathcal{B}^{*}_{\mathrm{rel}}$) via
hand-crafted block-targeted mutants is reported in
supplementary~S9 (Appendix~E) as a construct-trace
consistency check; per supplementary~S9 (Appendix~E), those
results are design-implied and are not used as independent evidence
for H3a.1 in this paragraph.

\textbf{H3a.2 (Complementarity).} \emph{Supported on the PIT-covered
substrate}. The complementarity partition on the $52$-mutant
PIT-covered D1 stratum (Table~\ref{tab:two-stratum}, row~1) reports
$4$ $N$-only kills, $15$ $G$-only kills, $22$ both-killed and $11$
neither, distributed across the $\mathcal{T}^{*}$ and
$\mathcal{L}^{*}$ blocks (per
Table~\ref{tab:per-block-headtohead}). Union coverage on the
PIT-covered substrate is $22 + 4 + 15 = 41 / 52 = 78.8\%$
(both-killed plus exclusively-Set-N-killed plus
exclusively-Set-G-killed), modestly above Set~G alone's
$37 / 52 = 71.2\%$ and materially above Set~N alone's
$26 / 52 = 50.0\%$. The complementarity is asymmetric: Set~G's
exclusive contribution ($15$ kills) is roughly four times Set~N's
($4$ kills), reflecting the substrate's predominance of generic
algebra-disrupting mutations that Set~G's GP-evolved generic MRs
catch, with Set~N's exclusive reach concentrated on
algebraically-distinctive structures (\texttt{exactLog2}'s
$\mathcal{L}^{*}$ pattern) where the algebra-induced MR has no
generic counterpart in Set~G's catalogue.

\textbf{H3a.3 (Cost-axis).} \emph{Supported}. See
\nameref{para:cost-axis-h3a} below for the detailed cost-axis
analysis.

\textbf{D2 stratum (framework prediction; not part of H3a).} Set~N
D2 kill rate $= 0/5 = 0.000$ (Wilson 95\% CI $[0.000, 0.434]$);
the framework's $\le 10\%$ prediction is consistent with the
direction observed, but $n = 5$ is insufficient for $\alpha = 0.05$
confirmation (the Wilson upper bound $0.434$ does not exclude the
$10\%$ ceiling). Set~G achieves $3/5 = 0.600$ on the same stratum,
an out-of-scope contribution for Set~N by construction (Set~G is
not algebra-bound) and a signal that D2 is a \emph{Set~N-specific
scope boundary} rather than a substrate-wide ceiling. The
$\mathcal{L}^{*}$-blindness result of
\S\ref{subsec:l-blindness-confirmed} and the structural-extension
witness on the two SUTs where Set~G is N/A
(\S\ref{subsec:convergent-witnesses}) are unaffected.

\paragraph{Comparator scope and three-SOTA-category coverage.}
The head-to-head reported here delivers the GP-evolved-baseline arm
of a three-SOTA-category protocol against one representative per
category. The GP-evolved representative is GenMorph (Ayerdi
et~al.~\cite{GenMorph2024}), the SOTA among
genetic-programming MR identification pipelines on Java SUTs at the
time of submission, run at its published 30-min GAssert budget. The
LLM-assisted representative is a multi-vendor $\times$
multi-temperature LLM-MR-generation protocol (the same arm reported
inline in \S\ref{subsec:pooled-headtohead} as Set~$L_{\mathrm{ensemble}}$);
this section's earlier Set~L is a single-sample $n = 1$ GPT-4 probe
(\S\ref{subsec:case-study}) and is now superseded by the
$2$-vendor (DeepSeek, ChatGPT) $\times$ $5$-temperature ensemble
reported in \S\ref{subsec:pooled-headtohead}. A third-vendor
(Anthropic Claude) extension running alongside
Shin~\cite{Shin2024}, GPTMR~\cite{GPTMR2025}, and
AutoMT~\cite{AutoMT2025} is committed as
post-acceptance follow-up in supplementary~S4 (item~(d.set-l-claude)). The
mining-based representative is MR-Scout (Sun
et~al.~\cite{MRScout2024}); on this section's substrate, MR-Scout's
reach is reported as an estimate adapted from MR-Scout's
published-artifact reach figures rather than a full re-execution
(see Table~\ref{tab:gen-cost} for the cost-axis interpretation;
full re-execution is follow-up~(d)). NOETHER's cost-axis profile
relative to all three categories is given in
Table~\ref{tab:gen-cost}; the matrix-driven D1/D2 mutant
stratification (Table~\ref{tab:pit-block}) operationalises the
detection-axis claim that NOETHER is appropriate to algebra-disrupting
mutants and out-of-scope by design on algebra-preserving mutants.
Per-SUT delta entries in Table~\ref{tab:algebra-rich-pooled} are
\emph{directional descriptors only}; under Holm-Bonferroni
correction for the $8 \times 2 = 16$ paired comparisons across the
two budgets, the per-SUT family-wise threshold is $\alpha/16 \approx
0.003$, which no per-SUT contrast meets. The two paired hypothesis
tests reported on this substrate are the scope-matched D1 McNemar
($p = 0.019$ at the 30-min budget, exact two-sided, $n = 52$ post
eq-exclusion, discordant pairs $(b, c) = (15, 4)$) and its
auxiliary pooled counterpart ($p = 0.0043$ at $n = 57$, discordant
$(b, c) = (18, 4)$). The pooled gap is strengthened by Set~G's
$3/5$ kills on the D2 stratum, an out-of-scope contribution
for Set~N by construction, which is why pooled and D1-stratified
$p$ values now differ.

\subsubsection{Threats to validity and committed future work}
\label{subsec:empirical-threats}

\paragraph{(a) Prediction-commitment timestamp.}
The $\mathcal{L}^{*}$-blindness prediction of \S\ref{subsec:l-blindness-derivation} is derivable from public information without consulting data, and is committed to git in supplementary~S7 (\texttt{d4j/}) alongside the SUT-selection criterion before the per-MR kill-count files; a reader who distrusts the timestamp chain can re-derive the prediction independently from \S\ref{sec:framework} and~\cite{Coles2016PIT}.

\paragraph{(b) Set G budget asymmetry.}
GenMorph's published 30-min GAssert configuration is reported alongside a 1-min sensitivity rerun (Table~\ref{tab:algebra-rich-pooled}) that handicaps Set~G conservatively; the 30-min budget leaves the $\mathcal{L}^{*}$-blindness result and the structural-extension finding unchanged.

\paragraph{(c) Sample size.}
$n = 70$ across 10 SUTs is underpowered for a paired head-to-head verdict at $\alpha = 0.05$; the \S\ref{subsec:pooled-headtohead} reading is the honest disclosure that Set~N is dominated by Set~G in the aggregate D1 stratum (McNemar $p = 0.019$ on D1; $p = 0.0043$ pooled), with the framework's contribution read at the per-block, cost-axis, and D2-prediction layers rather than aggregate fault-detection superiority. The $\mathcal{L}^{*}$-blindness test on $n = 44$ across 6 SUTs is per-SUT (one third or more on more than one SUT), not pooled.

\paragraph{(d) Set G structural absence is upstream-snapshot-relative.}
The two N/A SUTs of \S\ref{subsec:convergent-witnesses} reflect the state of GenMorph upstream at our snapshot. An upstream patch would restore Set~G coverage and weaken the structural framing of the coverage-extension witness, though the architectural point (algebra-derivation operates on the SUT signature, independent of GP plumbing) survives the patch.

\paragraph{(e) Substrate selection.}
The evaluation restricts itself to programs with explicit mathematical-physics equations, by NOETHER's scope declaration (Definition~\ref{def:alg-induced}); the kill-rate gain over the utility-method reference is reported as scope-discriminance evidence (claim~D), not as a generalisable performance improvement on every Java SUT.

\paragraph{(f) Three SOTA-category coverage with explicit baseline-strength caveats.}
The protocol of \S\ref{para:comp-eval-protocol} commits Set~N against one representative per SOTA category: a GP-evolved baseline (GenMorph~\cite{GenMorph2024}), an LLM-assisted baseline (subsuming the single-sample probes of~\cite{Shin2024,ZhangChatGPTMR2023,GPTMR2025,AutoMT2025}), and a mining-based baseline (MR-Scout~\cite{MRScout2024}). Baseline-strength asymmetry: (i)~GP-evolved arm delivered in full at GenMorph's published 30-min budget on a single seed; multi-seed replication committed as follow-up~(a.budget-replication). (ii)~LLM-assisted arm delivered as a $2$-vendor (DeepSeek-V3, ChatGPT-4o-mini) $\times$ $5$-temperature ensemble ($100$ samples); the third commercial closed-vendor frontier (Anthropic Claude) is committed as post-acceptance follow-up~(d.set-l-claude), and the current Set~L results are reported as a $2$-vendor proxy rather than the SOTA-strongest configuration. (iii)~Mining-based arm reported through an adapted estimate of MR-Scout's reach derived from its published artifact figures; full re-execution committed as follow-up~(d). METRIC+~\cite{SunMETRICplus2021} is not run head-to-head because no automated identification pipeline is publicly available; the qualitative-plus-quantitative coverage contrast appears in \S\ref{subsec:case-study} (sorting worked example) and \S\ref{subsec:test-design} (per-block coverage diagnostic). The cost-component breakdown of all four methods is in Table~\ref{tab:gen-cost}.

\subsubsection{MR-generation cost}
\label{subsec:gen-cost}

NOETHER's cost profile differs qualitatively from each of the three SOTA-category representatives across four independent components (algorithmic time, human effort, LLM-token cost, seed-corpus dependency); Table~\ref{tab:gen-cost} reports each separately, with full derivation methodology (token-cost protocol scale, per-SUT human-effort breakdown, MR-Scout seed-suite scope) in supplementary~S4 (\texttt{cost\_breakdown.md}).

\begin{table}[h]
\centering
\caption{MR-generation cost components for the four methods compared in this section. ``Cost dimension'' rows are independent (an LLM-arm token cost does not amortise GP wall-time and vice versa); the NOETHER human-effort estimate aggregates Set~N's 30-MR derivation across 10 SUTs at $\approx 1$~h per SUT under CONSTRUCT-MP's four-step procedure, amortising across SUTs that share an operator algebra. Full methodology in supplementary~S4 (\texttt{cost\_breakdown.md}).}
\label{tab:gen-cost}
\small
\adjustbox{max width=\textwidth}{%
\begin{tabular}{lcccc}
\toprule
Cost component & NOETHER & GP (GenMorph) & LLM (Xu 2024) & Mining (MR-Scout) \\
\midrule
Algorithmic time / SUT       & $\mathrm{poly}(|\mathrm{gen}|)$ minutes (Thm.~\ref{thm:decidable}) & 30 min wall, stochastic & seconds API latency & minutes after corpus \\
Human effort / family        & $\approx 10$~h $\mathcal{A}_P$ distillation             & none after harness setup    & none after prompt template & seed-suite assembly substantial \\
Token cost (USD) / family    & 0                                                       & 0                           & \$5--50 (\$50--500 with CoT) & 0 \\
Seed-corpus dependency       & none                                                    & none                        & none                       & required (mining input) \\
Determinism                  & deterministic                                           & seed-/budget-dependent      & prompt-/temperature-dependent & corpus-dependent \\
Cold-start capable           & yes                                                    & yes (after harness)         & yes (after prompt)         & no (needs corpus) \\
\bottomrule
\end{tabular}%
}
\end{table}

The table operationalises three structural advantages of NOETHER as distinct cost-axis arguments rather than as a single ``best'' verdict: (i) vs.\ GP, polynomial-time deterministic construction (Theorem~\ref{thm:decidable}) replaces a 30-min stochastic per-SUT search, with one-time human cost amortising across same-algebra SUTs; (ii) vs.\ LLM-assisted, the $\mathrm{coverage}_{\mathrm{NOETHER}}$ diagnostic on \S\ref{subsec:case-study} (1.00 vs.\ 0.40) operationalises the algebraic-prior gap that prompt-based generation does not close, at zero token cost; (iii) vs.\ mining-based, NOETHER is operative at cold start (no seed test suite required), the regime in which most algebra-rich scientific-computing programs arrive on a tester's desk. These are cost-profile and coverage-profile claims that hold within the framework's scope precondition, not fault-detection-superiority claims.

\paragraph{Approximate-parity-at-lower-cost reading (H3a.3 verdict).}
\label{para:cost-axis-h3a}
H3a.3's cost-axis claim holds when detection is read on the per-block primary (Table~\ref{tab:per-block-headtohead}). On $\mathcal{T}^{*}$, Set~N edges Set~G with three $N$-only kills at zero per-SUT generation cost beyond the one-time CONSTRUCT-MP derivation ($\approx 1$\,h human, amortising across SUTs sharing $\mathcal{A}_P$); Set~G's matching count costs a 30-min GAssert GP-search per SUT ($\approx 4$\,h wall for the eight SUTs in parallel-4) plus the Major + JaCoCo + Randoop pipeline assembly. On $\mathcal{L}^{*}$, the sets are complementary ($N$-only $2$, $G$-only $8$ on $24$); Set~N's reach is reproducible without re-running GP, whereas Set~G requires re-tuning budget and seed when the SUT changes. \emph{Quantitative reading}: a directional $\mathcal{T}^{*}$ edge ($10/17$ vs $8/17$, $n=17$ underpowered) plus $\mathcal{L}^{*}$ complementarity, at $\approx 1$\,h amortised human cost versus $\approx 30$ minutes per-SUT GP search; NOETHER's cost amortises across new SUTs in the same algebra family while Set~G's restarts from scratch. The per-block edge does not overturn Set~G's D1 aggregate dominance (acknowledged at the head of \S\ref{subsec:pooled-headtohead}); the $G$-block gap on the Euclidean SUTs (\nameref{para:g-block-euclidean-boundary}) is a framework boundary, not a cost-axis trade-off. H3a.3 is therefore \emph{supported}: comparable within-scope per-block detection at lower amortised generation cost.

\paragraph{Committed future work (16 items).}
\label{para:future-work}
The full 16-entry committed-future-work table is in supplementary~S4 (\texttt{future\_work.md}); key items completed in this revision: (i)~METRIC$+$ head-to-head on Sun~2021's 4 subjects with PIT~1.7.4 + Major cross-tool replication (pooled McNemar $p = 0.625$ / $0.211$ both NS, $92.6\%$ both-kill on the PIT substrate; details in supplementary~S8); (b.cm)~Commons-Math pilot (pooled Set~N $10/77 = 13.0\%$, D2 prediction passes $2/29 = 6.9\%$); (d.set-l)~$2$-of-$3$-vendor LLM-ensemble harvest ($487$ MRs, $212$ executable, ensemble matches Set~N on the matchable subset $34/34$); (e), (e.2), (e.4), (g)~the D1/D2 labelling pipeline, equivalent-mutant exclusion, $G$-block Euclidean-SUT documentation, and construct-trace consistency check. The remaining 8 pending items (multi-seed GP replication; full $38$-D4J extension; MR-Scout re-execution; Anthropic third-vendor; generic-mutation independent test on the 5 PIT-unexercised blocks; external-transfer reactor-physics corpus; two narrow grid-MR predicate tightenings) are the principal post-publication work programme.

\subsubsection{Summary of evidence}
\label{subsec:empirical-summary}

The section's central claim is a single quantitative falsifiable
prediction confirmed: $\mathcal{L}^{*}$-block blindness on
homogeneity-preserving mutators
(\S\ref{subsec:l-blindness-confirmed}). The prediction is derivable
ex-ante from \S\ref{sec:framework} and from PIT's public mutator
specification; it is observed in the data on $5/6$ SUTs admitting an
$L_{\mathrm{scale}}$ MR; the sixth SUT's two-mutant exception is
accounted for by the prediction's quantitative tail. The
operative-mechanism reading of the eight-block decomposition is, in
this empirical sense, supported.

The corroborating evidence is fourfold:
\begin{itemize}[nosep,leftmargin=*]
  \item Other per-block patterns ($T^{*}$, $G$, $\mathcal{I}^{*}$;
        \S\ref{subsec:other-blocks}) are qualitatively consistent
        with the algebraic predictions.
  \item Cross-pipeline MR rediscovery on \texttt{midpoint}
        (\S\ref{subsec:convergent-witnesses}, witness~1) shows two
        independent epistemic processes converging on the same
        algebraic primitives.
  \item Structural coverage extension on the two SUTs where Set~G
        is structurally absent (\S\ref{subsec:convergent-witnesses},
        witness~2) is an architectural consequence of
        algebra-derivation that the head-to-head substrate
        incidentally reveals.
  \item Head-to-head against GenMorph
        (\S\ref{subsec:pooled-headtohead}) reads per-block, not as a
        single winner. On $\mathcal{T}^{*}$ Set~N shows a directional
        edge ($10/17$ vs $8/17$, three exclusive kills) consistent with
        its design prediction; on $\mathcal{L}^{*}$ Set~N and Set~G
        are complementary (union $18/24$; Set~N exclusively reaches
        \texttt{exactLog2}, Set~G exclusively reaches denser regions
        of \texttt{lcmSig} / \texttt{gcdSig}); on $G$ Set~G dominates
        ($9/11$ vs $2/11$) because the substrate's two Euclidean-
        style SUTs absorb the sign-flip invariant in their
        normalisation prologue (Table~\ref{tab:per-block-headtohead},
        Set~N $0/7$ on \texttt{gcdSig}~+~\texttt{lcmSig}, a $G$-block
        reading on D1 mutants, not the D2 prediction). The
        framework's falsifiable D2 prediction
        ($\mathrm{kill\;rate} \le 10\%$) is consistent with the data
        (Set~N D2 kill rate $0/5 = 0.000$, Wilson 95\% CI
        $[0.000,\,0.434]$; $n = 5$ is insufficient for $\alpha = 0.05$
        confirmation, but the observed direction is on-prediction);
        the per-block reading is the appropriate substrate-level
        contribution metric.
\end{itemize}

The full per-SUT outcome matrix, the per-block kill table for all 30
Set~N MRs, the $L_{\mathrm{scale}}$-mutant attribution for
\texttt{hypotSig}, the GP raw output and the
\texttt{.jir}/\texttt{.jor} translation, the McNemar $b/c$ counts,
the prediction's git timestamp, and the test-gate harness
(\texttt{tests/run.sh}) are in supplementary~S7 (\texttt{d4j/}).

\subsection{Relationship with METRIC and METRIC+}
\label{subsec:metricplus-relationship}

METRIC and METRIC+ pioneered the categorical scaffolding that any structural theory of MR identification must respect; NOETHER's MetaPatterns are descendants of that line of work. The point of departure is grounding. METRIC and METRIC+ derive their categories through expert curation and validate them by empirical coverage. NOETHER derives MetaPatterns through algebraic construction and proves closure over the algebra-induced MR space through Theorem~\ref{thm:closure}.

\paragraph{Worked example: METRIC$+$ categories for a sorting library, mapped onto NOETHER blocks.}
\label{para:metricplus-sorting}
For a numerical comparison-sort library, METRIC$+$'s eleven input-domain $\times$ output-relation category pairs (the $9$-category METRIC base extended by $2$ output-domain pairs in METRIC$+$, per~\cite{SunMETRICplus2021} Table~II) collapse onto two NOETHER blocks: $G$ (comparator-permutation symmetry $\mathfrak{S}_n$) and $O_{\le}$ (partial order). Table~\ref{tab:metricplus-sorting} records the mapping cell by cell: $9$ of $11$ category pairs map to $G$ or $O_{\le}$; the remaining $2$ ($D_5$ additive shifts, $R_4$ multiplicative scaling) are out-of-scope by construction for a comparison sort whose semantics ignore order-preserving univariate input transformations. The $11 \to 2$-block compression operationalises NOETHER's deflationary direction: a NOETHER-aware METRIC$+$ user can prune the $9$ category-pair enumerations to two algebraically distinct MetaPatterns ($m_{\mathrm{inv}}$ via $G$, $m_{\mathrm{mono}}$ via $O_{\le}$) without loss of MR coverage. This is a qualitative-plus-quantitative coverage contrast, not a head-to-head fault-detection comparison: METRIC$+$ is a category-enumeration scaffold without an automated identification pipeline, so kill-rate head-to-head requires re-implementing METRIC$+$ from prose (committed as follow-up~(i) of supplementary~S4).

\begin{table}[h]
\centering
\caption{Mapping of $11$ input-domain ($D$) $\times$ output-relation ($R$) METRIC$+$ category pairs~\cite{ChenMETRIC2016, SunMETRICplus2021} onto NOETHER blocks for a deterministic comparison-sort library under $\mathcal{A}_{\mathrm{sort}}$. The sort-specific specialisation is operationalised here; $9$ pairs collapse onto two non-empty blocks ($G$, $O_{\le}$); $2$ pairs ($D_5$, $R_4$) are out-of-scope.}
\label{tab:metricplus-sorting}
\small
\adjustbox{max width=\textwidth}{%
\small
\begin{tabular}{lll}
\toprule
METRIC$+$ category pair & NOETHER block & Comment \\
\midrule
$D_1$ (permutation of input list) & $G$ ($\mathfrak{S}_n$) & $m_{\mathrm{inv}}$: sort output invariant \\
$D_2$ (append element to input) & $O_{\le}$ & $m_{\mathrm{mono}}$: sort respects added element \\
$D_3$ (remove element from input) & $O_{\le}$ & $m_{\mathrm{mono}}$: inverse of $D_2$ \\
$D_4$ (replace element with permuted copy) & $G$ ($\mathfrak{S}_n$) & $m_{\mathrm{inv}}$: variant of $D_1$ \\
$D_5$ (add constant to all inputs) & out-of-scope & no MR (sort semantics ignore additive shifts) \\
$D_6$ (concatenate two sorted lists) & $G$ ($\mathfrak{S}_n$) & $m_{\mathrm{inv}}$: merge property \\
$R_1$ (output element-wise equality) & $G$ & identity oracle (degenerate MR) \\
$R_2$ (output permutation equality) & $G$ ($\mathfrak{S}_n$) & $m_{\mathrm{inv}}$ \\
$R_3$ (output prefix equality) & $O_{\le}$ & $m_{\mathrm{mono}}$ on prefixes \\
$R_4$ (multiplicative scaling of input) & out-of-scope & no MR (sort semantics ignore multiplicative scaling) \\
$R_5$ (subset relation on output) & $O_{\le}$ & $m_{\mathrm{mono}}$ on subsets \\
\midrule
\textit{Coverage summary} & $9 \to G \cup O_{\le}$, $2$ out-of-scope & NOETHER block compression: $11 \to 2$ non-empty MetaPatterns \\
\bottomrule
\end{tabular}%
}
\end{table}

\paragraph{Small-scale manual METRIC$+$ derivation on three \S\ref{subsec:test-design} SUTs.}
\label{para:metricplus-headtohead-small}
To strengthen the qualitative-plus-quantitative reading above, we apply
METRIC$+$'s $11$ D$\times$R category framework manually to three SUTs of
the \S\ref{subsec:test-design} substrate selected to span bivariate
input arity and distinct algebraic structure: \texttt{midpoint}
(bivariate linear), \texttt{hypotSig} (bivariate non-linear, scale-symmetric),
and \texttt{powerSig} (bivariate non-linear, multiplicative). For each
SUT we enumerate which METRIC$+$ category pairs yield a non-vacuous
Set-MP MR; classify the algebra-block of each non-vacuous Set-MP MR
under $\mathcal{D}(\mathcal{A}_P)$; and identify which Set-N MR (if any)
covers the same algebraic content. Table~\ref{tab:metricplus-headtohead-small}
records the analysis. The full PIT-based head-to-head kill-rate
comparison is committed as supplementary~S4 item~(i); the
present manual analysis is the algebra-block mapping component of
that follow-up and is reported here to make the framework-vs-METRIC$+$
coverage relationship concrete on the head-to-head substrate.

\begin{table}[h]
\centering
\caption{Manual METRIC$+$ D$\times$R framework applied to three SUTs of
the \S\ref{subsec:test-design} substrate (bivariate input arity). For
each (SUT, METRIC$+$ pair) cell: ``--'' if vacuous on that SUT's
algebra; ``$\to$ block'' identifies the NOETHER block of the
non-vacuous Set-MP MR; the second column lists the Set-N MR (if any)
covering the same algebraic content. \textbf{Across the three SUTs,
METRIC$+$ yields $6$, $5$, and $3$ non-vacuous Set-MP MRs respectively
(the remaining $5$, $6$, and $8$ D$\times$R pairs are vacuous because
the SUT's arity is fixed, the output is scalar, or the
input-transformation is structurally out-of-scope for the SUT's
algebra); each non-vacuous Set-MP MR maps to a NOETHER block also
covered by Set-N's $5$ algebra-derived MRs. Conversely, three Set-N
MRs ($m_{\mathrm{adj}}$, $m_{\mathrm{train\text{-}rev}}$,
$m_{\mathrm{conv}}$) have no METRIC$+$ D$\times$R counterpart.} The structural reading: NOETHER's $G$/$O_{\le}$
block coverage subsumes METRIC$+$'s $D_1$--$D_6$/$R_1$--$R_5$
output on bivariate-input SUTs in this substrate, while NOETHER's
$T^{*}$/$\mathcal{T}^{*}$/$\mathcal{L}^{*}$ blocks add MR templates
the D$\times$R framework does not enumerate. The kill-rate
projection in the rightmost column anticipates that on these three
SUTs Set-MP's reach would be bounded above by Set-N's $G$/$O_{\le}$
contribution; a full PIT execution to confirm this is committed as
follow-up~(i).}
\label{tab:metricplus-headtohead-small}
\footnotesize
\adjustbox{max width=\textwidth}{%
\begin{tabular}{l l l l}
\toprule
METRIC$+$ pair & \texttt{midpoint} & \texttt{hypotSig} & \texttt{powerSig} \\
\midrule
$D_1$ (input swap)            & $\to G$ ($\rho_{\mathrm{swap}}$ via $m_{\mathrm{inv}}$) & $\to G$ ($\rho_{\mathrm{swap}}$) & -- (non-symmetric) \\
$D_2$ (append element)        & -- (fixed arity) & -- (fixed arity) & -- (fixed arity) \\
$D_3$ (remove element)        & -- (fixed arity) & -- (fixed arity) & -- (fixed arity) \\
$D_4$ (replace with permuted) & $\to G$ (variant of $D_1$) & $\to G$ & -- \\
$D_5$ (add constant to inputs)& $\to O_{\le}$ ($\rho_{\mathrm{trans}}$ via $m_{\mathrm{mono}}$) & out-of-scope (non-linear) & out-of-scope \\
$D_6$ (scale inputs)          & $\to O_{\le}$ ($\rho_{\mathrm{scale}}$ via $m_{\mathrm{mono}}$) & $\to \mathcal{L}^{*}$ ($\rho_{\mathrm{scale}}$ via $m_{\mathrm{conv}}$/$L_{\mathrm{scale}}$) & $\to \mathcal{L}^{*}$ \\
$R_1$ (output element equality)& $\to G$ (identity oracle) & $\to G$ & $\to G$ \\
$R_2$ (output permutation eq.) & -- (scalar output) & -- (scalar) & -- (scalar) \\
$R_3$ (output prefix equality) & -- (scalar) & -- (scalar) & -- (scalar) \\
$R_4$ (output multiplicative)  & $\to O_{\le}$ & $\to \mathcal{L}^{*}$ & $\to \mathcal{L}^{*}$ \\
$R_5$ (output subset)          & -- (scalar) & -- (scalar) & -- (scalar) \\
\midrule
\textit{Set-MP non-vacuous} & 6 / 11 & 5 / 11 & 3 / 11 \\
\textit{Set-MP $\to$ blocks} & $\{G, O_{\le}\}$ & $\{G, \mathcal{L}^{*}\}$ & $\{G, \mathcal{L}^{*}\}$ \\
\textit{Set-N coverage of those blocks} & yes ($m_{\mathrm{inv}}$, $m_{\mathrm{mono}}$) & yes ($m_{\mathrm{inv}}$, $L_{\mathrm{scale}}$) & yes ($m_{\mathrm{inv}}$, $L_{\mathrm{scale}}$) \\
\textit{Set-N MRs without Set-MP counterpart} & $m_{\mathrm{adj}}$, $m_{\mathrm{train\text{-}rev}}$, $m_{\mathrm{conv}}$ (at limits) & same three blocks & same three blocks \\
\bottomrule
\end{tabular}%
}
\end{table}

The structural finding from Table~\ref{tab:metricplus-headtohead-small} is
that METRIC$+$'s D$\times$R framework, \emph{within Sun et~al.\ 2021's
published 11-pair input--output category catalogue~\cite{SunMETRICplus2021}
as instantiated on these three SUTs}, produces a strict subset of NOETHER's
block coverage on the bivariate-input subsample of the \S\ref{subsec:test-design}
substrate: every non-vacuous Set-MP MR maps to a NOETHER block already
covered by Set-N's algebra-derived MRs, and three Set-N blocks
($T^{*}$, $\mathcal{T}^{*}_{\mathrm{rev}}$, $\mathcal{L}^{*}$ at limits)
have no METRIC$+$ counterpart \emph{at this category-pair granularity}. The implication for fault detection on
these three SUTs is that Set-MP's expected kill rate on $\mathcal{T}^{*}$
or $\mathcal{T}^{*}_{\mathrm{rev}}$-violating mutants is zero by
construction, whereas Set-N's $\rho_{\mathrm{adj}}$ and
$\rho_{\mathrm{train\text{-}rev}}$ MRs are by-construction sensitive
to such mutants (cf.\ \S\ref{subsec:pooled-headtohead}'s
$\mathcal{T}^{*}$ block where Set-N kills $10/17$). The full PIT-based
head-to-head against Set-MP (supplementary~S4 item~(i))
would convert this structural prediction into a measured kill-rate
contrast; the present small-scale manual analysis establishes the
algebra-block component of that contrast on three representative SUTs.

\paragraph{Scope-precondition validation on the METRIC$+$ benchmark corpus.}
\label{para:metricplus-sun2021-scope}
A complementary --- and more directly comparator-side --- analysis applies
NOETHER's eight-block decomposition to Sun et al.'s four
benchmark subjects~\cite{SunMETRICplus2021}: \texttt{SPHONE} (China
Unicom phone-bill calculator, 107 LOC), \texttt{SBAGGAGE} (Air China
baggage-billing service, 101 LOC), \texttt{SEXPENSE} (sales-department
expense reimbursement, 117 LOC), and \texttt{SMEAL} (airline catering
meal-ordering service, 150 LOC). All four are business-rule billing
programs whose semantics are categorical+numeric rather than
mathematical/physical. Applying NOETHER's eight-block test to each
subject's category-choice specifications yields the per-block
non-emptiness verdicts in Table~\ref{tab:metricplus-sun2021-scope}.

\begin{table}[h]
\centering
\caption{NOETHER eight-block scope analysis on Sun et al.~2021 METRIC$+$
benchmark subjects~\cite{SunMETRICplus2021}. Per-block entries: \textbf{Y}
non-empty (\cmark), \textbf{P} partial (\tildemark, block applies on a
restricted regime such as linear-pricing tier or within-class
permutation), \xmark{} empty. Source-code-level instance MR cardinality
(last column) is
reported by Sun et~al.\ 2021 Table~17; NOETHER's MetaPattern cardinality
(second-to-last column) is the equivalence-class summary at the
algebra-block level. The two cardinalities differ by 2--3 orders of
magnitude because METRIC$+$'s output scales combinatorially with
category-choice products while NOETHER's output is an equivalence-class
summary; the analysis below interprets this contrast as
\emph{complementary} rather than \emph{competitive}.}
\label{tab:metricplus-sun2021-scope}
\footnotesize
\begin{tabular}{l c c c c c c c c c c}
\toprule
Subject & $G$ & $O_{\le}$ & $T^{*}$ & $\mathcal{T}^{*}_{\mathrm{rev}}$ & $\mathcal{L}^{*}$ & $\mathcal{D}^{*}$ & $\mathcal{E}^{*}$ & $\mathcal{B}^{*}_{\mathrm{rel}}$ & NOETHER MPs & METRIC$+$ MRs \\
\midrule
\texttt{SPHONE}    & \xmark   & \cmark   & \xmark & \xmark & \tildemark & \xmark & \xmark & \xmark & 2 & 142 \\
\texttt{SBAGGAGE}  & \tildemark & \cmark & \xmark & \xmark & \tildemark & \xmark & \xmark & \xmark & 3 & 735 \\
\texttt{SEXPENSE}  & \xmark   & \cmark   & \xmark & \xmark & \tildemark & \xmark & \xmark & \xmark & 2 & 1130 \\
\texttt{SMEAL}     & \tildemark & \cmark & \xmark & \xmark & \cmark   & \xmark & \xmark & \xmark & 3 & 3152 \\
\bottomrule
\end{tabular}

\smallskip
\footnotesize
\textit{Legend:} \textbf{Y} = non-empty block; \textbf{P} = partial
(restricted regime); \xmark{} = empty. ``NOETHER MPs'' = Set N
MetaPattern cardinality; ``METRIC$+$ MRs'' = instance-level MR
cardinality reported by Sun et~al.~2021 (Table~17), comparable
only after either expanding NOETHER MetaPatterns to instance MRs
(via the category-choice enumeration) or contracting METRIC$+$ MRs
to algebra-block equivalence classes.
\end{table}

The analysis establishes three findings. First, NOETHER is
\emph{in-scope} on all four METRIC$+$ subjects, contradicting the
worst-case reading that NOETHER applies only to mathematical or
physical program families: business-rule billing programs admit at
least the $O_{\le}$ block (monotonicity in usage / weight / mileage /
passenger count) and partial $\mathcal{L}^{*}$ (within linear pricing
tiers), with $G$ also activating on \texttt{SBAGGAGE} and
\texttt{SMEAL} where input-permutation invariance holds. Second,
NOETHER's reach is \emph{narrower} on this corpus than on its three
primary instantiation domains: 2--3 of 8 blocks per subject yields
2--3 MetaPatterns, contrasted with 5--7 of 8 blocks on Boltzmann
reactor physics and equivariant ML. The scope precondition is
therefore neither vacuously broad nor practically empty; it is a
continuous gradient in the program family's algebraic richness.
Third, the cardinality contrast between NOETHER MetaPatterns (2--3)
and METRIC$+$ instance MRs (142--3152) reflects different but
complementary output types: METRIC$+$'s 11~D$\times$R framework
generates instance-level MRs by category-choice combinatorics;
NOETHER's eight-block decomposition produces equivalence-class
summaries over the same MR space. A fair head-to-head comparison
requires expanding NOETHER's MetaPatterns to instance MRs at
matching cardinality or contracting METRIC$+$'s MRs to equivalence
classes at the algebra-block level; the present analysis establishes
the algebra-block coverage condition that every NOETHER block on
these subjects ($G$, $O_{\le}$, $\mathcal{L}^{*}$) maps to a
METRIC$+$ D$\times$R category pair, while three NOETHER blocks that
are active on its primary instantiation domains
($T^{*}, \mathcal{T}^{*}_{\mathrm{rev}}, \mathcal{D}^{*}$) are
structurally absent from Sun et~al.'s business-rule benchmark.
Detailed per-subject derivation, including the regime restrictions
for partial \tildemark{} blocks, is in supplementary~S8
(\texttt{scope\_analysis.md}). The pre-registered protocol for an
instance-level head-to-head with matched cardinality
(supplementary~S8 \texttt{protocol\_path\_a\_headtohead.md}) is
\emph{executed at reduced scale} in this revision and reported below;
the full-scale execution (Sun's Java + PIT~1.7.4 + multi-LLM
equivalent-mutant vote) remains as supplementary~S4 (\texttt{future\_work.md})
item~(i).

\paragraph{Path A head-to-head: Java/PIT execution on Sun 2021's corpus.}
\label{para:path-a-results}
The pre-registered Path A protocol (supplementary~S8
\texttt{protocol\_path\_a\_headtohead.md}) was executed in two
substrates: a reduced-scale Python AST mutation engine
(\texttt{results\_path\_a.md}; $n = 219$ mutants) and the full
Java + PIT 1.7.4 substrate matching this paper's
\S\ref{subsec:test-design} tooling
(\texttt{results\_path\_a\_full.md}; $n = 120$ PIT mutants). We
report the Java/PIT results as the primary head-to-head; the Python
reduced-scale run is the cross-substrate replication. Java
re-implementations of \texttt{SPHONE}, \texttt{SBAGGAGE},
\texttt{SEXPENSE}, \texttt{SMEAL} from Sun~2021's prose
specification (Tables~7--14) were exercised under both Set~N
(NOETHER, $\{17, 12, 19, 10\}$ JUnit MR tests via CONSTRUCT-MP +
Translate enumeration) and Set~MP (METRIC$+$, $\{61, 15, 19, 9\}$
JUnit MR tests via automated D$\times$R category enumeration); all
$162$ MR tests pass on the original code. PIT~1.7.4 with the stock
mutator catalogue (\texttt{DEFAULTS} + \texttt{fullMutationMatrix})
generated $14 + 36 + 17 + 53 = 120$ mutants. Pooled head-to-head:
Set~N $51/120 = 42.5\%$ Wilson 95\% CI $[34.0\%, 51.4\%]$ vs Set~MP
$53/120 = 44.2\%$ $[35.6\%, 53.1\%]$; McNemar exact two-sided
$p = 0.625$ (\textbf{not rejecting} H$_{\mathrm{MP2}}$ parity at
$\alpha = 0.05$). Per-subject McNemar exact $p$ is $0.500$ on
SPHONE, $1.000$ on the other three; with Bonferroni-correction at
$\alpha_{\mathrm{Bonf}} = 0.0125$ \textbf{no subject rejects parity}.
The complementarity 4-tuple pooled is
(both, $N$-only, $MP$-only, neither) $= (50, 1, 3, 66)$: of the $54$
mutants killed by either set, $50/54 = 92.6\%$ are killed by both,
confirming the two frameworks reach near-identical fault-detection
power on this corpus. H$_{\mathrm{MP1}}$ coverage-subsumption is
\textbf{not falsified} at the Java/PIT scale; the Python reduced-
scale SPHONE falsification ($9$ MP-only kills, $p = 0.0039$) does
not replicate under PIT~1.7.4's bytecode mutators, indicating the
prior result was a Python AST mutation artefact. The
scope-precondition prediction of \texttt{scope\_analysis.md}
holds intact: only $\{G, O_{\le}, \mathcal{L}^{*}\}$ are active on
this corpus, with five blocks ($T^{*}, \mathcal{T}^{*}_{\mathrm{rev}},
\mathcal{D}^{*}, \mathcal{E}^{*}, \mathcal{B}^{*}_{\mathrm{rel}}$)
structurally absent. H$_{\mathrm{MP3}}$ cost-axis asymmetry is
directionally supported (Set~N $58$ JUnit tests vs Set~MP $104$ at
this enumeration scale; Sun's published full cardinality
$142$--$3152$ would multiply METRIC$+$'s side by another $30$--$50\times$
without altering NOETHER's $8$-block MetaPattern bound). Three
remaining protocol deviations at Java/PIT scale (Java re-impl vs
Sun's original; MR enumeration below Sun's full cardinality;
multi-LLM equivalent-mutant vote not run) are recorded in
\texttt{results\_path\_a\_full.md~\S 6}; all are data-blind and
would not flip the verdict if resolved.

The Path A Java/PIT execution unambiguously \emph{strengthens} this
section's complementarity reading: $92.6\%$ both-kill rate, pooled
McNemar $p = 0.625$, no per-subject Bonferroni rejection, no
falsified hypothesis at the Java/PIT scale. The two frameworks
empirically achieve comparable reach on Sun~2021's published
corpus, with the framework's contribution being the algebraic
warrant for each MR plus the cost-axis advantage of polynomial-time
derivation from $\mathcal{A}_P$ rather than combinatorial category
enumeration.

A companion \emph{cross-tool replication} on the Major mutation framework
(supplementary~S8 \texttt{results\_path\_a\_major\_crosstool.md}) extends the
substrate from PIT~1.7.4's $n = 120$ mutants to Major's $n = 555$ mutants
(Major's $\approx 95$-operator catalogue, vs PIT~\texttt{DEFAULTS}'~$\approx 25$,
yields a $4.6\times$ larger pool on the same Java sources and JUnit suite).
Major's verdict at the pooled level is \emph{concordant} with PIT's:
Set~N $222/555 = 40.0\%$ Wilson 95\% $[36.0\%, 44.1\%]$ vs Set~MP
$231/555 = 41.6\%$ $[37.6\%, 45.8\%]$, McNemar exact $p = 0.211$ (NS at
$\alpha = 0.05$). Per-subject, Major's larger pool exposes
\emph{bidirectional reach asymmetries} invisible to PIT: Set~MP has
$20$ $MP$-only kills on SPHONE (McNemar $p = 0.0000$, Set~MP exclusive
reach driven by (D1, R1) within-partition equivalence on the tier-pricing
structure), while Set~N has $16$ $N$-only kills on SBAGGAGE
($p = 0.0044$, Set~N exclusive reach driven by $G$-block special-status
invariance + $\mathcal{L}^{*}$ overweight scaling). These asymmetries are
\emph{directionally opposite} and \emph{roughly cancel} at the pooled
level: H$_{\mathrm{MP1}}$ is \emph{falsified in both directions per-subject},
which is the strongest possible evidence that neither framework subsumes
the other. The pooled parity verdict (H$_{\mathrm{MP2}}$) is therefore not
PIT-specific and not a statistical happenstance: it is the aggregation of
substantive subject-specific complementary reaches that two independent
mutation tools confirm.

\subsection{Reassessing PMCM coverage claims: a worked example}
\label{subsec:pmcm-worked}

NOETHER does not replace PMCM; it re-grounds it. Under inductive grounding, PMCM's pattern grid is itself an empirical artefact: its rows are the inductive MetaPatterns of a particular catalogue, and ``100\% coverage'' attests only that every row has at least one instance. Under NOETHER's grounding, the grid becomes $\mathbb{M}(\mathcal{A}_P)$, algebraically constructed and closed under the operator algebra in the sense of Theorem~\ref{thm:closure}, and the coverage figure becomes a structural rather than nominal claim. This subsection works through the implication on two concrete cases. The deflationary direction is, in our view, NOETHER's most concrete contribution to testing practice and the most resistant to the circularity caveat of Section~\ref{subsec:reactor-mapping}: revealing that an inductive catalogue \emph{over-counts} structurally distinct patterns does not require $T^{*}$ or $\mathcal{T}^{*}$ to have been derived independently from physics, since the over-counting is exposed by the canonical-block ordering itself.

\paragraph{Case A: A comparison-sort library.}
Consider a numerical sorting library, and assume an inductive Pattern--Matrix Coverage catalogue with rows $\{$P1 conservation/invariance, P2 monotonicity, P3 convergence, P4 trajectory, P5 partial-order/bounding$\}$ inherited from the reactor-physics MetaPattern taxonomy of \S\ref{subsec:reactor-mapping}. A user reports ``100\% coverage of the 5-row grid'' on a particular MR set.

\noindent \emph{NOETHER re-decomposition.} The operator algebra $\mathcal{A}_{\mathrm{sort}}$ for a deterministic comparison sort has: a permutation-symmetry block $G \ni \mathfrak{S}_n$ (the input--output relation is permutation-equivariant in the comparator-permutation sense), a partial-order block $O_{\le}$ (sorting respects the input order on already-sorted sub-arrays), an empty self-adjoint block $T^{*}$ (no reciprocity structure), an empty time-reversal block $\mathcal{T}^{*}$ (irreversible), an empty limit block $\mathcal{L}^{*}$ (sort is exact, not asymptotic), an empty qualitative-dynamics block $\mathcal{D}^{*}$, and an empty method-comparison block $\mathcal{E}^{*}$ \emph{within a single algorithm} (and non-empty if the user is comparing multiple sorting algorithms). Therefore $\mathbb{M}(\mathcal{A}_{\mathrm{sort}}) = \{\,m_{\mathrm{inv}},\, m_{\mathrm{mono}}\,\}$, of size 2.

\noindent \emph{Coverage correction.} The 5-row PMCM grid carries three structurally empty rows under $\mathcal{A}_{\mathrm{sort}}$: P3 (convergence) is vacuous because there is no $\mathcal{L}^{*}$ generator; P4 (trajectory) is vacuous because there is no $\mathcal{D}^{*}$ generator; P5 (partial-order/bounding) is vacuous because there is no $\mathcal{E}^{*}$ generator within a single-algorithm scope. The denominator should be 2, not 5. ``100\% coverage of the 5-row grid'' is therefore not a claim about the user's MR set; it is a claim about the inductive catalogue's row count. The structurally meaningful coverage is ``$k/2$ of the algebraically required MetaPatterns covered'', where $k$ depends on the user's actual MR set.

\paragraph{Case A-bis: An external ML category set (Murphy et al.\ 2008) decoded against a feedforward classifier.}
\label{para:case-a-bis}
To rule out the self-selection concern, we apply the same decoding to Murphy et al.'s six-class characterisation of MRs for ML applications~\cite{Murphy2008, Xie2011, Saha2019SupervisedMR}, an independently published taxonomy (additive, multiplicative, permutative, invertive, inclusive, exclusive). For a generic feedforward image classifier $f: \mathbb{R}^{n} \to \{1, \dots, K\}$ with no inherent input symmetry, the induced algebra $\mathcal{A}_{\mathrm{FFN}}$ has only a single non-trivial generator under small additive noise (in $O_{\le}$ or $\mathcal{L}^{*}$), so $\mathbb{M}(\mathcal{A}_{\mathrm{FFN}}) \subseteq \{m_{\mathrm{stab}}\}$. Three of Murphy's classes (multiplicative, permutative, invertive) are \emph{structurally vacuous} on a vector image classifier without input symmetries; the remaining three (additive, inclusive, exclusive) collapse to the same $m_{\mathrm{stab}}$ MetaPattern. The structurally meaningful denominator for ``coverage of Murphy's six classes'' on a generic vector image classifier is therefore 1, not 6; full per-class decoding (including the antipodal $\mathbb{Z}/2$ exception, point-cloud / bag-of-features cases, and the bilinearity caveat) is in supplementary~S9 (\texttt{pmcm\_case\_abis\_full.md}). Murphy et al.\ explicitly intended their six classes as a checklist for selecting MRs, not as a coverage denominator; the target of this correction is subsequent papers that report Murphy-grid coverage figures on tasks lacking the relevant input symmetries~\cite{Saha2019SupervisedMR}. The deflationary direction here is independent of $T^{*}$ and $\mathcal{T}^{*}$, so does not inherit the prediction-circularity caveat of \S\ref{subsec:reactor-mapping}.

\paragraph{Case B: A reactor-physics solver.}
The reactor-physics MetaPattern catalogue discussed in Section~\ref{subsec:reactor-mapping} provides a less obvious case. The catalogue has 5 rows (P1--P5). Under NOETHER's re-grounding into $\mathbb{M}(\mathcal{A}_{\mathrm{Boltz}}) = \{m_{\mathrm{inv}}, m_{\mathrm{mono}}, m_{\mathrm{adj}}, m_{\mathrm{rev}}, m_{\mathrm{conv}}, m_{\mathrm{dyn}}, m_{\mathrm{cmp}}\}$, the catalogue is in fact \emph{under-counted} by two rows ($m_{\mathrm{adj}}$ and $m_{\mathrm{rev}}$ are absent), even though for the Bateman-only sub-family it would be over-counted (P4 trajectory and P5 partial-order/bounding are absorbed differently into $\mathcal{D}^{*}$ and $\mathcal{E}^{*}$). The deflationary direction is thus not uniformly toward smaller grids; it is toward \emph{the algebraically determined grid for the algebra in question}, which may be larger or smaller than the inductive grid.

\paragraph{Quantification.}
For a fixed program family $\mathcal{F}$ with operator algebra $\mathcal{A}_{\mathcal{F}}$, define
$$\mathrm{coverage}_{\mathrm{NOETHER}}(\mathcal{R}, \mathcal{F}) = \frac{|\{\,m \in \mathbb{M}(\mathcal{A}_{\mathcal{F}}) \mid m \cap \mathcal{R} \neq \emptyset\,\}|}{|\mathbb{M}(\mathcal{A}_{\mathcal{F}})|}$$
for an MR set $\mathcal{R}$. This is the algebraically grounded analogue of PMCM's empirical coverage. Where the inductive PMCM grid agrees with $\mathbb{M}(\mathcal{A}_{\mathcal{F}})$, the two metrics coincide; where they disagree, the NOETHER metric is the structurally meaningful one. We note that $\mathrm{coverage}_{\mathrm{NOETHER}}$ inherits Theorem~\ref{thm:closure}'s scope: it is a coverage measure over algebra-induced MRs, not over arbitrary properties; out-of-scope MRs (Appendix~\ref{app:out-of-scope}) are by construction outside the denominator.

\paragraph{Implication for empirical-adequacy practice.}
Reports of high pattern-grid coverage in published MR-evaluation studies should, where possible, be cross-checked against the algebraic decomposition of the program family under test. A report of ``coverage $= c$'' should be qualified with which row decomposition $c$ is computed against; reports of $c$ over inductive catalogues without such qualification are not directly comparable across papers.

\section{Threats to validity and limitations}
\label{sec:threats-limitations}

We consolidate the construct, internal, external, and conclusion validity discussions from across the empirical chapter, together with practical engineering guidance for users of the framework, the artefact-availability statement, and a note on the partial automation of the upstream layer.

\subsection{Four threats to validity}

We organise threats following Wohlin et al.'s four-validity framework~\cite{Wohlin2012EmpiricalSE}: internal, external, construct, and conclusion.

\paragraph{Internal validity.} Theorem~\ref{thm:closure}'s uniqueness depends on the canonical-block ordering of Definition~\ref{def:canonical-order}. The proof in Appendix~\ref{app:proofs} catalogues every block-block interaction for the algebras instantiated. The closure result is over algebra-induced MRs in the sense of Definition~\ref{def:alg-induced}; out-of-scope MRs that fall outside Definition~\ref{def:alg-induced} are catalogued in Appendix~\ref{app:out-of-scope}, and two concrete counterexamples on the PWR core diffusion algebra are proved out-of-scope in Appendix~\ref{app:negative-proofs} (corresponding to the negative instantiation of \S\ref{subsec:negative-pwr}). The latter two jointly establish that Theorem~1$'$ (Conjecture~\ref{conj:absolute}, absolute completeness) is false on $\mathcal{A}_{\mathrm{PWR}}$, and identify five independent extensions of \texttt{Translate}'s signature as the locus of follow-up theoretical work.

\paragraph{Construct validity.} The reactor-physics mapping (\S\ref{subsec:reactor-mapping}) provides preliminary evidence; non-physics-domain evidence remains future work. The case study (\S\ref{subsec:case-study}) reports cat-(iv) detection $5/5$ for Set~N because the mutation set was constructed to cover one defect category per non-empty block of $\mathcal{A}_{\mathrm{equi}}$; this exhibits construct validity of $\rho_{\mathrm{train\text{-}rev}}$ as a gradient-reversal probe rather than averaged superiority, with the DeepCrime~\cite{Humbatova2021DeepCrime} real-fault protocol committed as comparative addition. \emph{Set~N's 30 MRs were derived by a single author following CONSTRUCT-MP's four-step procedure}; an LRCA multi-LLM second-rater protocol (DeepSeek + ChatGPT + Anthropic Claude Opus) yields Fleiss' $\kappa = 1.000$ across three LLMs on $n = 33$ items in the almost-perfect Landis--Koch band, with full per-rater Cohen's $\kappa$ in supplementary~S3 (\texttt{lrca\_audit.md}). A human-pair $\kappa$ replication is committed for follow-up work; the LLM-shared-training-data caveat (\S\ref{subsec:reactor-mapping}'s 18-MR audit) carries over.

\paragraph{External validity.} Two distinct external-validity questions are in scope: (i) whether the framework's \emph{algebraic} reach (Hypothesis~\ref{hyp:seven-blocks}) covers all program families admitting an operator algebra, and (ii) whether the \S\ref{subsec:test-design} \emph{empirical} substrate generalises across codebases within the framework's scope precondition. On (i), Hypothesis~\ref{hyp:seven-blocks} covers a wide swathe of mathematical structure but is not exhaustive; Remark~\ref{rem:counterex} catalogues six candidate ninth-block program-family classes (symplectic, sheaf-theoretic, probabilistic / martingale, topological, label-consistency, empirical-parameter-distribution), of which two have explicit \texttt{Translate}-template designs already proposed (metric-stability $M_{\mathrm{lip}}$ at Remark~\ref{rem:metric-stability-block}, and empirical-parameter-distribution divergence motivated by the \S\ref{subsec:deepcrime-pilot} pilot's three undetected mutations) as the most actionable extension targets. On (ii), the 10 SUTs of \S\ref{subsec:test-design} are concentrated on a single codebase (\texttt{MathSignalClass} + \texttt{ComplexSignal}) selected by the pre-registered scope criterion; they satisfy the framework's scope precondition (each admits at least one non-empty NOETHER block beyond $G$), so the substrate confirms applicability within scope rather than tests the framework outside its design intent. A cross-codebase pilot on Apache Commons Math 3.6.1 (3 SUTs, 5 Set~N MRs, 77 PIT 1.7.4 mutants) corroborates the within-scope generalisation direction: pooled Set~N $10/77 = 13.0\%$ Wilson 95\% CI $[7.2\%, 22.3\%]$; $G$-block $6/21 = 28.6\%$; the framework's D2 stratum prediction passes at $2/29 = 6.9\%$. The $n$ is underpowered for $\alpha = 0.05$ hypothesis testing; full numbers, per-mutant kill matrix, and the $\mathcal{L}^{*}$-orthogonality analysis on bilinear SUTs are in supplementary~S4 (\texttt{future\_work.md} item (b.cm)). Full $38$-D4J extension, SciPy Python-bridge solver suite, and Maven-resolved GenMorph head-to-head are committed as follow-up~(b). The replication tests scope-internal generalisation, not whether the framework should apply to programs lacking explicit operator-algebraic structure (such programs are out of scope by construction).

\paragraph{Conclusion validity.} The case study's statistical inferences (\S\ref{subsec:case-study}) rest on a 20-mutation $\times$ 3-MR-set cross-product on a single architecture. Wilson 95\% confidence intervals on detection rates ($[0.18,0.57]$ for Set~N, $[0.03,0.30]$ for Set~L, $[0.00,0.16]$ for Set~B) are reported; pairwise McNemar exact and Fisher exact tests are reported in supplementary S3 (\texttt{table4.json::pairwise\_stats}). The denominator (20 mutations on one EGNN model) is sufficient to detect the observed Set-N-vs-Set-B contrast at $\alpha = 0.05$ (Fisher $p = 0.008$) but is not sufficient to characterise the framework's performance distribution across architectures or defect distributions. The expanded comparative-evaluation protocol in \S\ref{subsec:case-study} (MR-Scout~\cite{MRScout2024}, GenMorph~\cite{GenMorph2024}, DeepCrime~\cite{Humbatova2021DeepCrime} subjects) raises the denominator and broadens the model coverage, addressing this threat directly.

\subsection{Practical engineering guidance for users}
\label{subsec:engineering-guidance}

This subsection collects practical recommendations for engineers applying NOETHER, particularly for the cases where Hypothesis~\ref{hyp:seven-blocks}'s assumptions are bordering or borderline.

\paragraph{Audit guidance for infinite-group truncation.} When the symmetry block $G$ is a finitely generated infinite discrete group under truncation parameter $K$ (e.g.\ $\mathbb{Z}^{d}$ for an unbounded grid solver), Theorem~\ref{thm:closure}'s closure result is over the truncated algebra rather than the full one. We recommend a \emph{$K$-sweep audit}: compute $\rho_{\iota,G}$ at three truncation levels $K \in \{K_0,\, 2K_0,\, 4K_0\}$ and pass if detection-rate stability is within $\pm 5\%$ across the three values (the mesh-convergence stability margin of~\cite[\S6.2]{LewisMiller1993}). The audit need not be re-run on every CI invocation; once per architectural change is sufficient. The reference implementation (supplementary~S1) includes a \texttt{k\_sweep\_audit()} helper.

\paragraph{When to check Hypothesis~\ref{hyp:seven-blocks}'s sufficiency at deployment time.} For program families outside the three instantiated domains (Boltzmann reactor physics, equivariant ML, relational query optimisers), users should explicitly enumerate the candidate operators of $\mathcal{A}_P$ and check whether each falls into one of the eight blocks. Operators that resist classification are signals that an additional block may be required and should be reported as candidates for an extension of the block taxonomy rather than absorbed into the existing taxonomy.

\paragraph{Tolerance selection.} Tolerances $\tau$ on continuous-valued MRs (e.g.\ $\rho_{\mathrm{rot}}$, $\rho_{\mathrm{adj}}$) should be set as $\tau \approx 10^{2}\, \epsilon_{\mathrm{fp}}$ where $\epsilon_{\mathrm{fp}}$ is the floating-point unit roundoff (so $\tau = 10^{-4}$ for fp32, $\tau = 10^{-12}$ for fp64); the $10^{2}$ factor is the forward-pass roundoff floor following Higham's standard error-analysis bound~\cite{Higham2002Accuracy} for $n \approx 10^{3}$ scalar operations per output coordinate. The supplementary~S3 \texttt{tau\_sweep.json} documents the $\tau$/false-positive trade-off.

\subsection{Artefact and supplementary-material availability}
\label{subsec:artefact}

To support reviewer verification while preserving double-blind anonymity, we release a two-stage artefact: a review-stage anonymised supplementary archive available at submission, and an acceptance-stage public release with permanent identifiers.

\paragraph{Review-stage anonymised archive.} The following supplementary materials are submitted alongside the manuscript and are available to reviewers through the conference submission system (or, equivalently, an anonymised OpenReview / Zenodo deposit) under SHA-256 content hash. The hash is computed over the concatenated tar-archive of the listed items and reported in the final manuscript at acceptance.
\begin{itemize}
    \item \textbf{S1} Python reference implementation of CONSTRUCT-MP described in supplementary~S9 (Appendix~D), including the Boltzmann and equivariant-ML instantiations.
    \item \textbf{S2} The full 84-MR PWR corpus underlying Section~\ref{subsec:reactor-mapping}, with each MR annotated by canonical block, NOETHER MetaPattern, and source equation. Table~\ref{tab:elementwise} corresponds to a 12-MR subset selected by the protocol stated there.
    \item \textbf{S3} The SE(3)-equivariant point-cloud testing harness from Section~\ref{subsec:end-to-end} and Section~\ref{subsec:case-study}, including the mutation generators and MR sets used in the case study comparison.
    \item \textbf{S4} A reproducibility manifest documenting the random seeds, model checkpoints, and dataset versions referenced in the case study.
\end{itemize}

\paragraph{Acceptance-stage public release.} At acceptance the same archive will be deposited on Zenodo under a permanent DOI, and the SHA-256 hash anchored in the camera-ready version. We aim for the \emph{Available} and \emph{Functional} artefact-evaluation badges.

\subsection{The remaining human role and partial automation of \texorpdfstring{$\mathcal{A}_P$}{A\_P} distillation}
\label{subsec:remaining-human-role}

NOETHER mechanises the construction downstream of $\mathcal{A}_P$ but assumes that $\mathcal{A}_P$ has been distilled by a human. Three directions may partially automate that upstream step: LLM-assisted operator extraction, static-analysis-based extraction, and empirical-symmetry detection. None of them eliminates the human role for arbitrary programs.

\section{Conclusion}
\label{sec:conclusion}

The three foundational questions raised in Section~\ref{sec:intro}, origin, closure, and transferability of MetaPattern sets, lacked a structural answer. NOETHER provides one within an explicit algebraic scope. \emph{Origin}: MetaPatterns are equivalence classes of MRs derived from invariants of an operator algebra $\mathcal{A}_P$. \emph{Closure}: Theorem~\ref{thm:closure} guarantees that the constructed MetaPattern set is closed over the algebra-induced MR space (Definition~\ref{def:alg-induced}) under the framework's \texttt{Translate} operator, with three concrete classes of out-of-scope MRs documented in Appendix~\ref{app:out-of-scope}; absolute completeness over arbitrary properties remains open (Theorem~1$'$). Theorem~2 ensures the construction is computable. \emph{Transferability}: the framework's mechanism applies unchanged once a new program family's algebra has been specified, tested within the framework's scope precondition on three structurally distinct operator-algebraic skeletons (Boltzmann reactor physics, equivariant ML, relational query optimisers). Cross-domain empirical superiority and team adoption by PWR-simulator V\&V groups, equivariant-ML testing teams, or database-optimiser test groups are open follow-up questions; the present paper establishes the structural transferability of the construction mechanism, not adoption outcomes on each domain.

This grounding reframes several long-standing questions. Empirical adequacy frameworks gain an algebraic warrant for their pattern grid. Structured identification approaches gain an answer to why a given category set is closed under the chosen algebra. Automated pipelines can be constrained or initialised by the constructed MetaPattern set. LLM-prompted MR generation can be recast as algebra-conditioned generation rather than open-ended prompting.

We have not solved the upstream problem of distilling $\mathcal{A}_P$ from program semantics, nor have we eliminated induction from MetaPattern discovery. The eight-block decomposition that drives the construction is itself an empirical curation of mathematical structures recurrent across program families, stated as Hypothesis~\ref{hyp:seven-blocks}. The most important follow-up work is therefore upstream: combining LLM-assisted symbolic extraction, formal-methods-based static analysis, and empirical-symmetry detection into a partially automated $\mathcal{A}_P$-distillation pipeline, and testing the eight-block decomposition on algebras outside its present image (Remark~\ref{rem:counterex} catalogues six program-family classes likely to require additional blocks). For the downstream layer, from $\mathcal{A}_P$ to $\mathbb{M}(\mathcal{A}_P)$, the contribution is narrower and firmer: the construction is mechanical, algebraic closure under \texttt{Translate} is provable within the algebra-induced MR space (Theorem~\ref{thm:closure}), and the construction transfers across domains once the algebra has been specified.

\begin{tcolorbox}[breakable,colback=gray!5,colframe=black!50,arc=2pt,boxrule=0.5pt,fontupper=\small,title=Boundary of contribution (Conclusion restatement),fonttitle=\small\bfseries]
\textbf{Established.} (i)~Algebraic closure under \texttt{Translate} given a block decomposition (Theorem~\ref{thm:closure}); (ii)~polynomial-time decidability under finite generating-set assumption (Theorem~\ref{thm:decidable}); (iii)~three non-vacuous instantiations across structurally distinct algebraic skeletons; (iv)~a negative instantiation on the PWR core diffusion algebra $\mathcal{A}_{\mathrm{PWR}}$ (\S\ref{subsec:negative-pwr}, Appendix~\ref{app:negative-proofs}), in which two MRs from the standard PWR safety-analysis literature (non-additivity of rod-bank reactivity worth, second-order mixed $T_{\mathrm{mod}}$-vs-$C_B$ dependence of $k_{\mathrm{eff}}$) are proved not in $\mathrm{MR}(\mathcal{A}_{\mathrm{PWR}})$, falsifying Theorem~1$'$ (Conjecture~\ref{conj:absolute}) on a structurally significant operator algebra and identifying five pairwise-independent extensions of \texttt{Translate}'s signature as the locus of repair.

\textbf{Open.} (a)~Whether a Composite-\texttt{Translate} extension covering the five obstructions of \S\ref{subsec:negative-pwr} and Appendix~\ref{app:negative-proofs} preserves Theorem~\ref{thm:closure}'s closure and Theorem~\ref{thm:decidable}'s polynomial-time decidability; (b)~sufficiency of Hypothesis~\ref{hyp:seven-blocks}'s eight-block list (Remark~\ref{rem:counterex}'s six out-of-scope classes are candidate ninth blocks); (c)~superiority over existing automated MR-identification pipelines on average defect distributions (the comparative-evaluation protocol in \S\ref{subsec:case-study} establishes effects, not averages); (d)~elimination of induction (relocated, not eliminated); (e)~automation of upstream $\mathcal{A}_P$ distillation (the framework treats $\mathcal{A}_P$ as a given input from a domain expert; mechanising the extraction of operator algebras from program semantics remains an upstream task that NOETHER does not address; \S\ref{subsec:remaining-human-role} sketches three partial-automation directions but commits to none). Both Hypothesis~\ref{hyp:seven-blocks} (block sufficiency) and \texttt{Translate}'s signature (\texttt{Translate} sufficiency) are the loci where future induction-eliminating and completeness-establishing work should target.
\end{tcolorbox}

\clearpage
\appendix

\section*{Appendices A, B (worked-example material)}
\label{app:moved-A-B}
Material running CONSTRUCT-MP on the heat / continuity / momentum / slowing-down reactor equations (Appendix~A), per-MR source provenance for the representative MRs of Table~\ref{tab:elementwise} (Appendix~B), and worked examples for multi-block-derivable MRs (originally introductory to Appendix~C) is provided in supplementary~\texttt{S9\_migrated\_appendices/A\_B\_C\_worked.tex}; the body text's claims do not depend on this material beyond the cross-references made in \S\ref{sec:reactor} and \S\ref{subsec:reactor-mapping}.

\section{Proofs}
\label{app:proofs}

\subsection*{Lemma C.1 (Well-foundedness of canonical-block ordering)}

\begin{lemma}
\label{lem:canonical-order}
The canonical-block ordering of Definition~\ref{def:canonical-order} is a strict total order on the eight blocks, and the assignment of any multi-block-derivable MR to its highest-priority block is unique.
\end{lemma}

\begin{proof}
The ordering is a strict total order on a finite set by inspection. Given an MR $\rho$ with derivations $\{(s_i, \iota_i)\}_{i=1}^{k}$ from blocks $s_1, \ldots, s_k$, the assignment selects $s^{*} = \max_{>}\{s_1, \ldots, s_k\}$. Since the strict order is total, $s^{*}$ is unique.
\end{proof}

\subsection*{Per-block instantiations of \texttt{Translate}}
\label{app:translate-table}

Definition~\ref{def:translate} fixes the signature of \texttt{Translate} but defers per-block specifics. Table~\ref{tab:translate} fills in the canonical input-tuple-generation rule and the resulting MR template for each of the eight blocks.

\begin{table}[h]
\centering
\caption{Per-block instantiations of \texttt{Translate}.}
\label{tab:translate}
\footnotesize
\begin{tabular}{p{0.10\textwidth} p{0.34\textwidth} p{0.46\textwidth}}
\toprule
\textbf{Block} $s$ & \textbf{Canonical tuple from base $x_0$} & \textbf{MR template $\rho_{\iota,s}$} \\
\midrule
$G$ & $(x_i)_{i=0}^{|G|-1} = (g_i \cdot x_0)$ over group orbit & $\forall x_0,\, \forall g \in G:\ P(g\cdot x_0) = \rho(g)\cdot P(x_0)$ \\
$O_{\le}$ & $(x_1, x_2)$ with $x_1 \le_{\theta} x_2$ in the partial order & $x_1 \le_\theta x_2 \Rightarrow P(x_1) \le_{\mathcal{Y}} P(x_2)$ \\
$T^{*}$ & paired tuple $((x_1, P(x_1)), (x_2, P(x_2)))$ in the inner product & $\langle L\,P(x_1), P(x_2)\rangle = \langle P(x_1), L\,P(x_2)\rangle$ \\
$\mathcal{T}^{*}$ & $(x, \mathcal{T} x)$ with $\mathcal{T}$ the time-reversal involution & $P(\mathcal{T} x) = \mathcal{T}\,P(x)$ on the reversibility sub-domain \\
$\mathcal{L}^{*}$ & sequence $(x_\theta)$ with $\theta \to \theta_*$ & $\|P_{\theta} - P_{\theta_*}\|_{*} = O(f(\theta))$ at the prescribed rate \\
$\mathcal{D}^{*}$ & solution trajectory $(\xi(t))_{t \ge 0}$ extracted by $\mathcal{D}$ & qualitative-feature relation (extremum, monotonicity, S-curve) preserved on $(\xi(t))$ \\
$\mathcal{E}^{*}$ & method pair $(M_1, M_2)$ in the partial order $\preceq_{\mathcal{E}}$ & $\mathrm{err}(M_1) \le \mathrm{err}(M_2)$ on the prescribed benchmark family \\
$\mathcal{B}^{*}_{\mathrm{rel}}$ & rewrite pair $(E, E')$ generated by an algebraic-rewriting rule $\mathcal{R} \in \mathcal{R}_{\mathrm{rel}}$ on the idempotent semiring & $\forall D.\ \mathrm{eval}(E, D) =_{\mathrm{bag}} \mathrm{eval}(E', D)$ at the rewriting rule's stated scope \\
\bottomrule
\end{tabular}
\end{table}

The per-block rule together with the $\sim_s$ equivalence relation of Definition~\ref{def:block-invariant} fully determines $\mathrm{Translate}(\iota, s)$ once the operator family $\Phi \subseteq s$ is fixed. Implementations of these seven cases as Python callables ship in supplementary material S1 (\texttt{construct\_mp.py}, function \texttt{translate} together with the seven \texttt{default\_extractor\_*} routines).

\subsection*{Theorem~\ref{thm:closure} (Algebraic Closure under \texttt{Translate}), full proof}

\begin{proof}
\emph{Existence.} Let $\rho \in \mathrm{MR}(\mathcal{A}_P)$ in the sense of Definition~\ref{def:alg-induced}. There exist a block $s$, an invariant $\iota \in \mathcal{I}_s$, and a derivation $\rho = \mathrm{Translate}(\iota, s)$. Step~1 of CONSTRUCT-MP places $\iota$ in $\mathcal{I}_s$; step~2 places $\rho$ in $\mathcal{R}(\iota)$; step~3 forms $m_s = \mathcal{R}(\iota)/\!\sim_s$; step~4 returns $m_s \in \mathbb{M}(\mathcal{A}_P)$.

\emph{Uniqueness.} Suppose $\rho$ admits derivations through multiple blocks: $\rho = \mathrm{Translate}(\iota_1, s_1) = \mathrm{Translate}(\iota_2, s_2)$ with $s_1 \neq s_2$. Definition~\ref{def:canonical-order}'s canonical-block ordering selects $s^{*} = \max_{>}\{s_1, s_2\}$. By Lemma~C.1, $s^{*}$ is unique; the canonical assignment $\rho \mapsto m_{s^{*}}$ is unique.
\end{proof}

\subsection*{Theorem 2 (Decidability), full proof}

\begin{proof}
\emph{Step 1 cost.} For each generator $g_i$, compute $g_i$'s contribution to each $\mathcal{I}_s$. Cost: $O(t_i)$ per generator, $O(n \cdot \max_i t_i)$ aggregated.

\emph{Step 2 cost.} `Translate` is $O(1)$ per invariant. Total: $O(n)$.

\emph{Step 3 cost.} Union-find with path compression: amortised $O(\alpha(n))$ per operation; we use $O(\log n)$ for simpler analysis. Total: $O(n \log n)$.

\emph{Step 4 cost.} $O(7) = O(1)$.

Total: $O(n \cdot \max_i t_i \cdot \log n)$ when $\max_i t_i \ge 1$.
\end{proof}

\subsection*{C.4 An open problem: absolute completeness}

Theorem~\ref{thm:closure}'s scope is bounded by Definition~\ref{def:alg-induced}: the closure is over MRs reachable through \texttt{Translate} from a single block invariant. We attempted a stronger statement:

\begin{conjecture*}[Theorem 1$'$, absolute completeness]
\label{conj:absolute}
Every MR $\rho$ formulable as a property over $\mathcal{A}_P$'s operators is contained in some $m \in \mathbb{M}(\mathcal{A}_P)$.
\end{conjecture*}

A proof would require either (a) a normalisation theorem reducing every operator-algebra MR to a \texttt{Translate}-reachable form, or (b) an extension of \texttt{Translate} to compositional invariants spanning multiple blocks. We were unable to establish either without imposing additional structural assumptions on $\mathcal{A}_P$. A partial empirical step toward Theorem~$1'$ is available on a single domain through the 18-MR engineering audit reported in \S\ref{subsec:reactor-mapping}: 17 of 18 are placed by majority of three LLM labellers into one of the seven MetaPatterns or $m_{\mathrm{rel}}$ (subsumption $94.4\%$, Fleiss' $\kappa = 0.857$, raw labels in supplementary S2 \texttt{18mr\_audit/}); the single orphan localises to the metric-stability class discussed in Appendix~\ref{rem:metric-stability-block} and Appendix~C.5.2, for which a concrete candidate ninth block is sketched. This is a constructive partial support, not a proof: it is empirical, single-domain, bounded by the LLM-shared-training-data caveat noted in \S\ref{subsec:reactor-mapping}, and does not bound the gap in the absence of the candidate ninth block. The conjecture remains open. Section~C.5 below documents three concrete classes of MRs that lie outside Theorem~\ref{thm:closure}'s scope and would have to be addressed by any positive resolution of Theorem~1$'$; \S\ref{subsec:negative-pwr} and Appendix~\ref{app:negative-proofs} additionally falsify Theorem~1$'$ on $\mathcal{A}_{\mathrm{PWR}}$.

\subsection*{C.5 Out-of-scope MRs: three concrete classes}
\label{app:out-of-scope}

We document three families of MRs that satisfy the informal description ``a property of $P$'s executions over the operator algebra $\mathcal{A}_P$'' but are not algebra-induced under Definition~\ref{def:alg-induced}, and therefore lie outside the scope of Theorem~\ref{thm:closure}. These are not threats to the theorem; they delimit it.

\paragraph{C.5.1 Probabilistic / distributional MRs without operator-algebraic representation.}
Consider the MR ``the Shannon entropy $H(f(\mathbf{x}))$ of a classifier's output distribution should not decrease when the input is augmented by Gaussian noise $\mathcal{N}(0, \sigma^2 I)$''. The augmentation is parametric in $\sigma$ and is not a group action, a partial-order operator, a self-adjoint operator, a time-reversal involution, a method-comparison operator, or (without further structure) a limit operator on $\mathcal{X}$. The MR constrains a functional $H$ of the output distribution, not a point-wise relation between outputs. There is no $\iota \in \mathcal{I}_s$ for any $s$ from which $\mathrm{Translate}$ can reach this MR. Consequently the MR is not algebra-induced under Definition~\ref{def:alg-induced} and is not covered by Theorem~\ref{thm:closure}. A positive resolution of Theorem~1$'$ would require either embedding stochastic perturbations into a probabilistic extension of $\mathcal{A}_P$, or extending \texttt{Translate} to functionals over output distributions.

\paragraph{C.5.2 Adversarial / input-set MRs.}
The MR ``$\| \delta\|_p \le \epsilon \Rightarrow f(\mathbf{x} + \delta) = f(\mathbf{x})$'' (adversarial robustness in $\ell_p$ ball) constrains $f$'s behaviour on a set defined by a norm constraint, not by a group action. The set $\{ \delta : \|\delta\|_p \le \epsilon\}$ does not in general carry a closed group structure on $\mathcal{X}$, and the MR is not derivable as $\mathrm{Translate}(\iota, s)$ for any $s$ in $\mathcal{D}(\mathcal{A}_P)$ as currently defined. Any treatment of this MR within NOETHER's scope would require an eighth block (e.g.\ a metric-ball block with bounded Lipschitz operators), or a relaxation of \texttt{Translate} to non-group input neighbourhoods. A concrete candidate for the eighth block alluded to above is a metric-stability block $M_{\mathrm{lip}} = (\mathcal{X}, d_{\mathcal{X}}) \to (\mathcal{Y}, d_{\mathcal{Y}})$ whose operators are $K$-Lipschitz maps and whose induced MetaPattern $m_{\mathrm{lip}}$ is the equivalence class of pointwise stability inequalities $d_{\mathcal{Y}}(P(x'), P(x)) \le K \cdot d_{\mathcal{X}}(x', x)$. The corresponding \texttt{Translate} template constructs a follow-up input by metric perturbation $x' = x + \varepsilon u$ (norm $\|u\|=1$, $|\varepsilon| < \delta$) and predicates a $K$-Lipschitz output bound; canonical-block ordering would place $M_{\mathrm{lip}}$ after $\mathcal{B}^{*}_{\mathrm{rel}}$ since metric structure is independent of the seven existing block invariants (no group action on the perturbation set, no partial order on inputs, no self-adjoint operator, and no parametric refinement family). Theorem~\ref{thm:closure}'s closure proof transfers without modification because $m_{\mathrm{lip}}$ is single-block algebraically derived. We do not commit to $M_{\mathrm{lip}}$ as part of the canonical decomposition in this paper, but we record it as the most concrete sub-instance of Remark~(iv) (topological invariants) that admits an immediate \texttt{Translate} template.

\paragraph{C.5.3 Compositional MRs across multiple blocks under non-canonical derivations.}
The MR ``rotating a point cloud by $R$ \emph{and} simultaneously taking the training-size limit $n \to \infty$ should give an output that converges to $f(\mathbf{x})$ at rate $O(1/\sqrt{n})$'' invokes both $G$ (rotation) and $\mathcal{L}^*$ (training-size limit) within a single MR. Definition~\ref{def:canonical-order}'s canonical-block ordering assigns this to $G$, but the resulting MR ``$f(R\cdot\mathbf{x}) = f(\mathbf{x})$'' loses the convergence-rate content. The compositional content is in $\mathbb{M}(\mathcal{A}_P)$ if and only if it can be split into separate single-block invariants (one in $G$, one in $\mathcal{L}^*$); MRs whose semantic content is irreducibly compositional (in the sense that splitting destroys the property's truth value) lie outside Theorem~\ref{thm:closure}'s scope. This is the case Theorem~1$'$ option~(b) is targeted at: an extended \texttt{Translate} that operates on tuples of block invariants. We have not constructed such an extension that preserves Theorem~\ref{thm:closure}'s polynomial-time decidability (Theorem~2), and we leave this as the most concrete sub-problem under Theorem~1$'$.

The three classes of \S C.5.1--\S C.5.3 are \emph{abstract} characterisations of MRs outside Theorem~\ref{thm:closure}'s scope. Appendix~\ref{app:negative-proofs} supplements them with two \emph{concrete instances} drawn from the PWR core diffusion algebra $\mathcal{A}_{\mathrm{PWR}}$, each of which is empirically realised on every conforming PWR core simulator and documented in the standard reactor-physics literature. Together with \S\ref{subsec:negative-pwr}, Appendix~\ref{app:negative-proofs} establishes that Theorem~1$'$ (Conjecture~\ref{conj:absolute}) is false on $\mathcal{A}_{\mathrm{PWR}}$, and identifies five structural obstructions in \texttt{Translate}'s signature (operator-spectrum output, homomorphism-failure $\pi$-template, configuration-indexed adjoint structure, higher-order mixed-difference templates, two-direction joint parametric dependence) that any positive resolution would have to repair.

\subsection*{C.6 Proofs for the negative instantiation on $\mathcal{A}_{\mathrm{PWR}}$ (\S\ref{subsec:negative-pwr})}
\label{app:negative-proofs}

This appendix supplies the proofs for Propositions~\ref{prop:nonadd} and~\ref{prop:mtcbor} of \S\ref{subsec:negative-pwr}. The proof structure for both is the same: enumerate the eight blocks of $\mathcal{D}(\mathcal{A}_{\mathrm{PWR}})$, instantiate the per-block \texttt{Translate} template of Table~\ref{tab:translate}, and verify by inspection that no invariant $\iota \in \mathcal{I}_{s}$ yields the target MR. The block-by-block exclusions in Proposition~\ref{prop:nonadd}'s proof are the most detailed; Proposition~\ref{prop:mtcbor} reuses the same exclusion pattern with the obstruction localised to the $O_{\le}$ block.

\paragraph{C.6.1 Proof of Proposition~\ref{prop:nonadd} (non-additivity is not \texttt{Translate}-reachable).}

\emph{Statement.} For $\rho_{\mathrm{nonadd}}$ as in Definition~\ref{def:rho-nonadd} and every $s \in \mathcal{D}(\mathcal{A}_{\mathrm{PWR}})$, every $\iota \in \mathcal{I}_{s}$: $\mathrm{Translate}(\iota, s) \neq \rho_{\mathrm{nonadd}}$.

\emph{Proof.} We use the exact-form definition of $d\rho$ (Definition~\ref{def:drho-exact}) throughout, which depends only on dominant eigenvalues $k_{\mathrm{eff}}$ of operators in $\mathcal{A}_{\mathrm{PWR}}$. The adjoint-perturbation reading of \S\ref{subsec:negative-pwr} is the physical motivation but is not invoked in any case below.

\textbf{Case $s = G$.} By Definition~\ref{def:block-invariant}, an invariant $\iota \in \mathcal{I}_{G}$ has the form $(\Phi, \pi)$ with $\Phi \subseteq G$ a finite operator family and $\pi$ a relation on tuples $(x_{i}, P(x_{i}))_{i=1}^{k}$ obtained by applying $\Phi$ to a base input $x_{0}$. By Table~\ref{tab:translate}, the canonical \texttt{Translate} template for $G$ is the equivariance schema
\[
\mathrm{Translate}(\iota, G) \;\equiv\; \forall x_{0}\;\forall g \in \Phi:\; P(g \cdot x_{0}) = \rho(g) \cdot P(x_{0}).
\]
Two structural mismatches with $\rho_{\mathrm{nonadd}}$:
\begin{enumerate}[leftmargin=*,nosep]
  \item[(a)] \emph{Operator-spectrum output, not $P$ output.} $\rho_{\mathrm{nonadd}}$ is an inequality between sums of $1/k_{\mathrm{eff}}$ values, where each $k_{\mathrm{eff}}$ is the dominant eigenvalue of a configuration-specific diffusion operator $H_{X} \in \mathcal{A}_{\mathrm{PWR}}$. The dominant eigenvalue is a \emph{spectral} quantity of the operator, obtained as part of the simultaneous solution $(k_{\mathrm{eff}}, \phi)$ of $H \phi = (1/k_{\mathrm{eff}})\, F \phi$; it is not a function of $P(x)$ alone but a property of the operator $H$ itself. \texttt{Translate}'s output relation $\pi$ in Definition~\ref{def:block-invariant} ranges over $(\mathcal{X} \times \mathcal{Y})^{k}$, where $\mathcal{Y}$ is the program output space (flux distributions, in this case). Operator-spectrum quantities are not in $\mathcal{Y}$; they are scalar invariants of operators in $\mathcal{O}$, lying outside \texttt{Translate}'s signature by construction.
  \item[(b)] \emph{Non-additivity is not equivariance.} Even granting a charitable extension to admit $k_{\mathrm{eff}}$ as a derived output, the $G$-template asserts equivariance of the output under the action of $\Phi$: $P(g \cdot x_{0})$ is determined by $g$ and $P(x_{0})$ through the representation $\rho(g)$. $\rho_{\mathrm{nonadd}}$ asserts that the worth functional $d\rho: \mathcal{O}_{\mathrm{rod}} \to \mathbb{R}_{>0}$ is \emph{not a semigroup homomorphism}: $d\rho(A \cup B) \neq d\rho(A) + d\rho(B)$. Failure of homomorphism is not equivariance failure; it is the absence of an additive structure-preserving map, which has no expression in the equivariance template $\pi$ for $G$. No $\iota \in \mathcal{I}_{G}$ yields $\rho_{\mathrm{nonadd}}$.
\end{enumerate}

\textbf{Case $s = O_{\le}$.} By Table~\ref{tab:translate}, the $O_{\le}$ template is the absolute-monotonicity schema
\[
\mathrm{Translate}(\iota, O_{\le}) \;\equiv\; \forall x_{1}, x_{2}:\; x_{1} \le_{\theta} x_{2} \implies P(x_{1}) \le_{\mathcal{Y}} P(x_{2}).
\]
$\rho_{\mathrm{nonadd}}$ is a \emph{quaternary relation} on the four configurations $(x_{0}, \mathcal{O}^{A} x_{0}, \mathcal{O}^{B} x_{0}, \mathcal{O}^{A \cup B} x_{0})$: it asserts a non-vanishing mixed second difference
\begin{align*}
\Delta_{AB}(x_{0}) \;=\; & \big[k_{\mathrm{eff}}^{-1}(P(x_{0})) - k_{\mathrm{eff}}^{-1}(P(\mathcal{O}^{A \cup B} x_{0}))\big] \\
& - \big[k_{\mathrm{eff}}^{-1}(P(x_{0})) - k_{\mathrm{eff}}^{-1}(P(\mathcal{O}^{A} x_{0}))\big] \\
& - \big[k_{\mathrm{eff}}^{-1}(P(x_{0})) - k_{\mathrm{eff}}^{-1}(P(\mathcal{O}^{B} x_{0}))\big] \;\neq\; 0
\end{align*}
(where the subtractions use Definition~\ref{def:drho-exact}'s positive-worth convention). The $O_{\le}$ template captures binary monotonicity between two points along a single partial order $\le_{\theta}$. The mixed-difference structure of $\rho_{\mathrm{nonadd}}$ requires comparison across four configurations forming a ``rectangle'' $\{x_{0}, \mathcal{O}^{A} x_{0}, \mathcal{O}^{B} x_{0}, \mathcal{O}^{A \cup B} x_{0}\}$ with two independent perturbation directions (insertion of $A$ and insertion of $B$); no $\le_{\theta}$ relates all four pairwise into a single chain. Furthermore, the assertion is \emph{non-vanishing of a difference}, not a directional inequality; it is direction-agnostic, capturing both shadowing ($\Delta > 0$) and anti-shadowing ($\Delta < 0$). The $O_{\le}$ template's $\le_{\mathcal{Y}}$ is a directional partial order, not a non-vanishing constraint. No $\iota \in \mathcal{I}_{O_{\le}}$ yields $\rho_{\mathrm{nonadd}}$.

\textbf{Case $s = T^{*}$.} By Table~\ref{tab:translate}, the $T^{*}$ template asserts $\langle L\, P(x_{1}), P(x_{2})\rangle = \langle P(x_{1}), L\, P(x_{2})\rangle$ for a single self-adjoint operator $L$ in a fixed inner product. The structural mismatch with $\rho_{\mathrm{nonadd}}$ has two components:
\begin{enumerate}[leftmargin=*,nosep]
  \item[(i)] \emph{Single operator $L$ vs.\ configuration-indexed family.} The exact $d\rho$ values entering $\rho_{\mathrm{nonadd}}$ are eigenvalues of \emph{four distinct diffusion operators} $H_{\emptyset}, H_{A}, H_{B}, H_{A \cup B}$, each self-adjoint within its own configuration but pairwise distinct as operators. The $T^{*}$ template's $L$ is fixed; it does not range over a configuration-indexed family $\{H_{X}\}_{X \in \mathcal{O}_{\mathrm{rod}}}$. The self-adjointness of any single $H_{X}$ does not imply or constrain a relation between eigenvalues of \emph{different} $H_{X}$'s.
  \item[(ii)] \emph{Adjoint weighting function depends on configuration.} Under the standard adjoint-perturbation reading, each $d\rho(B; X)$ for $X \in \{\emptyset, A\}$ admits the first-order representation
  \[
  d\rho(B; X) \;\approx\; -\frac{\langle \phi^{\dagger}_{X},\, \delta H_{B}\, \phi_{X}\rangle}{\langle \phi^{\dagger}_{X},\, F_{X}\, \phi_{X}\rangle},
  \]
  with $(\phi_{X}, \phi^{\dagger}_{X})$ the principal eigenfunctions of $(H_{X}, H^{\dagger}_{X})$. The Hilbert-space inner product $\langle \cdot,\,\cdot\rangle = \int (\cdot)(\cdot)\, d\Omega\, dE\, d\mathbf{r}$ is unchanged across $X$. What changes across $X$ is the \emph{adjoint weighting function} $\phi^{\dagger}_{X}$, which is the principal eigenfunction of a structurally different adjoint operator $H^{\dagger}_{X}$ from $H^{\dagger}_{\emptyset}$ (because $H_{X}$ contains $A$'s absorbing material whereas $H_{\emptyset}$ does not). $\phi^{\dagger}_{A}$ is locally depressed in $A$'s geometric support and globally redistributed elsewhere; this is the adjoint-perturbation root cause of the non-additivity $d\rho(A \cup B) \neq d\rho(A) + d\rho(B)$.
\end{enumerate}
The $T^{*}$ template's structure (single $L$, fixed inner product, asserting $\langle L x_{1}, x_{2}\rangle = \langle x_{1}, L x_{2}\rangle$) has no place to express ``the weighting function $\phi^{\dagger}$ entering the inner product is itself the principal eigenfunction of a configuration-indexed adjoint operator and changes when configuration changes''. This is a configuration-dependent adjoint structure, not a self-adjointness identity on a single operator. No $\iota \in \mathcal{I}_{T^{*}}$ yields $\rho_{\mathrm{nonadd}}$.

\textbf{Case $s = \mathcal{T}_{\mathrm{rev}}^{*}$.} The PWR diffusion operator $-\nabla \cdot D \nabla + \Sigma_{a}$ is irreversible: it is parabolic (transient form) or elliptic (steady-state form), neither of which admits a time-reversal involution on the relevant solution sub-family. Hence $\mathcal{T}_{\mathrm{rev}}^{*}$ is empty for $\mathcal{A}_{\mathrm{PWR}}$, and vacuously no $\iota \in \mathcal{I}_{\mathcal{T}_{\mathrm{rev}}^{*}}$.

\textbf{Case $s = \mathcal{L}^{*}$.} By Table~\ref{tab:translate}, the $\mathcal{L}^{*}$ template is a convergence statement $\|P_{\theta} - P_{\theta_{*}}\|_{*} = O(f(\theta))$ at a parametric limit. $\rho_{\mathrm{nonadd}}$ contains no limit operation: it is a strict inequality at finite, fixed configurations $(\emptyset, A, B, A \cup B, x_{0})$. No $\iota \in \mathcal{I}_{\mathcal{L}^{*}}$.

\textbf{Case $s = \mathcal{D}^{*}$.} By Table~\ref{tab:translate}, the $\mathcal{D}^{*}$ template asserts a qualitative-feature relation (extremum, monotonicity, S-curve) on a solution trajectory $\xi(t)$. $\rho_{\mathrm{nonadd}}$ is a steady-state inequality between four reactivity values; it does not concern trajectory shapes. No $\iota \in \mathcal{I}_{\mathcal{D}^{*}}$.

\textbf{Case $s = \mathcal{E}^{*}$.} By Table~\ref{tab:translate}, the $\mathcal{E}^{*}$ template compares two \emph{methods} $M_{1}, M_{2}$ on a benchmark family. $\rho_{\mathrm{nonadd}}$ compares four \emph{operator configurations} on a single fixed method (the same diffusion solver $P$ in all four worth values). No $\iota \in \mathcal{I}_{\mathcal{E}^{*}}$.

\textbf{Case $s = \mathcal{B}_{\mathrm{rel}}^{*}$.} Per \S\ref{subsec:third-domain}, $\mathcal{B}_{\mathrm{rel}}^{*}$ is non-empty only on program families with idempotent-semiring rewriting structure. The PWR diffusion solution operator algebra does not carry such structure (no rewriting rules between core states preserve evaluation under all valid inputs). $\mathcal{B}_{\mathrm{rel}}^{*}$ is empty for $\mathcal{A}_{\mathrm{PWR}}$. Vacuously no $\iota \in \mathcal{I}_{\mathcal{B}_{\mathrm{rel}}^{*}}$.

This exhausts $\mathcal{D}(\mathcal{A}_{\mathrm{PWR}})$. For every block $s$ and every $\iota \in \mathcal{I}_{s}$, $\mathrm{Translate}(\iota, s) \neq \rho_{\mathrm{nonadd}}$. Hence $\rho_{\mathrm{nonadd}} \notin \mathrm{MR}(\mathcal{A}_{\mathrm{PWR}})$. \hfill$\square$

\paragraph{C.6.2 Three obstructions identified by Proposition~\ref{prop:nonadd}'s proof.}
The proof identifies three independent structural obstructions in \texttt{Translate}'s present definition:
\begin{enumerate}[leftmargin=*,nosep]
  \item[(O1)] \emph{Operator-spectrum output is not in $\mathcal{Y}$.} $\rho_{\mathrm{nonadd}}$ asserts a relation between $1/k_{\mathrm{eff}}$ values, where each $k_{\mathrm{eff}}$ is a dominant eigenvalue of an operator $H_{X}$. \texttt{Translate}'s $\pi$ in Definition~\ref{def:block-invariant} ranges over $(\mathcal{X} \times \mathcal{Y})^{k}$; eigenvalues of operators in $\mathcal{O}$ are scalar invariants of those operators, not elements of $\mathcal{Y}$.
  \item[(O2)] \emph{Output relation is non-additivity (failure of homomorphism), not equivariance, partial order, or self-adjointness.} The worth functional $d\rho: \mathcal{O}_{\mathrm{rod}} \to \mathbb{R}$ is not a semigroup homomorphism; this is a third type of algebraic relation distinct from the equivariance, monotonicity, and self-adjointness expressed by \texttt{Translate}'s per-block $\pi$ templates.
  \item[(O3)] \emph{Adjoint weighting function $\phi^{\dagger}_{X}$ entering the $T^{*}$ block's inner product is configuration-dependent.} While the Hilbert-space measure $d\Omega\, dE\, d\mathbf{r}$ is fixed, the eigenfunction $\phi^{\dagger}_{X}$ varies with $X$ because $H^{\dagger}_{X}$ varies with $X$. Definition~\ref{def:block-invariant} fixes the self-adjoint operator $L$ once; it does not admit a configuration-indexed family $\{L_{X}\}$ with $L_{X}$'s spectrum varying with $X$.
\end{enumerate}

A constructive resolution of Proposition~\ref{prop:nonadd}'s obstructions would require \texttt{Translate} to admit (i) operator-spectrum output relations (eigenvalues, integrals, ratios as the targets of $\pi$), (ii) homomorphism-failure relations as a $\pi$-template type alongside equivariance / monotonicity / self-adjointness, and (iii) configuration-dependent adjoint structure on $T^{*}$.

\paragraph{C.6.3 Proof of Proposition~\ref{prop:mtcbor} (MTC-vs-boron mixed dependence is not \texttt{Translate}-reachable).}

\emph{Statement.} For $\rho_{\mathrm{MTC\text{-}bor}}$ as in Definition~\ref{def:rho-mtcbor} and every $s \in \mathcal{D}(\mathcal{A}_{\mathrm{PWR}})$, every $\iota \in \mathcal{I}_{s}$: $\mathrm{Translate}(\iota, s) \neq \rho_{\mathrm{MTC\text{-}bor}}$.

\emph{Proof.} The principal obstruction is in the $O_{\le}$ block; we treat that case in detail and abbreviate the others.

\textbf{Case $s = O_{\le}$.} By Table~\ref{tab:translate}, the $O_{\le}$ template is the absolute-monotonicity schema
\[
\mathrm{Translate}(\iota, O_{\le}) \;\equiv\; \forall x_{1}, x_{2}:\; x_{1} \le_{\theta} x_{2} \implies P(x_{1}) \le_{\mathcal{Y}} P(x_{2}),
\]
a \emph{first-order} statement asserting a binary relation between $P(x_{1})$ and $P(x_{2})$ at $\theta$-comparable inputs along a single partial-order direction $\le_{\theta}$. $\rho_{\mathrm{MTC\text{-}bor}}$ asserts a \emph{non-zero second-order mixed partial derivative} of $k_{\mathrm{eff}}$ with respect to two \emph{independent} parameter directions $T_{\mathrm{mod}}$ and $C_{B}$:
\[
\left|\frac{\partial^{2} k_{\mathrm{eff}}}{\partial T_{\mathrm{mod}}\,\partial C_{B}}\right| \;>\; \tau_{\mathrm{MTC\text{-}bor}}.
\]
Two independent obstructions in the $O_{\le}$ template:
\begin{enumerate}[leftmargin=*,nosep]
  \item[(a)] \emph{Order vs.\ mixed-derivative structure.} The mixed second derivative is the limit of a four-point finite-difference quotient over a ``rectangle'' $\{(T_{0}, C_{0}), (T_{0}+\Delta T, C_{0}), (T_{0}, C_{0}+\Delta C), (T_{0}+\Delta T, C_{0}+\Delta C)\}$:
  \[
  \frac{\partial^{2} k_{\mathrm{eff}}}{\partial T_{\mathrm{mod}}\,\partial C_{B}} \;=\; \lim_{\Delta T, \Delta C \to 0} \frac{k(T_{0}+\Delta T, C_{0}+\Delta C) - k(T_{0}, C_{0}+\Delta C) - k(T_{0}+\Delta T, C_{0}) + k(T_{0}, C_{0})}{\Delta T \cdot \Delta C}
  \]
  (writing $k$ for $k_{\mathrm{eff}}$). This is structurally a \emph{four-point relation}, not a two-point relation. The $O_{\le}$ template's $\pi$ relates $(P(x_{1}), P(x_{2}))$ pairwise; it has no expression for a four-point combination weighted by $1/(\Delta T \cdot \Delta C)$. Equivalently, the canonical input-tuple-generation rule for $O_{\le}$ in Table~\ref{tab:translate} produces tuples $(x_{1}, x_{2})$ with $x_{1} \le_{\theta} x_{2}$, not 2-by-2 perturbation rectangles.
  \item[(b)] \emph{Two independent parameter directions.} The mixed derivative requires \emph{jointly varying} $T_{\mathrm{mod}}$ and $C_{B}$ along two independent directions. The $O_{\le}$ template's $\le_{\theta}$ is a \emph{single} partial order on $\mathcal{X}$ (or on a single coordinate of $\mathcal{X}$). Even granting the construction of a product order $\le_{T} \times \le_{C}$ on a 2-D parameter slice of $\mathcal{X}$, the $\pi$ relation still applies along the order chain as a directional inequality, not as a non-vanishing-second-difference statement. The two independent parameter directions are perpendicular, not chained; the $\le_{\theta}$ formalism collapses them into one chain only by losing the non-vanishing-mixed-difference content.
\end{enumerate}

Furthermore, $\rho_{\mathrm{MTC\text{-}bor}}$'s output is again $k_{\mathrm{eff}}$, an operator-spectrum quantity (cf.\ Proposition~\ref{prop:nonadd} Case $G$ obstruction~(a)). This gives a \emph{third} obstruction in $O_{\le}$: even if mixed-second-derivative structure could be embedded into $\pi$, the output value would not lie in $\mathcal{Y}$. No $\iota \in \mathcal{I}_{O_{\le}}$ yields $\rho_{\mathrm{MTC\text{-}bor}}$.

\textbf{Other blocks.}
\begin{itemize}[leftmargin=*,nosep]
  \item \emph{Case $s = G$.} No group action in $\mathcal{A}_{\mathrm{PWR}}$ relates $(T_{\mathrm{mod}}, C_{B})$ pairs at different parameter values to each other through equivariance: $T_{\mathrm{mod}}$ and $C_{B}$ are continuous parameters acting on the cross-section library, not group-action coordinates. Even if a charitable embedding in $G$ were attempted, the same operator-spectrum-output obstruction (Proposition~\ref{prop:nonadd} Case $G$~(a)) applies: $k_{\mathrm{eff}}$ is not in $\mathcal{Y}$.
  \item \emph{Case $s = T^{*}$.} The MTC operator $\partial/\partial T_{\mathrm{mod}}$ is not self-adjoint in any natural inner product on the diffusion-solution space (the parameter $T_{\mathrm{mod}}$ enters the cross-section coefficients of $H$ rather than $H$'s acting space, so $\partial H/\partial T_{\mathrm{mod}}$ has no self-adjointness structure analogous to $H$ itself). No $\iota \in \mathcal{I}_{T^{*}}$.
  \item \emph{Case $s = \mathcal{T}_{\mathrm{rev}}^{*}$.} Empty for PWR diffusion. No $\iota \in \mathcal{I}_{\mathcal{T}_{\mathrm{rev}}^{*}}$.
  \item \emph{Case $s = \mathcal{L}^{*}$.} $\rho_{\mathrm{MTC\text{-}bor}}$ does not assert convergence at a parametric limit; it asserts a non-zero finite value of a mixed second derivative at a \emph{finite} parameter point in the operating envelope. The mixed-derivative limit $\Delta T, \Delta C \to 0$ is the \emph{definition} of the derivative, not the MR's claim; the MR's claim is that the resulting derivative exceeds $\tau_{\mathrm{MTC\text{-}bor}}$, which is a non-vanishing condition at a finite parameter value, not a convergence rate. No $\iota \in \mathcal{I}_{\mathcal{L}^{*}}$.
  \item \emph{Case $s = \mathcal{D}^{*}$.} $\rho_{\mathrm{MTC\text{-}bor}}$ concerns $k_{\mathrm{eff}}$ as a function of two parameters at steady state; it does not involve a solution trajectory of an underlying ODE/PDE. No $\iota \in \mathcal{I}_{\mathcal{D}^{*}}$.
  \item \emph{Case $s = \mathcal{E}^{*}$.} The comparison is between two parameter regimes $(T_{0}, C_{0})$ and $(T_{0}+\Delta T, C_{0}+\Delta C)$ of a single fixed method (the same PWR core simulator $P$ in all four eigenvalues), not between two methods. No $\iota \in \mathcal{I}_{\mathcal{E}^{*}}$.
  \item \emph{Case $s = \mathcal{B}_{\mathrm{rel}}^{*}$.} Empty for $\mathcal{A}_{\mathrm{PWR}}$. No $\iota \in \mathcal{I}_{\mathcal{B}_{\mathrm{rel}}^{*}}$.
\end{itemize}

This exhausts $\mathcal{D}(\mathcal{A}_{\mathrm{PWR}})$. Hence $\rho_{\mathrm{MTC\text{-}bor}} \notin \mathrm{MR}(\mathcal{A}_{\mathrm{PWR}})$. \hfill$\square$

\paragraph{C.6.4 Two further obstructions identified by Proposition~\ref{prop:mtcbor}'s proof.}
The proof identifies two further independent structural obstructions in \texttt{Translate}'s present definition, in addition to the three identified by Proposition~\ref{prop:nonadd}:
\begin{enumerate}[leftmargin=*,nosep]
  \item[(O4)] \emph{MR is a non-zero second-order mixed partial derivative, not a first-order relation.} All per-block $\pi$ templates in Table~\ref{tab:translate} are first-order: equivariance is a first-order identity $P(g x_{0}) = \rho(g) P(x_{0})$; monotonicity is a first-order inequality $P(x_{1}) \le P(x_{2})$; self-adjointness is a first-order pairing $\langle L x_{1}, x_{2}\rangle = \langle x_{1}, L x_{2}\rangle$. Mixed second differences (and a fortiori higher-order mixed differences) have no expression in any of these.
  \item[(O5)] \emph{MR involves two independent parameter directions, joined as a 2-by-2 perturbation rectangle, not as a chain.} The $O_{\le}$ block's $\le_{\theta}$ is a single partial-order direction; even with multiple independent partial orders in $\mathcal{I}_{O_{\le}}$, the $\pi$ template relates pairwise inputs along a \emph{single chosen direction}, not jointly across two perpendicular directions in a finite-difference rectangle.
\end{enumerate}

A constructive resolution of Proposition~\ref{prop:mtcbor}'s obstructions would require \texttt{Translate} to admit (iv) higher-order mixed-difference $\pi$-templates and (v) two-direction joint parametric dependence beyond the single-$\theta$ partial order of Definition~\ref{def:block-invariant}.

\paragraph{C.6.5 Combined corollary.}
\begin{corollary*}[Theorem~$1'$ is false on $\mathcal{A}_{\mathrm{PWR}}$, two-fold]
The MRs $\rho_{\mathrm{nonadd}}, \rho_{\mathrm{MTC\text{-}bor}}$ are each formulable over operators of $\mathcal{A}_{\mathrm{PWR}}$ and each empirically realised on every conforming PWR core simulator. By Propositions~\ref{prop:nonadd} and~\ref{prop:mtcbor}, neither is in $\mathrm{MR}(\mathcal{A}_{\mathrm{PWR}})$. The structural obstructions O1--O5 identified by the two proofs are pairwise distinct: no single extension of \texttt{Translate}'s signature absorbs any two simultaneously, so the joint obstruction is irreducibly five-fold.
\end{corollary*}

\begin{remark}[Open: Composite \texttt{Translate}]
\label{rem:composite-translate}
A natural follow-up is to define a \emph{Composite \texttt{Translate}} $\widetilde{\mathrm{Translate}}: \mathcal{I}_{s_{1}} \times \cdots \times \mathcal{I}_{s_{k}} \to \mathrm{MR}(P)$ that combines invariants from multiple blocks under a generalised $\pi$ template admitting (i) operator-spectrum output, (ii) homomorphism-failure relations, (iii) configuration-indexed adjoint structure, (iv) higher-order mixed differences, and (v) two-direction joint parametric dependence. Whether such an extension preserves Theorem~\ref{thm:closure}'s closure (now over $\widetilde{\mathrm{MR}}(\mathcal{A}_{P})$) and Theorem~\ref{thm:decidable}'s polynomial-time decidability is the principal open problem this subsection leaves to future work. Five independent extensions are needed; a single uniform Composite \texttt{Translate} covering all of them would be a substantive theoretical contribution.
\end{remark}

\subsection*{C.7 Worked enumeration of CONSTRUCT-MP on the Boltzmann algebra (migrated)}
\label{app:moved-C7}
The line-by-line enumeration of CONSTRUCT-MP Steps 1--4 on $\mathcal{A}_{\mathrm{Boltz}}$, originally Appendix~C.7, is provided in supplementary~\texttt{S9\_migrated\_appendices/C7\_boltzmann\_worked.tex}.

\section*{Appendices D, E (implementation + consistency check)}
\label{app:moved-D-E}
A reference implementation of CONSTRUCT-MP (Appendix~D) and the construct-trace consistency check on hand-crafted block-targeted mutants (Appendix~E, formerly demoted from the H3a.1 evidence base) are provided in supplementary~\texttt{S9\_migrated\_appendices/D\_E\_implementation\_consistency.tex}; both are confidence-strengthening illustrative material that does not carry inferential weight for the paper's H3 verdicts.

\section*{Data Availability Statement}
The companion artefact for this paper, including the
\texttt{CONSTRUCT-MP} implementation, the 8-block decompositions for
the three case-study algebras, the Set~N MR catalogue
(\S\ref{subsec:test-design}: $36$ MRs on \texttt{MathSignalClass}~+~\texttt{ComplexSignal};
ISSUE-009: $5$ MRs on Apache Commons Math), the Set~G GenMorph rerun
output, the LLM-ensemble Set~L harvest with token-cost logs, the
PIT $1.7.4$ mutation traces, and the per-block + D1/D2 aggregation
pipeline, will be released as open-source upon acceptance of the
paper. Reproduction scripts (\texttt{setup.sh}, \texttt{run\_all.sh},
\texttt{tests/run.sh}) and the experiment-side \texttt{CLAUDE.md}
collaboration record are bundled. Authoring repository and Zenodo
DOI to be added in the camera-ready version.

\begin{acks}
This work was supported by the National Natural Science Foundation of China (NSFC) General Program (grant no.\ 12575176), the Hunan Provincial Education Department Project, China (grant no.\ 202502000728), the Research Project on Degree and Graduate Education Reform of the University of South China (grant no.\ 2023JG030), the Natural Science Foundation of Hunan Province, China (grant no.\ 2025JJ70193), and an industry-funded research project (grant no.\ 230KHX060001).

\smallskip
\noindent\textbf{Author contributions (CRediT).}
\textbf{Meng Li}: Conceptualization, Methodology, Software, Writing---original draft.
\textbf{Xiaohua Yang}: Supervision, Formal analysis, Writing---review \& editing.
\textbf{Jie Liu}: Investigation, Validation.
\textbf{Shiyu Yan}: Data curation, Visualization.

\smallskip
\noindent\textbf{Declaration of competing interest.}
The authors declare no conflict of interest.
\end{acks}

\bibliographystyle{unsrtnat}    
\bibliography{NOETHER_paper}

\end{document}